\documentclass{lmcs}
\usepackage[utf8]{inputenc}
\pdfoutput=1

\usepackage{lastpage}
\lmcsdoi{18}{2}{21}
\lmcsheading{}{\pageref{LastPage}}{}{}%
{Apr.~02,~2020}{Jun.~17,~2022}{}

\usepackage{booktabs}


\usepackage{hyperref}

\usepackage{amsmath,amsfonts,amstext,amssymb,amsbsy}
\usepackage{latexsym,mathtools,mathrsfs,stmaryrd,amsthm}
\usepackage{url}
\usepackage[babel]{csquotes}
\usepackage{comment}
\usepackage{quiver}

\newcommand{\mypar}[1]{\vspace{0.5cm}\noindent{\bf #1} }

\usepackage{graphicx,xspace}
\usepackage{stmaryrd}
\usepackage{pifont}

\usepackage{mathpartir}

\usepackage{tikz}
\usetikzlibrary{matrix}
\usetikzlibrary{cd}

\newcommand{\dset}{\mathcal D}

\newcommand{\lset}{A}
\newcommand{\funct}{F}

\newcommand{\onev}{\star}

\newcommand{\beh}[1]{\llbracket #1 \rrbracket}

\newcommand{\succes}{\bullet}

\newcommand{\Ctx}{\mathit{Ctx}}

\newcommand{\term}{S}
\newcommand{\termone}{S_1}
\newcommand{\termtwo}{S_2}

\newcommand{\intervals}{\mathcal{I}}
\newcommand{\psuminterval}{+_p^{\intervals}}
\DeclareMathOperator{\minmax}{min-max}

\newcommand{\dirac}[1]{\delta_{#1}}

\newcommand{\pplus}[1]{+_{#1}}

\newcommand{\cplus}{\oplus}

\newcommand{\msum}{\sum}

\newcommand{\pss}{NPLTS\xspace}

\newcommand{\ttrel}[1]{\ensuremath{\overset{#1}{\longrightarrow}}}

\DeclareMathOperator{\corr}{corr}
\DeclareMathOperator{\pprob}{prob}

\newcommand{\quotientT}{q^{{T}}}
\newcommand{\quotientB}{q^{{B}}}
\newcommand{\barot}{\bar{o}_{{T}}}
\newcommand{\barob}{\bar{o}_{{B}}}
\newcommand{\barO}{\bar{o}}
\newcommand{\bartB}{\bar{t}_{B}}
\newcommand{\bartT}{\bar{t}_{T}}
\newcommand{\bart}{\bar{t}}
\newcommand{\bb}[1]{[\![ #1 ]\!]}
\newcommand{\bbangle}[1]{\langle\!\langle #1 \rangle\!\rangle}
\newcommand{\bbressup}[1]{\lfloor\!\lfloor #1 \rfloor\!\rfloor}
\newcommand{\bbresinf}[1]{\lceil\!\lceil #1 \rceil\!\rceil}
\DeclareMathOperator{\reachres}{reach}
\DeclareMathOperator{\reach}{reach}

\newcommand{\R}{\mathcal{R}}
\newcommand{\RR}{\mathrel{\mathcal{R}}}

\newcommand{\PS}{\mathcal{PS}} %The Theory of Pointed Convex Semilattices
\newcommand{\TPS}{T_\PS} %The Monad of Pointed Convex Semilattices
\newcommand{\SB}{\mathcal{SB}} %The Theory of Pointed Convex Semilattices
\newcommand{\TSB}{T_\SB} %The Monad of Pointed Convex Semilattices
\newcommand{\ST}{\mathcal{ST}} %The Theory of Pointed Convex Semilattices
\newcommand{\TST}{T_\ST} %The Monad of Pointed Convex Semilattices

\newcommand{\PCS}{\mathcal{PCS}} %The Theory of Pointed Convex Semilattices
\newcommand{\TPCS}{T_\PCS} %The Monad of Pointed Convex Semilattices
\newcommand{\CSB}{\mathcal{CSB}} %The Theory of  Convex Semilattices with bottom
\newcommand{\TCSB}{T_\CSB} %The Monad of  Convex Semilattices with bottom
\newcommand{\CST}{\mathcal{CST}} %The Theory of  Convex Semilattices with top
\newcommand{\TCST}{T_\CST} %The Monad of  Convex Semilattices with top
\newcommand{\eqmaymust}{\equiv}
\newcommand{\eqmay}{\equiv_{{B}}}
\newcommand{\eqmust}{\equiv_{{T}}}

\newcommand{\bbmay}[1]{[\![ #1 ]\!]_{{B}}}
\newcommand{\bbmust}[1]{[\![ #1 ]\!]_{{T}}}
\newcommand{\bbressupfp}[1]{\lfloor\!\lfloor #1 \rfloor\!\rfloor_{fp}}
\newcommand{\bbresinffp}[1]{\lceil\!\lceil #1 \rceil\!\rceil_{fp}}

\usepackage{trsym}
\usepackage{color}
\usepackage{verbatim}
\usepackage{url}
\usepackage{wrapfig}
\usepackage{rotating}
\makeatletter
\def\url@leostyle{%
  \@ifundefined{selectfont}{\def\UrlFont{\sf}}{\def\UrlFont{\small\ttfamily}}}
\makeatother

\usepackage[curve]{xypic}

\usepackage{graphics}
\usepackage{graphicx}
\usepackage[all]{xy}
\usepackage{eucal}

\DeclareMathOperator{\supp}{supp}
\DeclareMathOperator{\maxalg}{{\mathbb M}ax}
\DeclareMathOperator{\minalg}{{\mathbb M}in}

\newcommand{\id}{\text{\emph{id}}}
\newcommand{\strength}{\text{\emph{st}}}
\newcommand{\inl}{\text{\emph{in}}_l}
\newcommand{\inr}{\text{\emph{in}}_r}
\newcommand{\after}{\mathrel{\circ}}

\newcommand{\Pow}{\mathcal{P}}
\newcommand{\nePow}{\mathcal{P}_{ne}}
\newcommand{\Powne}{\mathcal{P}_{ne}}

\newcommand{\Idmap}{\textrm{Id}}

\newcommand{\Cat}[1]{\ensuremath{\mathbf{#1}}}
\newcommand{\Sets}{\Cat{Sets}}

\newcommand{\Kl}{\mathcal{K}{\kern-.2ex}\ell}

\newcommand{\set}[2]{\{#1\;|\;#2\}}
\newcommand{\setin}[3]{\{#1\in#2\;|\;#3\}}

\newcommand{\Alg}[1]{\mathbb{#1}}			% convex algebra
\DeclareMathOperator{\convex}{conv}			    % convex hull
\DeclareMathOperator{\f}{f}			    % convex hull
\DeclareMathOperator{\ide}{id}			    % convex hull		    % convex hull
\DeclareMathOperator{\dder}{der}			    % convex hull

\DeclareMathOperator{\CoAlg}{Coalg}
\DeclareMathOperator{\EM}{E{\kern-.2ex}M}

\DeclareMathOperator{\Id}{Id}			    % identity functor

\newcommand{\Dis}{\mathcal{D}}				% distributions
\newcommand{\T}{M}				% generic monad T
\newcommand{\tr}[1]{ \stackrel{#1}{\to}}           %MACRO FOR TRANSITIONS IN PA
\newcommand{\Tthree}{T_i} %The Monad of Pointed Convex Semilattices
\newcommand{\tthree}{\bar{t}^{\sharp}_{i}}
\newcommand{\othree}{\bar{o}^{\sharp}_{i}}
\newcommand{\eqthree}{\equiv_i}

\newcommand{\ev}{\mbox{\sl ev}}
\newcommand{\weg}[1]{}

\DeclareMathOperator{\termfun}{term}

\begin{document}

\title[Traces for Systems with Nondeterminism, Probability, and Termination]{The Theory of Traces for Systems with Nondeterminism, Probability, and Termination}

\titlecomment{{\lsuper*} This paper is an extended version of a LICS 2019 paper ``The Theory of Traces for Systems with Nondeterminism and Probability''. It contains all the proofs, additional explanations, material, and examples.}

\author[Filippo Bonchi]{Filippo Bonchi\rsuper{a}}	%required
\address{University of Pisa, Italy}	%required
\email{}  %optional
\thanks{Filippo Bonchi has been supported by the Ministero dell’Università e della Ricerca of Italy under Grant No. 201784YSZ5, PRIN2017 – ASPRA (Analysis of Program Analyses).}	%optional

\author[Ana Sokolova]{Ana Sokolova\rsuper{b}}	%optional
\address{University of Salzburg, Austria}	%optional
\email{}
%optional
%\thanks{thanks 2, optional.}	%optional

\author[Valeria Vignudelli]{Valeria Vignudelli\rsuper{c}}	%optional
\address{Univ Lyon, CNRS, ENS de Lyon, UCB Lyon 1, LIP, France}	% ??? right ??? %optional
\email{}
%\urladdr{name3@url3\quad\rm{(optionally, a web-page can be specified)}}
%optional
\thanks{Valeria Vignudelli has been partially supported by the French projects ANR-20-CE48-0005 QuaReMe, the European Research Council (ERC) under the European Union's Horizon 2020 programme (CoVeCe, grant agreement No 678157), the LABEX MILYON (ANR-10-LABX-0070) of Universit\'e de Lyon, within the program ``Investissements d'Avenir'' (ANR-11-IDEX-0007) operated by the French National Research Agency (ANR)}	%optional % chktex 8

%\subjclass{Theory of computation $\sim$ Probabilistic computation; Theory of computation $\sim$ Concurrency; Theory of computation $\sim$ Formalisms; Theory of computation $\sim$ Semantics and reasoning}

\begin{abstract}
This paper studies trace-based equivalences for systems combining nondeterministic and probabilistic choices. We show how trace semantics for such processes can be recovered by instantiating a coalgebraic construction known as the generalised powerset construction. We characterise and compare the resulting semantics to known definitions of trace equivalences appearing in the literature. Most of our results are based on the exciting interplay between monads and their presentations via algebraic theories.
\end{abstract}

\maketitle

\section{Introduction}\label{sec:intro}

Systems exhibiting both nondeterministic and probabilistic behaviour are abundantly used in verification~\cite{Baier2008,HermannsKK14,KNP02,DehnertJK017,vardi1985automatic,hansson1991time,segala1995probabilistic}, AI~\cite{CPP09,Kaelbling:1998vs,RussellNorvig:2009}, and studied from semantics perspective~\cite{HeunenKSY17,StatonYWHK16,HermannsPSWZ11}.
Probability is needed to quantitatively model uncertainty and belief, whereas nondeterminism enables modelling of incomplete information, unknown environment, implementation freedom, or concurrency. At the same time, the interplay of nondeterminism and probability has been posing some remarkable challenges~\cite{VaraccaW06,DBLP:journals/corr/KeimelP16,Mio14,Jacobs08,Varacca03,Mislove00,Goubault-Larrecq08a,TixKP09a}.
Figure~\ref{fig:examplesys} shows a nondeterministic probabilistic system (NPLTS) that we use as a running example.

Traces and trace semantics~\cite{vG01} for nondeterministic probabilistic systems have been studied for several decades within concurrency theory and AI using resolutions or schedulers---entities that resolve the nondeterminism. Most proposals of trace semantics in  the literature~\cite{Seg95:thesis,Seg95b,BDL14a,BDL14b} are based on such auxiliary notions of resolutions and differ on how these resolutions are defined and combined. We call such approaches \emph{local-view} approaches.

On the other hand, the theory of coalgebra~\cite{Rut00:tcs,Jacobs:book} provides uniform generic approaches to trace semantics of various kinds of systems and automata, via Kleisli traces~\cite{DBLP:journals/lmcs/HasuoJS07}, generalised determinisation~\cite{SBBR10}, providing e.g.\ an abstract treatment of language equivalence for automata, or logics~\cite{KlinR16}. We use the term \emph{global-view} approaches for the coalgebraic methods via generalised determinisation.

In this paper, we propose a theory of trace semantics for nondeterministic probabilistic systems that unifies the local and the global view.
We start by taking the global-view approach founded on algebras and coalgebras and inspired by automata theory, and study determinisation of NPLTS in this framework. Then we find a way to mimic the local-view approach and show that we can recover known trace semantics from the literature.
We introduce now the main pieces of our puzzle, and show how everything combines together in the theory of traces for NPLTS\@.

In order to illustrate our approach, it is convenient to recall nondeterministic automata (NDA) and Rabin probabilistic automata (PA)~\cite{Rab63}. Both NDA and PA can be described as maps $\langle o, t\rangle \colon X \to O \times (\T X)^A$ where $X$ is a set of states, $A$ is the set of labels, $o\colon X \to O$ is the output function assigning to each state in $X$ an observation, and $t\colon X \to (\T X)^A$ is the transition function that assigns to each state $x$ in $X$ and to each letter $a$ of the alphabet $A$ an element of $\T X$ that describes the choice of a next state. For NDA, this is a nondeterministic choice; for PA, the choice is governed by a probability distribution. An NDA state observes one of two possible values which qualify the state as accepting (output $1$) or not (output $0$). A state in a PA observes a real number in $[0,1]$. Below we depict an example NDA (on the left) and an example PA (on the right) with labels $A=\{a,b\}$ and with outputs denoted by $\downarrow$.

%\hskip-0.8cm
{\footnotesize 
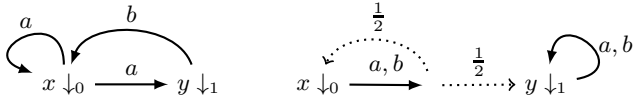
\begin{tikzpicture}[thick]

\matrix[matrix of nodes, row sep= 0.5cm, column sep=1cm,ampersand replacement=\&]
{
					\node (x) {$x \downarrow_0$};	\&\node (y) {$y \downarrow_1$};	\\
	};
\draw[-latex] (x) to node[above] {$a$} (y);
\draw[-latex] (x) to [out=90, in=160, looseness=5] node[above] {$a$} (x);
\draw[-latex] (y) to[bend right=70] node[above] {$b$} (x);

\begin{scope}[xshift=4cm]
\matrix[matrix of nodes, row sep= 0.5cm, column sep=1cm,ampersand replacement=\&]
{
					\node (x) {$x \downarrow_0$};	\&\node (d) {$^{}$};	\&\node (y) {$y \downarrow_1$};\\
	};
\draw[-latex] (x) to node[above] {$a,b$} (d);
\draw[-latex] (y) to [out=10, in=80, looseness=5] node[right] {$a,b$} (y);
\draw[dotted,->] (d) to[bend right=70] node[above] {$\frac 1 2$} (x);
\draw[dotted,->] (d) to node[above] {$\frac 1 2$} (y);

\end{scope}
\end{tikzpicture}%
}

The type of choice, modelled abstractly by a monad $\T$, is often linked to a concrete algebraic theory, the presentation of $\T$. Having such a presentation is a valuable tool, since it provides a finite syntax for describing finite branching. For nondeterministic choice this is the algebraic theory of semilattices (with bottom), for probabilistic choice it is the algebraic theory of convex algebras. Once we have such an algebraic presentation, we have a determinised automaton (as depicted below) and we inductively compute the output value after executing a trace by following the algebraic structure.

{ \footnotesize
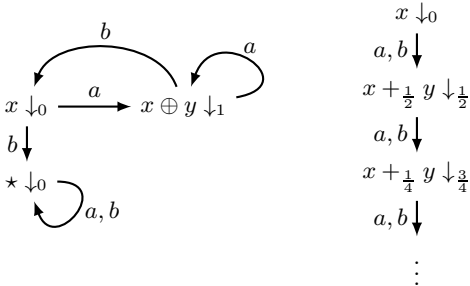
{\begin{tikzpicture}[thick]

\matrix[matrix of nodes, row sep= 0.5cm, column sep=1cm,ampersand replacement=\&]
{
					\node (x) {$x \downarrow_0$};	\&\node (y) {$x\cplus y \downarrow_1$};	\\
\node (z) {$\star \downarrow_0$};		\\
	};
\draw[-latex] (x) to node[above] {$a$} (y);
\draw[-latex] (x) to node[left] {$b$} (z);
\draw[-latex] (z) to [out=0, in=-70, looseness=5] node[right] {$a,b$} (z);
\draw[-latex] (y) to [out=10, in=70, looseness=4] node[above] {$a$} (y);
\draw[-latex] (y) to[bend right=70] node[above] {$b$} (x);

\begin{scope}[xshift=4cm]
\matrix[matrix of nodes, row sep= 0.5cm, column sep=1cm,ampersand replacement=\&]
{
					\node (x) {$x \downarrow_0$};	\\
\node (y) {$x\pplus{\frac{1}{2}}y \downarrow_{\frac{1}{2}}$};		\\
\node (z) {$x\pplus{\frac{1}{4}}y \downarrow_{\frac{3}{4}}$};		\\
\node (d) {$\vdots$};		\\[-0.1cm]
	};
\draw[-latex] (x) to node[left] {$a,b$} (y);
\draw[-latex] (y) to node[left] {$a,b$} (z);
\draw[-latex] (z) to node[left] {$a,b$} (d);
 \end{scope}

\end{tikzpicture}%
}
}

\noindent Here $x \cplus y$ denotes the nondeterministic choice of $x$ or $y$, and $x \pplus{p} y$ the probabilistic choice where $x$ is chosen with probability $p$ and $y$ with probability $1-p$.
For example, in the determinised PA we have, since $x \stackrel{a}{\to} x \pplus{\frac{1}{2}} y$ and $y \stackrel{a}{\to} y$:
\[x \pplus{\frac{1}{2}} y \stackrel{a}{\to} (x \pplus{\frac{1}{2}} y) \pplus{\frac{1}{2}} y = x \pplus{\frac{1}{4}} y\] and hence the output of $x \pplus{\frac{1}{4}} y$ is $o(x) \pplus{\frac{1}{4}} o(y) = \frac{3}{4}$ giving us the probability of $x$ executing the trace $aa$.
Our computation is enabled by having the right algebraic structure on the set of observations: a semilattice on $\{0,1\}$ and a convex algebra on $[0,1]$. The induced semantics is language equivalence and probabilistic language equivalence, respectively.

 This is the approach of trace semantics via a determinisation~\cite{SBBR10}, founded in the abstract understanding of automata as coalgebras and computational effects as monads.

We develop a theory of traces for NPLTS using such approach, by first identifying NPLTS as a special class of automata. For this purpose we take the monad for nondeterminism and probability~\cite{Jacobs08} with origins in~\cite{Mislove00,Goubault-Larrecq08a,TixKP09a,Varacca03,VaraccaW06}, namely, the monad $C$ of nonempty convex subsets of distributions, and provide all necessary and convenient infrastructure for generalised determinisation. The necessary part is having an algebra of observations, the convenient part is giving an algebraic presentation in terms of \emph{convex semilattices}. These are algebras that are at the same time a semilattice and a convex algebra, with a distributivity axiom distributing probability over nondeterminism.  Having the presentation we can write, for example,
 \[x \stackrel{a}{\to} x_1 \cplus (x_3 \pplus{\frac{1}{2}} x_2)\]
 for the NPLTS from Figure~\ref{fig:examplesys}.

 The presentation for $C$ is somewhat known, although not explicitly proven, in the community --- proving it and putting it to good use is part of our contribution which, in our opinion, clarifies and simplifies the trace theory of systems with nondeterminism and probability.

Remarkably, necessity and convenience go hand in hand on this journey. Having the presentation enables us to clearly identify what are the \emph{interesting} algebras necessary for describing trace and testing semantics (with tests being finite traces).
We identify three different algebraic theories: the \emph{theory of pointed convex semilattices}, the \emph{theory of convex semilattices with bottom}, and the \emph{theory of convex semilattices with top}.
These theories give rise to three interesting semantics arising in a canonical way by taking as algebras of observations those freely generated by a singleton set. We prove their concrete characterisations: the free convex semilattice with bottom is carried by $[0,1]$ with $\max$ as semilattice operation and standard convex algebra operations; the free convex semilattice with top is carried by $[0,1]$ with $\min$ as semilattice operation; and the pointed convex semilattice freely generated by $1$ is carried by the set of closed intervals in $[0,1]$ where the semilattice operation combines two intervals by taking their minimum and their maximum, and the convex operations are given by Minkowski sum.

  We call the resulting three semantics \emph{may trace}, \emph{must trace} and \emph{may-must trace semantics} since there is a close correspondence with \emph{probabilistic testing semantics}~\cite{YL92,JHY94,DGHMZ07a,DGHM09} when tests are taken to be just the finite traces in $A^*$. Indeed, the may trace semantics gives the greatest probability with which a state passes a given test; the must trace semantics gives the smallest probability with which a state passes a given test, and the may-must trace semantics gives the closed interval ranging from the smallest to the greatest.

\begin{figure}
\begin{center}
\begin{tikzpicture}[thick]

\matrix[matrix of nodes, row sep= 0.7cm, column sep=.1cm,ampersand replacement=\&]
{
	\&							\&\node (x) {$x$};		\\
	\&\node (x1) {$x_{1}$};	\&						\&\node (d2) {$\Delta_{2}$}; 	\& \\
	\&\node (d1) {$\Delta_{1}$};			\& 						\&						\&\node (x2) {$x_{2}$}; 	\& \\
  	\&		\&\node (z) {${x_{3}}$};		\& 	\&				 	 \\
	};
\draw[-latex] (x) to node[left] {$a$} (x1);
\draw[-latex] (x) to node[right] {$a$} (d2);
\draw[-latex] (x1) to node[left] {$b$} (d1);
\draw[-latex] (x2) to node[left] {$b$} (z);
\draw[-latex] (x2) to[bend right=70] node[right] {$c$}  (x);

\draw[dotted,->] (d1) to node[left] {$\frac 1 2$}  (z);
\draw[dotted,->] (d2) to node[left] {$\frac 1 2$}  (z);
\draw[dotted,->] (d2) to node[right] {$\frac 1 2$}  (x2);
\draw[dotted,->] (d1) to[bend left=70] node[left] {$\frac 1 2$}  (x);

\begin{scope}[yshift=-0.3cm,xshift=4.2cm]
\matrix[matrix of nodes, row sep= 0.7cm, column sep=.2cm,ampersand replacement=\&]
{
	\&							\&							\&\node (y) {$y$};			\&\\
	\&\& 	\node (y1) {$y_{1}$};							\&\node (t2) {$\Theta_{2}$}; 	\&	\&\node (t3) {$\Theta_{3}$}; \\
	\&\node (t1) {$\Theta_{1}$};			\& 							\&\node (y2) {$y_{2}$};		\&	\&				\&\node (y3) {$y_{3}$}; 	\\
  	\&		\& 		\node (z) {${y_{4}}$};					\& 	\&	\&				\&\\
  	\&							\& 	\\
	};
\draw[-latex] (y) to node[left] {$a$} (y1);
\draw[-latex] (y) to node[left] {$a$} (t2);
\draw[-latex] (y) to node[right] {$a$} (t3);
\draw[-latex] (y1) to node[left] {$b$} (t1);
\draw[-latex] (y2) to node[left] {$b$} (z);
\draw[-latex] (y3) to[bend right=70] node[right] {$c$} (y);

\draw[dotted,->] (t1) to node[left] {$\frac 1 2$}  (z);
\draw[dotted,->] (t2) to[bend right=20] node[left] {$\frac 1 2$}  (z);
\draw[dotted,->] (t2) to node[left] {$\frac 1 2$}  (y2);
\draw[dotted,->] (t3) to node[below] {$\frac 1 4$}  (y2);
\draw[dotted,->] (t3) to node[right] {$\frac 1 2$}  (y3);
\draw[dotted,->] (t1) to[bend left=70] node[left] {$\frac 1 2$}  (y);
\draw[dotted,->] (t3) to[bend left=30, looseness=1] node[right] {$\frac 1 4$}  (z);

\end{scope}

\end{tikzpicture}

\end{center}

\caption{NPLTS}\label{fig:examplesys}

\end{figure}

From the abstract theory, we additionally get that:
\begin{enumerate}
\item The induced equivalence can be proved coinductively by means of proof-techniques known as \emph{bisimulations up-to}~\cite{Milner89}. More precisely, it holds that up-to $\cplus$ and up-to $\pplus{p}$ are compatible~\cite{SP09b} techniques.
\item The equivalence is implied by the standard branching-time equivalences for NPLTS, namely bisimilarity and convex bisimilarity~\cite{segala1995probabilistic,Seg95:thesis}.
\item The equivalence is \emph{backward compatible} w.r.t.\ trace equivalence for LTS and for reactive probabilistic systems (RPLTS): When regarding an LTS and RPLTS as a nondeterministic probabilistic system, standard trace equivalence coincides with our may trace equivalence and with our three semantics, respectively.
\end{enumerate}

\noindent
Last but certainly not least, we show that the global view coincides with the local one, namely  that our three semantics can be elegantly characterised in terms of resolutions. The may-trace semantics assigns to each trace the greatest probability with which the trace can be performed, with respect to any resolution of the system; the must-trace semantics assigns the smallest one. It is important to remark here that our resolutions differ from those previously proposed in the literature in the fact that they are reactive rather than fully probabilistic. We observe that however this difference does not affect the greatest probability, and we can therefore show that the may-trace coincides with the randomized $\sqcup$-trace equivalence in~\cite{Cast18,BDL14a, BDL14b}.

Our theory is stated using the language of algebras, coalgebras and, more generally, category theory: this fact, on the one hand, guarantees some sort of canonicity of the constructions we introduce and, on the other, it allows for reusing general results and thus simplifying proofs.
Throughout the paper, we make an effort of keeping the presentation as accessible as possible also for those readers which are not familiar with the categorical language.

\mypar{Synopsis} We start by illustrating our global-view results concretely, i.e., without relying on categorical notions, in a process algebraic fashion, in Section~\ref{sec:global-view-alg}. We recall  monads and algebraic theories in Section~\ref{sec:monad}. We provide a presentation for the monad $C$ in Section~\ref{sec:C} (Theorem~\ref{th:pres-C}) and combine it with termination in Section~\ref{sec:termination}. We then recall, in Section~\ref{sec:GenDet}, the generalised determinisation and show an additional useful result (Theorem~\ref{thm:transfert}). All these pieces are put together in Section~\ref{sec:maymust}, where we introduce our three semantics and discuss their properties. The correspondence of the global view with the local one is illustrated in Section~\ref{sec:resolutions} (Theorem~\ref{thm:correspondence}).

\section{Trace semantics for NPLTS, concretely}\label{sec:global-view-alg}

\medskip

A nondeterministic probabilistic transition system with labels in $A$ is a pair $\langle X,t \rangle$ where $X$ is the set of states and $t\colon X \to (\Pow\Dis X)^A$ is a function assigning to each state $x\in X$ and label $a\in A$ a set of probability distributions over $X$.\footnote{To be completely precise, we should say a finite set of finitely supported distributions. This will be explained in full details in Section~\ref{sec:monad}.}

The idea is to construct for each NPLTS $\langle X,t \rangle$, a process calculus where the states in $X$ act as constants and the behaviour of such constants is determined by $t$. The terms of the calculus are defined inductively by the following grammar.
\begin{equation*}
s \, ::= \; x\in X \mid \star \mid  s\cplus s \mid s\pplus{p}s \quad \text{ for all }p\in [0,1]
\end{equation*}

\begin{table}
\[\infer{-}{x \stackrel{a}{\to}\termfun_X(t(x)(a))} \quad \infer{-}{\star \stackrel{a}{\to}\star} \quad \infer{s_1 \stackrel{a}{\to}s_1' \quad s_2\tr{a}s_2'}{s_1\cplus s_2 \stackrel{a}{\to}s_1' \cplus s_2'} \quad \infer{s_1 \stackrel{a}{\to}s_1' \quad s_2\tr{a}s_2'}{s_1\pplus{p} s_2 \stackrel{a}{\to}s_1' \pplus{p} s_2'}\]\\
{{\bf(a)} Transition function.}\label{fig:transitions}
\\
\[\infer{-}{x\downarrow_1} \quad \infer{-}{\star\downarrow_0} \quad \infer{s_1 \downarrow_{o_1} \quad s_2\downarrow_{o_2}}{s_1\cplus s_2 \downarrow_{\max(o_1,o_2)}} \quad  \infer{s_1 \downarrow_{o_1} \quad s_2\downarrow_{o_2}}{s_1\pplus{p} s_2 \downarrow_{p\cdot o_1 +(1-p)\cdot o_2}}\]\\
{{\bf(b)} Output function for may trace ($\eqmay$): $o_1,o_2 \in [0,1]$.}\label{fig:outputmay}
\\
\[\infer{-}{x\downarrow_1} \quad \infer{-}{\star\downarrow_0} \quad \infer{s_1 \downarrow_{o_1} \quad s_2\downarrow_{o_2}}{s_1\cplus s_2 \downarrow_{\min(o_1,o_2)}} \quad  \infer{s_1 \downarrow_{o_1} \quad s_2\downarrow_{o_2}}{s_1\pplus{p} s_2 \downarrow_{p\cdot o_1 +(1-p)\cdot o_2}}\]\\
{{\bf(c)} Output function for must trace ($\eqmust$): $o_1,o_2 \in [0,1]$.}\label{fig:outputmust}
\\
\[\infer{-}{x\downarrow_{[1,1]}} \quad \infer{-}{\star\downarrow_{[0,0]}} \quad \infer{s_1 \downarrow_{o_1} \quad s_2\downarrow{o_2}}{s_1\cplus s_2 \downarrow_{\minmax(o_1,o_2)}} \quad  \infer{s_1 \downarrow_{o_1} \quad s_2\downarrow_{o_2}}{s_1\pplus{p} s_2 \downarrow_{o_1\psuminterval o_2}}\]\\
{{\bf(d)} Output function for may-must trace ($\equiv$): $o_1,o_2 \in \intervals=\{[x,y] \mid x,y\in [0,1] \, \wedge \, x\leq y \}$. For all $[x_1,y_1], \, [x_2,y_2]\in \intervals$, $ \minmax ([x_1,y_1], \, [x_2,y_2])= [\min(x_1,x_2),\, \max(y_1,y_2)]$ and $[x_1,y_1] \psuminterval [x_2,y_2] = [p\cdot x_1+(1-p)\cdot x_2,\,\, p\cdot y_1+(1-p)\cdot y_2]$.}\\
\caption{Structural Operational Semantics (GSOS) for NPLTS\@.}\label{table:GSOS}
\end{table}

The inference rules in Table~\ref{table:GSOS}{.(a)} define a transition function over the terms of this grammar. The rules for $\star$, $\cplus$ and $\pplus{p}$ are self-explanatory. The rule for the constants $x\in X$ requires some explanation: the transitions of $x$ are defined using the transition function $t$ of the original NPLTS\@. Observe that for $a\in A$, $t(x)(a)$ is an element of $\Pow\Dis X$. The function  $\termfun_X$ assigns to each  of these elements its representation as term of the grammar. Indeed, sets of distributions, namely the elements of $\Pow\Dis X$, can be represented by the syntax above: distributions are represented by terms of the shape $(\dots((x_0 \pplus{p_0}x_1) \pplus{p_2} x_2) \dots )\pplus{p_n}x_n$; sets of probability distributions by terms of the shape $(\dots(d_0\cplus d_1) \cplus d_{n-1}) \dots \cplus d_n$ where each $d_i$ is a term representing a distribution. The empty set is represented by $\star$.

For instance, when starting with the NPLTS in the left of Figure~\ref{fig:examplesys}, the axiom for constants get instantiated to the following transitions:
\[\begin{array}{ccc}
x\tr{a}x_1 \cplus(x_2 \pplus{\frac{1}{2}}x_3) & x\tr{b}\star & x \tr{c} \star\\
x_1\tr{a}\star & x_1 \tr{b}x \pplus{\frac{1}{2}}x_3 & x_1\tr{c}\star\\
x_2\tr{a}\star & x_2 \tr{b} x_3 & x_2\tr{c}x\\
x_3\tr{a}\star & x_3 \tr{b}\star & x_3\tr{c}\star\\
\end{array}\]
Note that $x\tr{a}\star$ means $x\not \tr{a}$, i.e. $x$ cannot perform a transition with label $a$.
Now, with these axioms and the three other rules in Table~\ref{table:GSOS}{(a)} one can compute transitions for arbitrary terms of the syntax. For instance, starting from $x$, one obtains the transition system partially depicted below.

$
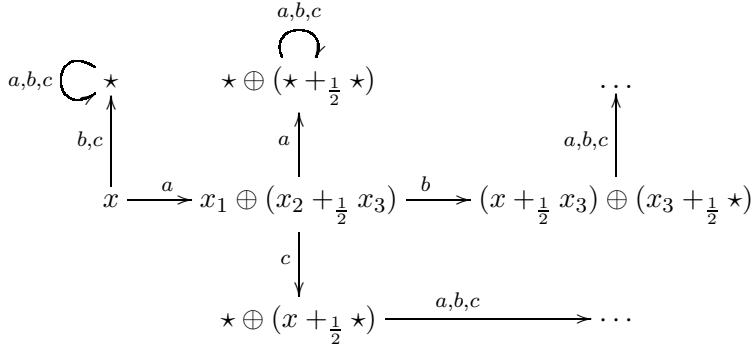
\xymatrix{
\star \ar@(ul,dl)_{a,b,c} & \star \cplus(\star \pplus{\frac{1}{2}}\star) \ar@(ul,ur)^{a,b,c} & \dots \\
x \ar[r]^(0.3)a \ar[u]^{b,c}& x_1 \cplus(x_2 \pplus{\frac{1}{2}}x_3) \ar[d]_c\ar[u]^a \ar[r]^(0.4)b & (x \pplus{\frac{1}{2}}x_3) \cplus (x_3 \pplus{\frac{1}{2}}\star ) \ar[u]^{a,b,c}  \\
& \star \cplus(x \pplus{\frac{1}{2}}\star) \ar[r]^{a,b,c} & \cdots
}$

It is important to note that this transition systems is deterministic: for each term $s$ and for each letter $a\in A$, there exists exactly one term $s'$ such that $s\tr{a}s'$. Therefore, for each word $w\in A^*$, there exists exactly one term reachable from $s$ with $w$. We denote this term by $t_w(s)$.

Furthermore, in order to define trace semantics, we decorate each term $s$ with  an output value: the trace semantics of $s$, written $\bb{s}$, is the function that maps each word $w\in A^*$ into the output value of $t_w(s)$.
In Table~\ref{table:GSOS}{(b),(c),(d)}, we define three different ways of decorating terms with outputs: we write $s\downarrow_o$ to mean that the output of $s$ is $o$. These three different ways lead to three different trace semantics: may ($\eqmay$), must ($\eqmust$) and may-must ($\equiv$).

For instance, the may trace and the must trace semantics of $x$ from Figure~\ref{fig:examplesys} are partially depicted as on the left and, respectively, on the right below.

\[\begin{array}{rcl}
\varepsilon & \mapsto & 1\\
a & \mapsto & 1\\
b & \mapsto & 0\\
c& \mapsto & 0\\
aa& \mapsto & 0\\
ab& \mapsto & 1 \\
ac& \mapsto & \frac{1}{2} \\
\dots & \mapsto & \dots
\end{array}
\qquad \qquad
\begin{array}{rcl}
\varepsilon & \mapsto & 1\\
a & \mapsto & 1\\
b & \mapsto & 0\\
c& \mapsto & 0\\
aa& \mapsto & 0\\
ab& \mapsto & \frac{1}{2} \\
ac& \mapsto & 0 \\
\dots & \mapsto & \dots
\end{array}\]

In the rest of the paper we study the three semantics that we have described so far. Our abstract treatment provides an explanation for each element of the above construction: all design choices of this semantics have a mathematical justification and nothing is arbitrary. Moreover, these semantics enjoy many desirable properties that can be easily proven by means of our abstract construction. We will come back to NPLTS and these trace semantics in Section~\ref{sec:maymust}, but first we need to set up the necessary algebraic (Sections~\ref{sec:monad},~\ref{sec:C} and~\ref{sec:termination}) and coalgebraic (Section~\ref{sec:GenDet}) playground.

\section{Monads and Algebraic Theories}\label{sec:monad}

In this paper, on the algebraic side, we deal with Eilenberg-Moore algebras of a monad on the category $\Sets$ of sets and functions, for which we also give presentations in terms of operations and equations, i.e., algebraic theories.

\subsection{Monads}\label{sec:monads}

A monad  on $\Sets$ is
a functor $\T\colon\Sets \rightarrow \Sets$ together with two
natural transformations: a unit $\eta\colon \Idmap
\Rightarrow \T$ and multiplication $\mu \colon \T^{2} \Rightarrow
\T$ that satisfy the laws
\[\xymatrix@C-.5pc@R-.5pc{
\T X\ar[rr]^-{\eta{\T}}\ar@{=}[drr] & & \T^{2}X\ar[d]^{\mu} & &
   \T X\ar[ll]_{\T\eta}\ar@{=}[dll]
&  & \T^{3}X\ar[rr]^-{\mu{\T}}\ar[d]_{\T\mu}
   & & \T^{2}X\ar[d]^{\mu} \\
& & \T X & &
& &
\T^{2} X\ar[rr]_{\mu} & & \T X
}\]

We next introduce several monads on $\Sets$, relevant to this paper. Each monad can be seen as giving side-effects.

\mypar{Nondeterminism} The finite powerset monad $\Pow$ maps a set $X$ to its finite powerset $\Pow X = \{U \mid U \subseteq X, \,\, U \textrm{ is finite}\}$ and a function $f\colon X\to Y$ to $\Pow f\colon \Pow X\to \Pow Y$, $\Pow f (U) = \{f(u) \mid u\in U\}$.
The unit $\eta$ of $\Pow$ is given by singleton, i.e., $\eta(x) = \{x\}$ and the multiplication $\mu$ is given by union, i.e., $\mu(S) = \bigcup_{U \in S} U$ for $S \in \Pow\Pow X$.
Of particular interest to us in this paper is the submonad $\Pow_{ne}$ of non-empty finite subsets, that acts on functions just like the (finite) powerset monad, and has the same unit and multiplication. We rarely mention the unrestricted (not necessarily finite) powerset monad, which we denote by  $\Pow_u$. We sometimes write $\overline{f}$ for $\Pow_u f$ in this paper.

\mypar{Probability} The finitely supported probability distribution monad $\Dis$ is defined, for a set $X$ and a function $f\colon X \to Y$,  as

\begin{align*}
&\Dis X
 =
\set{\varphi\colon X \to [0,1]}{\sum_{x \in X} \varphi(x) = 1,\, \supp(\varphi) \text{~is~finite}}\\
&\Dis f(\varphi)(y)
 =
\sum\limits_{x\in f^{-1}(y)} \varphi(x).
\end{align*}

The support set of a distribution $\varphi \in \Dis X$ is
$\supp(\varphi) = \setin{x}{X}{\varphi(x) \neq 0}$.
The unit of $\Dis$ is given by a Dirac
distribution $\eta(x) = \delta_x = ( x \mapsto 1)$ for $x \in X$ and
the multiplication by $\mu(\Phi)(x) = \sum_{\varphi \in \supp(\Phi)}
\Phi(\varphi)\cdot \varphi(x)$ for $\Phi \in \Dis\Dis X$.
\noindent We sometimes write $\sum_{i \in I} p_i x_i$ for a distribution $\varphi$ with $\supp(\varphi) = \{x_i \mid i\in I\}$ and $\varphi(x_i) = p_i$.

\mypar{Termination} The termination monad, also called lift and denoted by $\cdot + 1$ maps a set $X$ to the set $X + 1$, where $+$ denotes the coproduct in $\Sets$, which amounts to disjoint union, and $1 = \{\star\}$. For a coproduct $A+B$ we write $\inl\colon A \to A+B $ and $\inr\colon B \to A + B$ for the left and right coproduct injections, respectively. This monad maps a function $f\colon X \to Y$ to the function $f+1 \colon X + 1 \to Y+1$ defined, as expected, by $(f+1)(\inl(x)) = \inl(f(x))$ for $x \in X$ and $(f+1)(\inr(\star)) = \inr(\star)$. The unit of the termination monad is given by the left injection, $\eta\colon X \to X+1$ with $\eta(x) = \inl(x)$ and the multiplication by $\mu(\inl\after\inl(x)) = \inl(x)$ for $x \in X$, $\mu(\inl\after\inr(\star)) = \inr(\star)$, and $\mu(\inr(\star)) = \inr(\star)$. If clear from the context, we may omit explicit mentioning of the injections, and write for example $(f+1)(x) = x$ for $x \in X$ and $(f+1)(\star) = \star$.

\subsection{Monad Maps, Quotients and Submonads}\label{sec:quotients}

A monad map from a monad $\T$ to a monad $\hat\T$ is a natural transformation $\sigma\colon \T\Rightarrow\hat\T$ that makes the following diagrams commute, with $\eta, \mu$ and $\hat\eta, \hat\mu$ denoting the unit and multiplication of $\T$ and $\hat\T$, respectively, and
$\sigma\sigma$ denoting $\sigma\hat \T \after \T \sigma = \hat\T \sigma \after \sigma \T$.

\[\xymatrix@R-1pc{
&{X}\ar[dr]_{\hat\eta}\ar[r]^-{\eta} & \T X\ar[d]^{\sigma}
&
{\T\T X}\ar[d]_{\mu}\ar[r]^-{\sigma\sigma} & {\hat\T\hat\T X}\ar[d]^{\hat\mu} &
\\
&& {\hat\T X}
&
{\T X}\ar[r]_-{\sigma} & {\hat\T X}&
}\]

If $\sigma\colon \T \Rightarrow \hat\T$ is an epi monad map, then $\hat\T$ is a \emph{quotient} of $\T$. If it is a mono, then $\T$ is a \emph{submonad} of $\hat\T$. If it is an iso, the two monads are isomorphic.

\subsection{Distributive Laws}\label{subsec:distr-laws}

Let $(\T, \eta, \mu)$ be a monad and $F$ a functor. A natural transformation $\lambda\colon \T F \Rightarrow F\T$ is a \emph{functor distributive law} of the monad $\T$ over the functor $F$ if it commutes appropriately with the unit and the multiplication of $\T$, i.e.,

\[\xymatrix@R-1pc{
&{F X}\ar[dr]_{F\eta}\ar[r]^-{\eta F} & \T F X\ar[d]^{\lambda}
&
{\T\T F X}\ar[d]_{\T \lambda}\ar[rr]^-{\mu F} && {\T F X}\ar[d]^{\lambda} &
\\
&& {F \T X}
&
{\T F\T X}\ar[r]_-{\lambda \T} & {F\T \T X} \ar[r]_-{F \mu} & {F\T X}
}\]

Let $(\T, \eta, \mu)$ and $(\hat\T, \hat\eta, \hat\mu)$ be two monads. A \emph{monad distributive law} of $\T$ over $\hat\T$ is a natural transformation $\lambda\colon \T\hat\T \Rightarrow \hat\T\T$ that commutes appropriately with the units and the multiplications of the monads, i.e.,

\[\xymatrix@R-1pc{
&{\T X}\ar[dr]_{\hat\eta\T}\ar[r]^-{\T\hat\eta} & \T\hat\T X\ar[d]^{\lambda} & {\hat\T X}\ar[l]_{\eta\hat\T} \ar[dl]^{\hat\T\eta}
&
{\T\hat\T\hat\T X}\ar[d]_{\lambda\hat\T}\ar[rr]^-{\T\hat\mu} && {\T\hat\T X}\ar[d]^{\lambda} &
\\
&& {\hat\T \T X} &
&
{\hat\T\T\hat\T X}\ar[r]_-{\hat\T\lambda} & {\hat\T\hat\T \T X} \ar[r]_-{\hat\mu\T} & {\hat\T\T X}
}\]
and
\[\xymatrix@R-1pc{
{\T\T\hat\T X}\ar[d]_{\T\lambda}\ar[rr]^-{\mu\hat\T} && {\T\hat\T X}\ar[d]^{\lambda} &
\\
{\T\hat\T\T X}\ar[r]_-{\lambda\T} & {\hat\T\T \T X} \ar[r]_-{\hat\T\mu} & {\hat\T\T X}
}\]

Given a monad distributive law $\lambda\colon \T\hat\T \Rightarrow \hat\T\T$, we get a composite monad $\bar\T = \hat\T\T$ with unit $\bar\eta = \hat\eta\eta$ and multiplication $\bar\mu = \hat\mu\mu\after\hat\T\lambda\T$.

For any monad $\T$ on $\Sets$, there exists a distributive law $\iota\colon \T+1 \Rightarrow \T(\cdot+1)$ defined as
\begin{equation}\label{eq:iota}
\iota_X = \big( \xymatrix{{\T X+1}  \ar[rrr]^{[\T \inl, \eta_{X+1}\after \inr]}&&&{\T(X+1)}} \big).
\end{equation}
As a consequence, $\T(\cdot+1)$ is a monad.
Moreover, we get the following useful property.
\begin{lem}\label{lem:monad-map-termination}
	Whenever $\sigma\colon \T \Rightarrow \hat \T$ is a monad map, also $\sigma(\cdot + 1) \colon \T(\cdot + 1) \Rightarrow \hat \T(\cdot + 1)$ is a monad map. Injectivity of $\sigma$ implies injectivity of $\sigma(\cdot + 1)$. \qed%
\end{lem}

Lemma~\ref{lem:monad-map-termination} follows directly from Lemma~\ref{lem:map-distr-laws} and Lemma~\ref{lem:map-iotas} below. Before we state and prove these, let us consider another example of a distributive law in the other direction.

\begin{exa}\label{Example-D+1}
	We have that $\Dis + 1$ is a monad, since $\beta\colon \Dis(\cdot + 1) \Rightarrow \Dis + 1$ given by
	\begin{equation*}
	\beta(\varphi) = \left\{\begin{array}{ll}	 \varphi_X & * \not\in\supp(\varphi) \wedge \forall x \in X. \,\varphi_X(x) = \varphi(\inl(x))\\
* & * \in \supp(\varphi) \end{array}\right.
\end{equation*}
 is a monad distributive law.
This $\beta$ corresponds to termination as a ``black-hole'' of~\cite{SW2018}.
\end{exa}

\begin{lem}\label{lem:map-distr-laws}
	Given three monads $\T$, $\hat \T$, and $T$, two monad distributive laws $\lambda\colon T\T \Rightarrow \T T$ and $\hat\lambda\colon T\hat\T \Rightarrow \hat\T T$, ensuring that $\T T$ and $\hat\T T$ are monads, and a monad map $\sigma\colon \T \Rightarrow \hat\T$. If the following diagram commutes, in which case we say that $\sigma$ is a map of distributive laws,
		\[
\xymatrix{
T\T \ar[r]^{\lambda} \ar[d]_{T\sigma} & \T T \ar[d]^{\sigma T}\\
T \hat\T \ar[r]_{\hat\lambda} & \hat\T T
}
\]
	then $\sigma T \colon \T T \Rightarrow\hat\T T$ is a monad map. If $\sigma$ is injective, then $\sigma T$ is as well.
\end{lem}

\begin{proof}
	We denote by $\eta, \mu$ the unit and multiplication of $\T$, by $\hat\eta, \hat\mu$ those of $\hat\T$ and by $\eta^T, \mu^T$ those of $T$.
	Note that $\sigma T_X = \sigma_{TX}$ and hence, using that $\sigma$ is a monad map, we get immediately $\sigma T_X\after\eta_{TX} \after\eta^T_X = \hat\eta_{TX}\after\eta^T_X$.

	The following diagram commutes since $\sigma$ is a monad map.
\begin{equation}\label{eq:diag-floor1}\xymatrix@R-.5pc{
{\T\T TX}\ar[d]_{\mu T}\ar[r]^{\sigma {\T T}} &\hat\T\T TX \ar[r]^{\hat \T\sigma {T}} & {\hat\T\hat\T TX}\ar[d]^{\hat\mu {T}}
& &
\\
{\T TX}\ar[rr]^{\sigma {T}} && {\hat\T T X}
& &
}
\end{equation}
From the naturality of $\sigma$, the following diagram also commutes.
	\begin{equation}\label{eq:diag-floor2}\xymatrix@R-.5pc{
{\T\T TTX}\ar[d]_{\T\T\mu^T}\ar[r]^{\sigma {\T TT}} &{\hat\T\T TTX}\ar[d]^{\hat\T\T\mu^T} \ar[r]^{\hat\T\sigma {TT}} & \hat\T\hat\T TTX \ar[d]^{\hat\T\hat\T\mu^T}
\\
{\T\T TX}\ar[r]^{\sigma {\T T}} & {\hat\T\T TX} \ar[r]^{\hat\T\sigma {T}}
&  \hat\T\hat\T TX
}\end{equation}
Using once again the naturality of $\sigma$, for the left square, and the assumption that $\sigma$ is a map of distributive laws, for the square on the right, we get the commutativity of the following diagram.
\begin{equation}\label{eq:diag-floor3}
\xymatrix@R-.5pc{
{\T T \T TX}\ar[d]_{\T\lambda {T}}\ar[r]^{\sigma {T\T T}} &{\hat\T T \T TX}\ar[d]^{\hat\T\lambda {T}} \ar[r]^{\hat\T T\sigma {T}} & \hat\T
T\hat\T TX \ar[d]^{\hat\T\hat\lambda{T}}
\\
{\T\T TTX}\ar[r]^{\sigma {\T TT}} &{\hat\T\T TTX}\ar[r]^{\hat\T\sigma T {T}} & \hat\T\hat\T TTX
}
\end{equation}
Stacking diagram~(\ref{eq:diag-floor3}) on top of diagram~(\ref{eq:diag-floor2}) and further on top of diagram~(\ref{eq:diag-floor1}) gives the commutativity of
\[\xymatrix@R-.5pc{
{\T T\T TX}\ar[d]_{\mu^{\T T}}\ar[r]^-{\sigma T\sigma T} & {\hat\T T \hat\T T X}\ar[d]^{\mu^{\hat\T T}}
& &
\\
{\T TX}\ar[r]_-{\sigma T} & {\hat\T TX}
& &
}\]
with $\mu^{\T T}$ and $\mu^{\hat\T T}$ the multiplications of the monads $\T T$ and $\hat\T T$, respectively.
This completes the proof that $\sigma T$ is a monad map. Clearly, if all components of $\sigma$ are injective, then all components of $\sigma T$ (which are the components of $\sigma$ at $T X$) are injective as well.
\end{proof}

\begin{lem}\label{lem:map-iotas}
Let $\T$ and $\hat\T$ be two monads and $\sigma \colon \T \Rightarrow \hat\T$ be a monad map. Then the following commutes.
\[
\xymatrix{
\T X+1 \ar[r]^{\iota_X} \ar[d]_{\sigma_X+id_1} & \T(X+1) \ar[d]^{\sigma_{X+1}}\\
\hat\T X+1 \ar[r]_{\iota_X} & \hat\T(X+1)
}
\]
\end{lem}

\begin{proof} First observe that the following commutes: the left square commutes trivially; the right commutes since $\sigma$ is a monad map.
\[
\xymatrix{
1 \ar[r]^{\inr} \ar[d]_{id_1} & X+1 \ar[r]^{\eta_{X+1}} \ar[d]|{id_{X+1}}& \T(X+1) \ar[d]^{\sigma_{X+1}}\\
1 \ar[r]^{\inr} & X+1 \ar[r]_{\hat\eta_{X+1}} & \hat\T(X+1)
}
\]
The following diagram commutes by naturality of $\sigma$.
\[
\xymatrix{
\T X \ar[r]^{\T \inl} \ar[d]_{\sigma_X}& \T(X+1) \ar[d]^{\sigma_{X+1}}\\
\hat\T X \ar[r]_{\hat\T \inl} & \hat\T(X+1)
}
\]
The statement of the lemma follows from the commutativity of the two above diagrams and the universal property of the coproduct.
\end{proof}

\subsection{Algebraic  Theories}\label{sec:alg-th}

With a monad $\T$ one associates the  Eilenberg-Moore category
$\EM(\T)$ of $\T$-algebras. Objects of
$\EM(\T)$ are pairs $\Alg A = (A, a)$ of a set $A \in \Sets$ and a map
$a\colon \T A \rightarrow A$,
making the first two
diagrams below commute.

\[\xymatrix@R-1pc@C-1pc{
A\ar@{=}[dr]\ar[r]^-{\eta} & \T A\ar[d]^{a}
&
\T^{2}A\ar[d]_{\mu}\ar[r]^-{\T a} & \T A\ar[d]^{a}
&
\T A\ar[d]_{a}\ar[r]^-{\T h} & \T B\ar[d]^{b} \\
& A
&
\T A\ar[r]_-{a} & A
 &
A\ar[r]_-{h} & B
}\]

\noindent A homomorphism from an algebra $\Alg A = (A, a)$ to an algebra $\Alg B = (B, b)$ is a map $h\colon
A\rightarrow B$ between the underlying sets making the third diagram above commute. Hence the second diagram tells us that $a$ is a homomorphism from $(MA,\mu)$ to $(A,a)$. It is direct to check that $(MA,\mu)$ is an Eilenberg-Moore algebra. Moreover, $(MA,\mu)$ the free Eilenberg-Moore algebra generated by $A$.

In this paper we care for both categorical algebra, algebras of a monad, and their presentations in terms of algebraic theories and their models.  An algebraic theory is a pair $(\Sigma, E)$ of signature $\Sigma$ (a set of operation symbols) and a set of equations $E$ (a set of pairs of terms). A $(\Sigma,E)$-algebra, or a model of the algebraic theory $(\Sigma, E)$ is an algebra $\mathbb A = (A, \Sigma_A)$ with carrier set $A$ and a set of operations $\Sigma_A$, one for each operation symbol in $\Sigma$, that satisfies the equations in $E$. A homomorphism from a $(\Sigma,E)$-algebra $\mathbb A = (A, \Sigma_A)$ to a $(\Sigma,E)$-algebra $\mathbb{B} = (B, \Sigma_B)$ is a function $h\colon A \to B$ that commutes with the operations, i.e., $h \after f_A = f_B \after h^n$ for all $n$-ary $f \in \Sigma$, and $f_A, f_B$ its interpretations in $\mathbb{A}, \mathbb{B}$, respectively.
$(\Sigma,E)$-algebras together with their homomorphisms form a category.

Every algebraic theory gives rise to a monad, defined by the free $\dashv$ forgetful adjunction, see e.g.~\cite[V.6]{MacLane71}.
We now explicitly recall the construction.
In this construction, free algebras play an important role.
Recall that $\mathbb F_X = (FX, \Sigma_{FX})$ is the free $(\Sigma,E)$-algebra generated by $X \subseteq FX$ if for any $(\Sigma,E)$-algebra $\mathbb A = (A, \Sigma_A)$, any function $f\colon X \to A$ extends to a unique homomorphism $f^\#\colon \mathbb F_X\to \mathbb A$. Free algebras are unique up to isomorphism.
For any set $X$, the free $(\Sigma,E)$-algebra generated by $X$ is isomorphic to the algebra with carrier the set of $\Sigma$-terms with variable in $X$ modulo $E$-equations. The adjunction give rise to the monad $T_{\Sigma,E}$ of $\Sigma$-terms modulo $E$-equations. We next describe this monad.

Given a signature $\Sigma$, the monad $T_\Sigma$ of $\Sigma$-terms maps a set $X$ to the set of all $\Sigma$-terms with variables in $X$, and $f\colon X \to Y$ to the function that maps a term over $X$ to a term over $Y$ obtained by substitution according to $f$. The unit maps a variable in $X$ to itself, and the multiplication is term composition.
Given an algebraic theory $(\Sigma,E)$, the monad $T_{\Sigma,E}$ of $\Sigma$-terms modulo $E$-equations is defined analogously, using equivalence classes of terms rather than plain terms.
The monad $T_{\Sigma,E}$ is a quotient of $T_\Sigma$. Moreover, for two sets of equations $E_1 \subseteq E_2$, the monad $T_{\Sigma,E_2}$ is a quotient of $T_{\Sigma,E_1}$.

\begin{defi} A \emph{presentation} of a monad $M$ is an algebraic theory $(\Sigma, E)$  such that the monad $T_{\Sigma,E}$ is isomorphic to $M$.
\end{defi}

%Since $(\Sigma,E)$-algebras define the monad $T_{\Sigma,E}$, it follows  from~\cite[Theorem VI.8.1]{MacLane71} that $(\Sigma,E)$ is a presentation of the monad $T_{\Sigma,E}$, i.e., $(\Sigma,E)$-algebras are (isomorphic to) the Eilenberg-Moore algebras of $T_{\Sigma,E}$.
Clearly, $(\Sigma,E)$ is a presentation of the monad $T_{\Sigma,E}$.

\begin{rem}\label{rem:conc-iso}
A direct, syntactic, way to prove that $(\Sigma,E)$ is a presentation of a monad $M$ is to provide a monad isomorphism from $T_{\Sigma,E}$ to $M$. We give such a direct syntactic proof for our monad of interest in~\cite{CALCO21}. We now briefly discuss other ways to prove that $(\Sigma,E)$ presents $M$, which justify our proof of the presentation in this paper (Theorem~\ref{th:pres-C}).

Let $U\colon \Cat{A} \to \Sets$ be the forgetful functor from the category  $\Cat{A}$ of $(\Sigma,E)$-algebras to $\Sets$
 and let $F_{\Sigma,E}$ be the left adjoint to $U$ mapping a set X to the $(\Sigma,E)$-algebra of $\Sigma$-terms modulo $E$-equations. We have $T_{\Sigma,E} = U F_{\Sigma,E}$.
Now, suppose that $F$ is another left adjoint to $U$. Since adjoints are unique up to isomorphism, there exists an isomorphism $\iota:  F_{\Sigma,E} \Rightarrow F$. It is not difficult to derive from this that $U\iota: U F_{\Sigma,E} \Rightarrow UF$ is an isomorphism from the monad $T_{\Sigma,E} = U F_{\Sigma,E}$ to the monad $UF$.
Hence, in order to prove that $(\Sigma,E)$ presents a monad $M$, it suffices to show that there exists some left adjoint $F$ to the forgetful functor $U$ such that $M=UF$. This is what we use to prove the presentation in this paper, in the proof of Proposition~\ref{prop:presentations} below. The proposition itself identifies sufficient conditions for $M=UF$ to hold, for an adjoint $F$ to the forgetful functor $U$.

%It is not difficult to show that if $F$ and $\hat F$ are two left adjoints to the forgetful functor from $(\Sigma,E)$-algebras to $\Sets$, then $U\iota$ is an isomorphism from the monad $UF$ to the monad $U\hat F$ for $\iota: F \Rightarrow \hat F$ being the isomorphism due to the uniqueness of adjoints up to isomorphism. Therefore, since $T_{\Sigma,E} = UF_{\Sigma,E}$ where $F_{\Sigma,E}$ is the left adjoint to the forgetful functor mapping a set X to the $(\Sigma,E)$-algebra of $\Sigma$-terms modulo $E$-equations, in order to show that $(\Sigma,E)$ presents a monad $M$, it suffices to prove that $M=UF$ for another left adjoint of the forgetful functor from $(\Sigma,E)$-algebras to $\Sets$. This is what we use to prove the presentation in this paper, in the proof of Proposition~\ref{prop:presentations} below. The proposition itself identifies sufficient conditions that $M = UF$ for an adjoint $F$ of the forgetful functor $U$.

One can also work on the level of the categories of algebras.
The main observation is the following: a monad $M$ is isomorphic to a monad $T$ if and only if there is a \emph{concrete}\footnote{We thank Alo\"{\i}s Rosset for reminding us of the necessity of a concrete isomorphism. Previous versions of this paper omitted mentioning that the isomorphism needs to be concrete, although our presentation proofs were not affected by this.} isomorphism $I$ between their categories of Eilenberg-Moore algebras, as in the following diagram:

\[
\begin{tikzcd}
	{\EM(M)} && {\EM(T)} \\
	& {\Sets}
	\arrow["I", from=1-1, to=1-3]
	\arrow["{U^{M}}"', from=1-1, to=2-2]
	\arrow["{U^{T}}", from=1-3, to=2-2]
	&
\end{tikzcd}
\]

In such a case, using uniqueness of adjoints up to isomorphism, we also have an isomorphism $\iota\colon F^T \Rightarrow IF^M$ with $F^M$ and $F^T$ being the left adjoints of the forgetful functors $U^M$ and $U^T$, respectively, and the monad isomorphism is given by
\[M = U^MF^M = U^TIF^M  \stackrel{U^T\iota^{-1}}{\Longrightarrow} U^TF^T = T.\]
In the other direction, given a monad isomorphism $\sigma\colon M \Rightarrow T$, we have that
\[\left(\!\begin{gathered}\xymatrix{TX \ar[dd]^{a^T}\\ \\ X}\end{gathered}\right) \mapsto \left(\!\begin{gathered}\xymatrix{MX \ar[d]^{\sigma_X}\\ TX \ar[d]^{a^T}\\X}\end{gathered}\right)\]
is a concrete isomorphism of the algebras.
Hence, another way to prove that $(\Sigma, E)$ is a presentation of $M$ is to provide a concrete isomorphism between $\EM(T_{\Sigma,E})$ and $\EM(M)$.
%This way of providing a presentation was used in~\cite{MV20,MSV21}.

Moreover, as a consequence of Beck's theorem (see~\cite[Theorem VI.8.1]{MacLane71}), there is a concrete isomorphism between $(\Sigma,E)$-algebras and $\EM(T_{\Sigma,E})$ given by the comparison functor. Hence we get that $(\Sigma,E)$ is a presentation of $M$ if and only if there is a concrete isomorphism between the category of $(\Sigma,E)$-algebras and $\EM(M)$. This leads us to a third way to prove that $(\Sigma, E)$ is a presentation of $M$: providing a concrete isomorphism between $(\Sigma,E)$-algebras and $\EM(M)$.
This way of providing a presentation was used in~\cite{MV20,MSV21}.

Such concrete isomorphism can be also obtained using again~\cite[Theorem VI.8.1]{MacLane71}. The theorem states that if $M$ is the monad that arises from an adjunction
$F \dashv U$ for the forgetful functor $U\colon \Cat{A} \to \Sets$, then the comparison functor provides a concrete isomorphism from $(\Sigma, E)$-algebras to $\EM(M)$. Hence, our proof of the presentation using Proposition~\ref{prop:presentations} can also be (indirectly) justified this way.

\end{rem}

In the sequel, we present several algebraic theories that give presentations to the monads of interest.

\mypar{Presenting the monad $\nePow$}
Let $\Sigma_N$ be the signature consisting of a binary operation $\cplus$.
Let $E_N$ be the following set of axioms, stating that $\cplus$ is associative, commutative, and idempotent, respectively.
\[\begin{array}{ccc}
(x\cplus y)\cplus z& \stackrel{(A)}{=}& x\cplus(y\cplus z) \\
x\cplus y& \stackrel{(C)}{=}& y\cplus x\\
x\cplus x&\stackrel{(I)}{=}&x
\end{array}\]

The algebraic theory $(\Sigma_N,E_N)$ of \emph{semilattices} provides a presentation for the monad $\nePow$. We refer to this theory as the theory of nondeterminism. To avoid confusion later, it is convenient to fix here the interpretation of $\cplus$ as a join (rather than a meet)  and, thus, to think of the induced order as $x \sqsubseteq y$ iff $x\cplus y=y$.

\mypar{Presenting the monad $\mathcal{D}$}
Let $\Sigma_P$ be the signature consisting of binary operations $\pplus{p} $ for all $p\in (0,1)$.
Let $E_P$ be the following set of axioms.
\[\begin{array}{ccc}
(x\pplus{q} y)\pplus{p} z& \stackrel{(A_p)}{=}& x \pplus{pq}(y\pplus{\frac{p(1-q)}{1-pq}}z)\\
x\pplus{p} y& \stackrel{(C_p)}{=} & y\pplus{1-p}x\\
x\pplus{p} x&\stackrel{(I_p)}{=}&x
\end{array}\]
Here, $(A_p)$, $(C_p)$, and $(I_p)$ are the axioms of parametric associativity, commutativity, and idempotence.
The algebraic theory $(\Sigma_P,E_P)$ of \emph{convex algebras}, see~\cite{swirszcz:1974,semadeni:1973,doberkat:2006,doberkat:2008,jacobs:2010}, provides a presentation for the monad $\mathcal{D}$.

Another presentation of convex algebras is given by the algebraic theory with infinitely many operations denoting arbitrary (and not only binary) convex combinations $(\Sigma_{\hat P}, E_{\hat P})$ where $\Sigma_{\hat P}$ consists of operations $\sum_{i=1}^n p_i(\cdot)_i$ for all $n \in \mathbb N$ and $(p_1, \dots, p_n) \in [0,1]^n$ such that $\sum_{i=1}^n p_i = 1$ and
$E_{\hat P}$ is the set of the following two axioms.
\begin{align*}
	\sum_{i=0}^n p_ix_i & \stackrel{(P)}{=}  x_j  &\text{ if } p_j = 1\\
\sum_{i=0}^n p_i \left(\sum_{j=0}^m q_{i,j} x_j\right) & \stackrel{(BC)}{=}
			\sum_{j=0}^m \left( \sum_{i=0}^n p_i q_{i,j}\right) x_j.
\end{align*}
Here, $(P)$ stands for projection, and $(BC)$ for barycentre.
This allows us to interchangeably use binary convex combinations or arbitrary convex combinations whenever more convenient. Moreover, we can write binary convex combinations $\pplus{p}$ for $p \in [0,1]$ and not just $p \in (0,1)$.

Convex algebras are known under many names: ``convex modules'' in~\cite{pumpluen.roehrl:1995}, ``positive convex structures''
in~\cite{doberkat:2006} (where $X$ is taken to be endowed with the discrete topology), ``sets with a convex structure''
in~\cite{swirszcz:1974}, and barycentric algebras~\cite{stone:1949}.
We refer to the theory of convex algebras as the algebraic theory for probability.

\begin{rem}\label{rem:bin-suff}
	Let $\Alg{X}$ be a $(\Sigma_{\hat P},E_{\hat P})$-algebra. Then (for $p_n \neq 1$ and $\overline{p_n} = 1 - p_n$)
	\begin{equation}\label{eq:bin-suff}
		\sum_{i=1}^{n} p_i x_i = \overline{p_n} \left( \sum_{j=1}^{n-1} \frac{p_j}{\overline{p_n}} x_j\right) + p_n x_n
		.
	\end{equation}
	Hence, an $n$-ary convex combination can be written as a binary convex combination using an $(n-1)$-ary convex
	combination.

One can also see Equation~(\ref{eq:bin-suff}) as a definition ---  the classical definition of Stone~\cite[Definition~1]{stone:1949}.
The following property, whose proof follows by induction along the lines of~\cite[Lemma~1--Lemma~4]{stone:1949}, gives the connection:

Let $X$ be the carrier of a $(\Sigma_P,E_P)$-algebra.
	Define $n$-ary convex operations inductively by the projection axiom and the formula~(\ref{eq:bin-suff}). Then $X$ becomes an algebra in $(\Sigma_{\hat P},E_{\hat P})$.
\end{rem}

\mypar{Presenting $\cdot+1$}
The algebraic theory $(\Sigma_{T}, E_{T})$ for the termination monad consists of a single constant (nullary operation symbol) $\Sigma_{T} = \{\star\}$ and no equations $E_{T} = \emptyset$. This is called the theory of \emph{pointed sets}.

\mypar{Combining Algebraic Theories}
Algebraic theories can be combined  in a number of general ways: by taking their coproduct, their tensor, or by means of distributive laws (see e.g.~\cite{hyland2002combining}). Unfortunately, these abstract constructions do not lead to a presentation for the monad we are interested in. We will thus devote the next section to show a ``hand-made'' presentation for this monad.

We conclude this section with a well known fact that can be easily proved, for instance using the distributive law in~\eqref{eq:iota}: given a presentation $(\Sigma, E)$ for a monad $\T$, the monad $\T(\cdot+1)$ is presented by the theory $(\Sigma',E)$ where $\Sigma'$ is $\Sigma$ together with an extra constant $\star$.
For instance, the  subdistributions monad $\mathcal{D}(\cdot+1)$ is presented by the theory $(\Sigma_P \cup \Sigma_{T}, E_P)$ of \emph{pointed convex algebras}, also known as \emph{positive convex algebras}. The theory $(\Sigma_N\cup \Sigma_{T},E_N)$ of \emph{pointed semilattices} provides instead a presentation for the monad $\nePow(\cdot+1)$. It is interesting to observe that the powerset monad $\Pow$ is presented by adding to $(\Sigma_N\cup \Sigma_{T},E_N)$ the equation

\[\begin{array}{ccc}
x\cplus \star&\stackrel{(B)}{=}&x
\end{array}\]
leading to the theory of \emph{semilattices with bottom}. The theory of \emph{semilattices with top} can be obtained by adding instead the following equation:

\[\begin{array}{ccc}
x\cplus \star&\stackrel{(T)}{=}&\star  .
\end{array}\]
Similar axioms can be added to the theory of pointed convex algebras $(\Sigma_P \cup \Sigma_{T}, E_P)$. The axiom
\[\begin{array}{ccc}
x\pplus{p} \star&\stackrel{(B_p)}{=}&x
\end{array}\]
makes the probabilistic structure collapse, i.e., $x\pplus{p} y=x\pplus{q} y$ holds for any $p,q\in (0,1)$:
\[\begin{array}{ccc}
x\pplus{p} y &\stackrel{(B_p)}{=} &  (x\pplus{q} \star )\pplus{p} y \\
& \stackrel{(A_p)}{=}& x \pplus{pq}(\star\pplus{\frac{p(1-q)}{1-pq}}y)\\
& \stackrel{(B_p)}{=}& x \pplus{pq}y\\
& \stackrel{(B_p)}{=}& x \pplus{pq}(\star\pplus{\frac{q(1-p)}{1-pq}}y)\\
 &\stackrel{(A_p)}{=} &  (x\pplus{p} \star )\pplus{q} y \\
 &\stackrel{(B_p)}{=} & x\pplus{q} y
\end{array}\]

At the monad level, adding the axioms $(B_p)$ can be seen as the quotient of monads $\supp\colon \mathcal{D}(\cdot +1) \Rightarrow \Pow$ mapping each sub-distribution into its support (e.g., $(x\pplus{p}y)\pplus{q} \star$ becomes $x+y$).

On the other hand, the axiom
\[\begin{array}{ccc}
x\pplus{p} \star&\stackrel{(T_p)}{=}&\star
\end{array}\]
quotients the monad $\mathcal{D}(\cdot+1)$ into $\mathcal{D}+1$, recall Example~\ref{Example-D+1}. Intuitively, each term of this theory is either a sum of only variables (a distribution) or an extra element ($\star$). This axiom describes the unique functorial way of adding termination to a convex algebra, the so-called black-hole behaviour of $\star$, cf.~\cite{SW2018}.

\section{Algebraic Theory for Nondeterminism and Probability}\label{sec:C}

In this section we recall the definition of the monad $C$ for probability and nondeterminism, give its presentation via \emph{convex semilattices}, and present examples of $C$-algebras.

\subsection{The monad \texorpdfstring{$C$}{C} of convex subsets of distributions}%
\label{sec:C-new}

The monad $C$ origins in the field of domain theory~\cite{Mislove00,Goubault-Larrecq08a,TixKP09a}, and in the work of Varacca and Winskel~\cite{Varacca03,VaraccaW06}. Jacobs~\cite{Jacobs08} gives a detailed study of (a generalisation of) this monad.

For a set $X$, $CX$ is the set of non-empty, finitely-generated convex subsets of distributions on $X$, i.e.,
\[
CX = \{ S \subseteq \Dis X \mid  S \neq \emptyset, \convex(S) = S,
S \text{ is finitely generated}\}.
\]
Recall that, for a subset $S$ of a convex algebra, $\convex(S)$ is the convex closure of $S$, i.e., the smallest convex set that contains $S$, i.e., \[\convex(S) = \{\sum p_i x_i \mid p_i \in [0,1], \sum p_i = 1, x_i \in S\}.\]
We say that a convex set $S$ is generated by its subset $B$ if $S = \convex(B)$. In such a case we also say that $B$ is a basis for $S$.
A convex set $S$ is finitely generated if it has a finite basis.

For a function $f \colon X \to Y$, $Cf\colon CX \to CY$ is given by
\[Cf(S) = \{\Dis f(d) \mid d \in S\} = \overline{\Dis f}(S).\]

The unit of $C$ is $\eta\colon X \to CX$ given by $\eta(x) = \{ \delta_x\}$.

The multiplication of $C$, $\mu\colon CCX \to CX$ can be expressed in concrete terms as follows~\cite{Jacobs08}. Given $S\in CCX$,
\[\mu (S) = \bigcup_{\Phi \in S} \{ \sum_{U \in \supp(\Phi)} \Phi(U)\cdot d \mid d \in U\}.\]
Hence, $d \in \mu(S)$ if and only if there exists $\Phi \in S$ and for all $U \in \supp(\Phi)$ there exists $d_U \in U$ such that $d = \sum_{U \in\supp(\Phi)} \Phi(U)\cdot d_U$.

\medskip

For later use it is convenient to observe that the assignment $S \mapsto \convex(S)$ gives rise to a natural transformation~\cite{Mio14,BSS17}, that we refer hereafter as
\begin{equation}\label{eq:conv}
\convex \colon \Powne \Dis \Rightarrow C
\end{equation}
The following lemma summarises the relationship between the monads $\Powne$, $\Dis$, and $C$. Note here that $\Powne\Dis$ is not a monad.
\begin{lem}\label{lemma:injmapC}
Consider the natural transformations depicted in the diagram on the left below.
%\begin{minipage}{0.35\textwidth}
\[\xymatrix{
& C \\
\Dis \ar@{=>}[r]_{\eta^{\Powne}_{\Dis}}& \Powne\Dis \ar@{=>}[u]_{\convex}& \Powne \ar@{=>}[l]^{\Powne \eta^\Dis}
}\]
%\end{minipage}
%\begin{minipage}{0.65\textwidth}
Then
\begin{enumerate}
\item $\chi^{\nePow} = \convex \after \nePow\eta^{\Dis}$ is an injective monad map.
\item $\chi^{\Dis} = \convex \after \eta^{\nePow}_{\Dis}$ is an injective monad map.
\end{enumerate}
%\end{minipage}
\end{lem}
\begin{proof}
Here, it is most important to notice that $\chi^\Dis(\varphi) = \{\varphi\}$, for any distribution $\varphi$, as singleton sets are convex, and $\chi^{\nePow}(S) = \Dis S$. Keeping this in mind, it is not difficult to check that the diagrams needed for $\chi^\Dis$ and $\chi^{\nePow}$ to be monad maps indeed commute. For the unit diagrams, we have:
\[\chi^\Dis\after \eta^\Dis(x) = \chi^\Dis(\delta_x) = \{\delta_x\} = \eta^C(x)\]
and
\[\chi^{\nePow}\after \eta^{\nePow}(x) = \chi^{\nePow}(\{x\}) = \{\delta_x\} = \eta^C(x).\]
For the multiplication diagrams, we derive, given $\Phi \in \Dis\Dis X$,
\[\chi^{\Dis}\after \mu^\Dis(\Phi) = \chi^{\Dis}\left(\sum_{\varphi \in \Dis X}\Phi(\varphi)\cdot\varphi\right) = \left\{\sum_{\varphi \in \Dis X}\Phi(\varphi)\cdot\varphi\right\}\]
and
\begin{eqnarray*}
	\mu^C\after C\chi^{\Dis}\after \chi^{\Dis}\Dis(\Phi) & = & \mu^C\after C\chi^{\Dis}(\{\Phi\})\\
		& = & \mu^C\left(\sum_{\varphi \in \Dis X} \Phi(\varphi)\cdot\{\varphi\}\right)\\
		& = & \left\{\sum_{\varphi \in \Dis X}\Phi(\varphi)\cdot\varphi\right\}.
\end{eqnarray*}
Given $\mathcal S = \{S_i \mid i \in I\}$ for $S_i \subseteq X, S_i \neq \emptyset, I \neq \emptyset$, that is  $\mathcal S\in \Powne \Powne (X)$, we have
\[\chi^{\nePow}\after \mu^{\nePow}({\mathcal S}) = \chi^{\nePow}\left(\bigcup_{i \in I} S_i\right) = \Dis\left(\bigcup_{i\in I}S_i\right)\]
and
\begin{eqnarray*}
	\mu^C\after \chi^{\nePow}C\after \nePow\chi^{\nePow}({\mathcal S}) & = & \mu^C\after \chi^{\nePow}C(\{\Dis S_i \mid i\in I\})\\
		& = & \mu^C\left(\Dis\{\Dis S_i \mid i \in I\}\right)\\
		& = & \Dis\left(\bigcup_{i\in I}S_i\right)
\end{eqnarray*}
where the inclusion $\subseteq$ in the last equality is obvious, and the inclusion $\supseteq$ can be derived by simple grouping and normalisation, as follows. Given $d \in \Dis\left(\bigcup_{i\in I}S_i\right)$, i.e.,
$d = \sum_{j \in J} p_jx_j$ where $x_j \in \bigcup_{i\in I} S_i$, let $(J_i|i \in I)$ be a partition of $J$ with the property
$j \in  J_i \Rightarrow x_j \in S_i$, and let $p_i = \sum_{j \in J_i} p_j$. Let, moreover, \[d_i = \sum_{j \in J_i}\frac{p_j}{p_i}\cdot x_j\] if $p_i \neq 0$, and $d_i = \delta_x$ for some $x \in S_i$ otherwise.
Then $d = \sum_{i \in I} p_id_i \in \Dis\{\Dis S_i \mid i \in I\}$.

Injectivity follows directly from $\chi^\Dis(\varphi) = \{\varphi\}$ and $\chi^{\nePow}(S) = \Dis S$, the latter since $\Dis S = \Dis S'$ if and only if $S = S'$.
\end{proof}

\subsection{The presentation of \texorpdfstring{$C$}{C}}%
\label{sec:presentation-C}
We now introduce the algebraic theory $(\Sigma_{NP}, E_{NP})$ of \emph{convex semilattices}, that gives us the presentation of $C$ and thus provides an algebraic theory for nondeterminism and probability.

A convex semilattice $\mathbb A$ is an algebra $\mathbb A = (A, \cplus, \pplus{p})$ with a binary operation $\cplus$ and for each $p \in (0,1)$ a binary operation $\pplus{p}$ satisfying the axioms $(A), (C), (I)$ of a semilattice, the axioms $(A_p), (C_p), (I_p)$ for a convex algebra, and the following distributivity axiom:
\[ (x \cplus y) \pplus{p} z \stackrel{(D)}{=} (x \pplus{p} z) \cplus (y \pplus{p} z)\]
Hence, $(\Sigma_{NP}, E_{NP})$ is given by $\Sigma_{NP} = \Sigma_N \cup \Sigma_P$ and $E_{NP} = E_N \cup E_P \cup \{(D)\}$.

In every convex semilattice there also holds a convexity law, of which we directly present the generalised version in the following lemma.

\begin{lem}\label{lem:gen-conv}
	Let $\mathbb A = (A, \cplus, \pplus{p})$ be a convex semilattice. Then for all $n \in \mathbb N$, all $a_1, \dots, a_n \in A$ and all $p_1, \dots, p_n \in [0,1]$ with $\sum_{i=1}^n p_i = 1$ we have

	\[a_1 \cplus \dots \cplus a_n \cplus \sum_{i=1}^n p_ia_i \stackrel{(C)}{=} a_1 \cplus \dots \cplus a_n.\qed\]
\end{lem}

\begin{proof}%[of Lemma~\ref{lem:gen-conv}]
	For $n=1$ the property amounts to idempotence. Assume $n > 1$ and the property holds for $n-1$.

	Below, we will write $(D)$ also for generalised distributivity as in
	\[\bigoplus_{i}a_i \pplus{p} \bigoplus_{j}b_j \stackrel{(D)}{=} \bigoplus_{i,j}(a_i \pplus{p} b_j).\]
	First, we observe that
\begin{equation}\label{eq:gen-conv-aux}
a_1 \cplus \dots \cplus a_n = a_1 \cplus \dots \cplus a_n \cplus \bigoplus_i(a_i \pplus{p_1} \bigoplus_{j\neq i} a_j	)
\end{equation}
	which follows from
	\begin{eqnarray*}
		a_1 \cplus \dots \cplus a_n &\stackrel{(I_p)}{=}& (a_1 \cplus \dots \cplus a_n) \pplus{p_1} (a_1 \cplus \dots \cplus a_n)\\
		&\stackrel{(D)}{=}& \bigoplus_{i,j}(a_i \pplus{p_1} a_j) \\
		&\stackrel{(I_p,D)}{=}& a_1 \cplus \dots \cplus a_n \cplus \bigoplus_i(a_i \pplus{p_1} \bigoplus_{j\neq i} a_j)\\
	\end{eqnarray*}
Recall that we write $\overline{p}$ for $1-p$ if $p \in [0,1]$.
Furthermore, having in mind that $\sum_{i=1}^n p_ia_i  = a_1 \pplus{p_1} (\sum_{i=2}^n \frac{p_i}{\overline{p_1}} a_i)$ we have
	\begin{eqnarray*}
		a_1 \pplus{p_1} (\bigoplus_{j\neq 1} a_j) &\stackrel{IH}{=}& a_1 \pplus{p_1} (\bigoplus_{j\neq 1} a_j \cplus \sum_{j\neq 1} \frac{p_j}{\overline{p_1}}a_j)\\
		&\stackrel{(D)}{=}& (a_1 \pplus{p_1} \bigoplus_{j\neq 1} a_j) \oplus (a_1 \pplus{p_1} \sum_{j\neq 1} \frac{p_j}{\overline{p_1}}a_j)\\
		&=& (a_1 \pplus{p_1} \bigoplus_{j\neq 1} a_j) \oplus \sum_i p_ia_i.
	\end{eqnarray*}
	Using this in the second equality below, we get
	\begin{align*}
		a_1 \cplus \dots \cplus a_n \cplus \sum_{i=1}^n p_ia_i
		&\stackrel{(\ref{eq:gen-conv-aux})}{=} a_1 \cplus \dots \cplus a_n \cplus \bigoplus_i(a_i \pplus{p_1} \bigoplus_{j\neq i} a_j)	\cplus \sum_{i=1}^n p_ia_i\\
		&= a_1\cplus \dots \cplus a_n \cplus \bigoplus_i(a_i \pplus{p_1} \bigoplus_{j\neq i} a_j)\\
		&\stackrel{(\ref{eq:gen-conv-aux})}{=} a_1 \cplus \dots \cplus a_n.  \qedhere
	\end{align*}
\end{proof}

We next formulate a property that provides a way to prove that an algebraic theory is a presentation for a monad, which we will later apply in the proof that $(\Sigma_{NP}, E_{NP})$ is a presentation for $C$.

\begin{prop}\label{prop:presentations} Let $\Cat{A}$ be the category of $(\Sigma,E)$-algebras with signature $\Sigma$ and equations $E$. Let $U\colon \Cat{A} \to \Sets$ be the forgetful functor. In order to prove that $(\Sigma, E)$ is a presentation for a monad $(\T, \eta, \mu)$, it suffices to:
\begin{itemize}
\item[1.] For any set $X$, define $\Sigma$-operations $\Sigma_X$ on $\T X$ and prove that with these operations $(\T X, \Sigma_X)$ is an algebra in $\Cat{A}$. Moreover prove that for any map $f \colon X \to Y$, $\T f$ is an $\Cat{A}$-homomorphism from $(\T X, \Sigma_X)$ to $(\T Y, \Sigma_Y)$.
\item[2.] Prove that $(\T X, \Sigma_X)$ is the free algebra in $\Cat{A}$ generated by $X$, i.e., for any algebra $\Alg A = (A,\Sigma_A)$ in $\Cat{A}$ and any map $f\colon X \to A$, there is a unique homomorphism $f^\#\colon (\T X, \Sigma_X) \to \Alg A$ that extends $f$, i.e., that satisfies $f = Uf^\#\after\eta$.
\item[3.] Prove that $\mu_X = (\id_{\T X})^\#$.
\end{itemize}
\end{prop}

\begin{proof}
Assume that 1.-3.\ hold. We show that $\T$ is presented by $(\Sigma, E)$% using~\cite[Theorem VI.8.1]{MacLane71}
, as explained in Remark~\ref{rem:conc-iso}.
%The theorem states that the category of $(\Sigma,E)$-algebras is concretely isomorphic to the category of Eilenberg-Moore algebras of the monad defined by an adjunction $F \dashv U$ with $F$ mapping a set $X$ to the free $(\Sigma,E)$-algebra generated by $X$ and $U$ being the forgetful functor. We prove that the monad defined by this adjunction coincides with $\T$. We start with defining the adjunction explicitly.
We provide a left adjoint to the forgetful functor such that $\T$ is the monad that arises from this adjunction. We start with defining the adjunction explicitly.

Let $F\colon \Sets \to \Cat{A}$ be the functor defined on objects as $F X = (\T X,\Sigma_{X})$. Hence, $UFX = \T X$. On arrows $f\colon X \to Y$, we set $Ff = (\eta \after f)^{\#}$.

Then $F$ is a left adjoint of the forgetful functor $U$, and the adjunction is given by the bijective correspondence $(f \colon X \to U\mathbb A) \mapsto (f^{\#} \colon FX \to \mathbb A)$. This is injective since: $f^{\#} = g^{\#} \Rightarrow f = Uf^{\#}\after \eta = Ug^{\#}\after \eta = g$. It is surjective since: for a homomorphism $h^*\colon FX \to \mathbb A$, we consider $h = Uh^* \after \eta$ and since $h^\#$ is the unique homomorphism such that $h = Uh^\# \after \eta$, we get that $h^\# = h^*$.

Next, we see that $UF f = \T f$ as a consequence of naturality of $\eta$. Namely, we have
that $F f = (\eta\after f)^\#$ is the unique homomorphism with the property $UF f\after \eta = \eta \after f$. Hence, using 1., since $\T f\after \eta = \eta \after f$ by naturality of $\eta$, we get $UF f = \T f$.

Let $(T, \bar\eta, \bar\mu)$ be the monad of this adjunction.
This means that, see e.g.~\cite[VI.1, IV.1]{MacLane71}, $T = UF$, $(\bar\eta)^\# = \id_{FX}$ and $\bar\mu = U\varepsilon F$ where
$\varepsilon_{\mathbb A} = (\id_{U \mathbb A})^\#$ is the counit of the adjunction, and hence $\bar\mu_X = (\id_{\T X})^\#$.
We next show that
$\eta^\# = \id_{FX}$ which implies $\bar\eta = \eta$. All we need to observe is that $U\id_{FX} \after \eta = \id_{UFX} \after \eta = \eta$ and since $\eta^\#$ is the unique homomorphism with $U\eta^\# \after \eta = \eta$ and $\id_{FX}$ is a homomorphism from $FX$ to itself, we get $\eta^\# = \id_{FX}$. Finally, item 3.\ proves that $\bar\mu = \mu$.
\end{proof}

For $p \in [0,1]$ we set $\overline{p} = 1- p$.  Let $X$ be an arbitrary set. We define $\Sigma_{NP}$-operations on $CX$ by
\[S_1 \cplus S_2 = \convex(S_1 \cup S_2)\] and for $p \in (0,1)$
\[S_1 \pplus{p} S_2 = \{\varphi \mid \varphi = p\varphi_1 + \overline{p}\varphi_2 \text{ for some } \varphi_1 \in S_1, \varphi_2 \in S_2\}\]
where $p\varphi_1 + \overline{p}\varphi_2 = \varphi_1 \pplus{p} \varphi_2$ is the binary convex combination of $\varphi_1$ and $\varphi_2$ in $\Dis X$, defined point-wise. Note that $S_1 \pplus{p} S_2$ is the Minkowski sum of two convex sets. If convenient, we may sometimes also write, as usual, $pS_1 + \overline{p}S_2$ for the Minkowski sum $S_1 \pplus{p} S_2$.

The proof of the presentation follows the structure of Proposition~\ref{prop:presentations} via the following three lemmas.

\begin{lem}\label{lem:CX-alg} With the above defined operations
	$(CX, \cplus, \pplus{p})$ is a convex semilattice, for any set $X$. Moreover, for a map $f \colon X \to Y$, the map $Cf\colon CX \to CY$ is a convex semilattice homomorphism from $(CX, \cplus, \pplus{p})$ to $(CY, \cplus, \pplus{p})$. \qed%
\end{lem}

\begin{lem}\label{lem:CX--free-alg}
	The convex semilattice
	$(CX, \cplus, \pplus{p})$ is the free convex semilattice generated by $\eta(X)$. \qed%
\end{lem}

\begin{lem}\label{lem:CX--mu-pres}
	The multiplication $\mu$ of the monad $C$ satisfies $\mu = (\id_{CX})^\#$. \qed%
\end{lem}

All detailed proofs are in Section~\ref{app:proofs-for-C}, where we also list some additional helpful properties.
Now all ingredients are in place and we get the presentation for $C$ directly from Proposition~\ref{prop:presentations}.

\begin{thm}\label{th:pres-C}
	The theory for nondeterminism and probability $(\Sigma_{NP},E_{NP})$, i.e., the theory of convex semilattices, is a presentation for the monad $C$. \qed%
\end{thm}

\begin{rem}\label{rem:Varacca-Winskel}
Theorem~\ref{th:pres-C} is to some extent known\footnote{Personal communication with Gordon Plotkin.} but we could not find a proof of it in the literature. In~\cite{Varacca03,VaraccaW06} a monad for probability and nondeterminism is given starting from a similar algebraic theory (with somewhat different basic algebraic structure).
The observation that the distributive law $(D)$ gives rise to convex subsets was known at least since~\cite{MOW03}. Indeed, all the axioms in $E_{NP}$ already appear in~\cite{MOW03} that additionally contains the axiom $(B)$ and $(T_p)$ for dealing with termination. Nevertheless, in~\cite{MOW03} these axioms are not related to any particular monad.
There is also another possible way of combining probability with nondeterminism, by distributing $\cplus$ over $\pplus{p}$ (see e.g.~\cite{DBLP:journals/corr/KeimelP16,DBLP:conf/ictac/DahlqvistP018}).
\end{rem}

\begin{rem}\label{rem:syntax}
Having the presentation enables us to identify and interchangeably use convex subsets of distributions and terms in $\Sigma_{NP}$ modulo equations in $E_{NP}$. This is particularly useful in examples and our further developments. Note that in the syntactic view $\eta(x)$ is identified with the term $x$.
\end{rem}

The presentation is a valuable tool in many situations where reasoning with algebraic theories is more convenient than reasoning with monads. For instance, it is much easier to check whether a certain algebra is a $(\Sigma_{NP},E_{NP})$-model, than to check that it is an algebra for the monad $C$. We illustrate this with three $(\Sigma_{NP},E_{NP})$ models that play a key role in our further results and exposition.

\mypar{The max convex semilattice} $\mathbb{M}\rm{ax} = ([0,1],\max,\pplus{p})$ is  a $(\Sigma_{NP},E_{NP})$-algebra when taking $\cplus$ to be $\max\colon [0,1] \times [0,1] \to [0,1]$ and $\pplus{p}$ the standard convex combination
 $\pplus{p} \colon [0,1]\times [0,1] \to [0,1]$ with $x \pplus{p} y =p\cdot x +\overline{p} \cdot y$ for $x,y\in [0,1]$. To check that this is a  $(\Sigma_{NP},E_{NP})$ model, it is enough to prove that $\max$ satisfies the axioms in $E_N$, that $\pplus{p}$ satisfies the axioms in $E_P$, and that they satisfy the axiom $(D)$, namely that $\max(x,y)\pplus{p} z= \max(x\pplus{p} z, y\pplus{p} z)$.

\mypar{The min convex semilattice} $\mathbb{M}\rm{in} = ([0,1],\min,\pplus{p})$
 is obtained similarly by taking $\cplus$ to be $\min\colon [0,1] \times [0,1] \to [0,1]$ rather than $\max$, and gives another example of a $(\Sigma_{NP},E_{NP})$-algebra. It is indeed very simple to check that $([0,1],\min)$ forms a semilattice and that the distributivity law holds.

\mypar{The min-max interval convex semilattice}
We consider the algebraic structure $\mathbb{M}_{\intervals} = (\intervals,\minmax,\psuminterval)$ for $\intervals$ the set of intervals on $[0,1]$, i.e., \[\intervals = \{[x,y] \,|\, x,y\in [0,1] \text{ and } x\leq y\}.\] For $[x_1,y_1], [x_2,y_2]\in \intervals$, we define $\minmax \colon \intervals \times \intervals \to \intervals$  as \[\minmax( [x_1,y_1] , [x_2,y_2]) = [\min(x_1,x_2), \max(y_1,y_2)]\] and $\psuminterval \colon \intervals \times \intervals \to \intervals$ by \[[x_1,y_1] \psuminterval [x_2,y_2]= [ x_1 \pplus{p} x_2, \;  y_1\pplus{p}y_2]\text{.}\] The fact that this is a model for $(\Sigma_{NP},E_{NP})$ follows easily from the fact that $\mathbb{M}\rm{ax}$ and $\mathbb{M}\rm{in}$ are models for $(\Sigma_{NP},E_{NP})$.

\begin{rem}\label{rem:others}
The fact that $\maxalg$ and $\minalg$ are $C$-algebras on $[0,1]$ was already proven in~\cite{HHOS2018}, without an algebraic presentation. Having the algebraic presentation significantly simplifies the proofs.
\end{rem}

\subsection{Auxiliary Lemmas for the Proof of the Presentation of \texorpdfstring{$C$}{C}}\label{app:proofs-for-C}

Before we proceed with the proof of the presentation, we recall several properties that are known or immediate to check, but very helpful in our further proofs.

\begin{lem}\label{lem:convex-image-is-convex2}
Let $\mathbb A$ and $\mathbb B$ be two convex algebras, and $f\colon \mathbb A \to \mathbb B$ a convex homomorphism. Then for all $X \in \Pow_u A$, for $\Pow_u$ being the unrestricted (not necessarily finite) powerset, $\convex_{\mathbb B} \overline{f}(X) = \overline{f}(\convex_{\mathbb A} X)$.
In particular, if $X$ is convex then also $\overline{f}(X)$ is convex.
\end{lem}

\begin{proof}
	For $\subseteq$, for an arbitrary $p f(x) + \overline{p} f(y) \in \convex_{\mathbb B} \overline{f}(X)$ with $x, y \in X$, we have
	\[p f(x) + \overline{p} f(y) \stackrel{(*)}{=} f(px + \overline{p}y) \in \overline{f}(\convex_{\mathbb A} X)\]
	where the equality marked by $(*)$ holds by the assumption that $f$ is a convex homomorphism.
	For $\supseteq$, consider $f(a) \in \overline{f}(\convex_{\mathbb A}X)$.  Then $a = px + \overline{p}y$ for some $x, y \in X$. Since $f$ is convex, $f(a) = pf(x) + \overline{p} f(y)$ and $f(x), f(y) \in \overline{f}(X)$. Hence $f(a) \in \convex_{\mathbb B}\overline{f}(X)$.
\end{proof}

\begin{lem}\label{lem:convex-image-is-convex1} Let $\mathbb A$ and $\mathbb B$ be two convex algebras, and $f\colon \mathbb A \to \mathbb B$ a convex homomorphism. Then the image map $\overline{f} = \Pow_u f \colon \Pow_u A \to \Pow_u B$, for $\Pow_u$ being the unrestricted (not necessarily finite) powerset, is a convex map, i.e.\ if $S = X \pplus{p} Y$ for $X \in\Pow_u A, Y \in \Pow_u B$, then $\overline{f}(S) =\overline{f}(X) \pplus{p}\overline{f}(Y)$.
\end{lem}

\begin{proof}
	Let $S = X \pplus{p} Y$ for $X \in\Pow_u A, Y \in \Pow_u B$. Then
	\begin{eqnarray*}
		\overline{f}(S) &=& \{f(s) \mid s \in S\}\\
		&=& \{f(px + \overline{p}y) \mid x \in X, y \in Y\}\\
		&\stackrel{(*)}{=}& \{p f(x) + \overline{p}f(y) \mid x \in X, y \in Y\}\\
		&=& p\overline{f}(X) + \overline{p}\overline{f}(Y).
	\end{eqnarray*}
	and here, again, $(*)$ holds since $f$ is convex.
\end{proof}

\begin{lem}\label{lem:union-is-convex}Let $X$ be a set and let $S \in  \Pow_u CX$ be a convex set with respect to Minkowski sum. Then $\bigcup S \in CX$.
\end{lem}

\begin{proof}
	Let $S = \{S_i \mid i \in I\}$. Let $\Phi, \Psi \in \bigcup S$. Then there exist $i, j \in I$ with $\Phi \in S_i$ and $\Psi \in S_j$. We have $p\Phi + \overline{p}\Psi \in pS_i + \overline{p}S_j \in S$ as $S$ is convex.
\end{proof}

We can now prove Lemmas~\ref{lem:CX-alg},~\ref{lem:CX--free-alg}, and~\ref{lem:CX--mu-pres} announced in the previous section, from which the presentation of $C$ follows.

\begin{proof}[Proof of Lemma~\ref{lem:CX-alg}]
We want to prove the following:
	$(CX, \cplus, \pplus{p})$ is a convex semilattice, for any set $X$. Moreover, for a map $f \colon X \to Y$, the map $Cf\colon CX \to CY$ is a convex semilattice homomorphism from $(CX, \cplus, \pplus{p})$ to $(CY, \cplus, \pplus{p})$.

	In any convex algebra $\mathbb A$ for $S,T \subseteq A$ we have \[\convex(\convex(S) \cup T) = \convex(S \cup T).\]
	As a consequence, using the associativity of union, we get that the axiom $(A)$ holds. For $S_1, S_2, S_3 \in CX$:
	\begin{eqnarray*}
		S_1 \cplus (S_2 \cplus S_3) &=& \convex(S_1 \cup \convex(S_2 \cup S_3))\\
		&=& \convex(S_1 \cup (S_2 \cup S_3))\\
		&=& \convex((S_1 \cup S_2) \cup S_3)\\
		&=& \convex(\convex(S_1 \cup S_2) \cup S_3)\\
		&=& (S_1 \cplus S_2) \cplus S_3.
	\end{eqnarray*}
Commutativity and idempotence hold due to commutativity and idempotence of union.

Defining convex operations on $CX$ using Minkowski sum, see~\cite{BSS17}, leads to a convex algebra, i.e., $(A_p), (C_p), (I_p)$ hold.

The axiom $(D)$ holds as:
\begin{align*}
	(S_1 \cplus S_2) \pplus{p} S_3 \\
	& =  p\convex(S_1 \cup S_2) + \overline{p} S_3  \\
	& =  \{pq d_1 + p\overline{q} d_2 + \overline{p} d_3 \mid q \in [0,1], d_i \in S_i\}\\
	& =  \convex((pS_1 + \overline{p} S_3) \cup (pS_2 + \overline{p} S_3)).
\end{align*}

Finally, $Cf$ is a homomorphism from $(CX, \cplus, \pplus{p})$ to $(CY, \cplus, \pplus{p})$ as
\begin{eqnarray*}
	Cf(S_1 \cplus S_2) &=& \overline{\Dis f}(S_1 \cplus S_2)\\
	&\stackrel{(a)}{=}& \convex(\overline{\Dis f}(S_1 \cup S_2))\\
	&=& \convex(\overline{\Dis f}(S_1) \cup \overline{\Dis f}(S_2))\\
	& = & \overline{\Dis f}(S_1) \cplus \overline{\Dis f}(S_2)\\
	& = & Cf(S_1) \cplus Cf(S_2)
\end{eqnarray*}
where the equality marked by $(a)$ holds by Lemma~\ref{lem:convex-image-is-convex2}. Similarly
\begin{eqnarray*}
	Cf(S_1 \pplus{p} S_2) &=& \overline{\Dis f}(S_1 \pplus{p} S_2)\\
	&\stackrel{(b)}{=}& \overline{\Dis f}(S_1) \pplus{p} \overline{\Dis f}(S_2)\\
	& = & Cf(S_1) \pplus{p} Cf(S_2).
\end{eqnarray*}
where the equality marked by $(b)$ holds by Lemma~\ref{lem:convex-image-is-convex1}.
\end{proof}

\begin{proof}[Proof of Lemma~\ref{lem:CX--free-alg}]
We want to prove the following:
	The convex semilattice
	$(CX, \cplus, \pplus{p})$ is the free convex semilattice generated by $\eta(X)$.

We need to show that for any map $f\colon X \to A$ for a convex semilattice $\mathbb A = (A, \cplus, +_p)$, there is a unique convex semilattice homomorphism $f^\#\colon (CX,\cplus,+_p) \to \mathbb A$ such that $Uf^\# \after \eta = f$. So, let $\mathbb A = (A, \cplus, +_p)$ be a convex semilattice, and let $f\colon X \to A$ be a map. We use the same notation for the operations in $A$ and in $CX$ for simplicity.

	Note that, since any convex semilattice is a convex algebra, there is a unique convex homomorphism $f_{\Dis}^\#\colon \Dis X \to (A,+_p)$, as $\Dis X$ is the free convex algebra generated by $\eta_{\Dis}(X)$. Hence, $Uf_{\Dis}^\# \after \eta_\Dis = f$.

	Now, given a convex set $S = \convex\{d_1, \dots, d_n\} \in CX$ we put
	\[f^\#(S) = f_{\Dis}^\#(d_1) \cplus f_{\Dis}^\#(d_2) \cplus \dots \cplus f_{\Dis}^\#(d_n).\]

We first prove that $f^\#$ is well defined, which is the most important step. We show that whenever
\begin{equation}\label{eq:well-def-as}
	\convex\{d_1, \dots, d_n\} = \convex\{e_1, \dots, e_m\}
\end{equation}
then
\[\f_{\Dis}^\#(d_1) \cplus \dots \cplus f_{\Dis}^\#(d_n) = f_{\Dis}^\#(e_1) \cplus \dots \cplus f_{\Dis}^\#(e_m).\]
Clearly, if Equation~(\ref{eq:well-def-as}) holds, then for all $i \in \{1, \dots,n\}$, $d_i \in \convex\{e_1, \dots, e_m\}$ and for all $j \in \{1, \dots,m\}$, $e_j \in \convex\{d_1, \dots, d_n\}$.
Hence,
\[
\convex\{d_1, \dots, d_n, e_1, \dots, e_m\} = \convex\{d_1, \dots, d_n\}
= \convex\{e_1, \dots, e_n\}.
\]
If we can prove that whenever $e \in \convex\{d_1, \dots, d_n\}$ then
\[f_{\Dis}^\#(d_1) \cplus \dots \cplus f_{\Dis}^\#(d_n) \cplus f_{\Dis}^\#(e) = f_{\Dis}^\#(d_1) \cplus \dots \cplus f_{\Dis}^\#(d_n),\] we would be done with well defined-ness as then
\begin{eqnarray*}
&&	f_{\Dis}^\#(d_1) \cplus \dots \cplus f_{\Dis}^\#(d_n)\\
&&  = \quad f_{\Dis}^\#(d_1) \cplus \dots \cplus f_{\Dis}^\#(d_n) \cplus f_{\Dis}^\#(e_1) \cplus \dots \cplus f_{\Dis}^\#(e_m)\\
&&  = \quad f_{\Dis}^\#(e_1) \cplus \dots \cplus f_{\Dis}^\#(e_m).
\end{eqnarray*}
 So, let $e \in \convex\{d_1, \dots, d_n\}$. Then $e = \sum_{i} p_id_i$ and since $f_{\Dis}^\#$ is a convex algebra homomorphism, $f_{\Dis}^\#(e) = \sum_i p_i f_{\Dis}^\#(d_i)$.
 Now, by the convexity law, Lemma~\ref{lem:gen-conv}, we have
 \begin{eqnarray*} f_{\Dis}^\#(d_1) \cplus \dots \cplus f_{\Dis}^\#(d_n) \cplus f_{\Dis}^\#(e) & = &  f_{\Dis}^\#(d_1) \cplus \dots \cplus f_{\Dis}^\#(d_n) \cplus \sum_i p_i f_{\Dis}^\#(d_i)\\
& = & f_{\Dis}^\#(d_1) \cplus \dots \cplus f_{\Dis}^\#(d_n).
\end{eqnarray*}

It remains to show that $f^\#$ is a homomorphism and that it is uniquely extending $f$ on $\eta(X)$. Let $S, T \in CX$. Let $S = \convex\{d_1,\dots,d_n \}$, $T = \convex\{e_1, \dots, e_m\}$.

 Then $S \cplus T = \convex(S \cup T) = \convex\{d_1,\dots,d_n, e_1, \dots, e_m \}$ and we get
 \[
 f^\#(S\cplus T)
 =  f_{\Dis}^\#(d_1) \cplus \dots \cplus f_{\Dis}^\#(d_n) \cplus f_{\Dis}^\#(e_1) \cplus \dots \cplus f_{\Dis}^\#(e_m)
 = f^\#(S) \cplus f^\#(T).
\]

 Next, we first notice that $S \pplus{p} T = \convex\{pd_i + \overline{p} e_j \mid i \in \{1,\dots, n\}, j \in \{1, \dots,m\}\}$.
 For $\supseteq$, we see that \[\sum_{i,j} q_{i,j}(pd_i + \overline{p} e_j) = p\sum_{i,j}q_{i,j}d_i + \overline{p}\sum_{i,j}q_{i,j}e_j \in S \pplus{p} T.\]
 For $\subseteq$, take $pd + \overline{p}e \in S \pplus{p} T$. So, $d = \sum_i q_i d_i$ and $e = \sum_j r_je_j$ and we have
 \begin{eqnarray*}
 	pd + \overline{p}e &=& p\sum_i q_i d_i + \overline{p}\sum_j r_j e_j\\
 	&=& p\sum_i q_i \left(\sum_j r_j\right)d_i  + \overline{p}\sum_j r_j \left(\sum_i q_i\right) e_j \\
 	& = & \sum_{i,j} q_ir_j (pd_i + \overline{p}e_j).
 \end{eqnarray*}

 Now
 \begin{eqnarray*}
 f^\#(S \pplus{p} T)
	&=&  \bigoplus_{i,j} f_{\Dis}^\#(pd_i + \overline{p} e_j)\\
 	&=& \bigoplus_{i,j} pf_{\Dis}^\#(d_i) + \overline{p} f_{\Dis}^\#(e_j)\\
 	&=& \bigoplus_{i,j} f_{\Dis}^\#(d_i) \pplus{p} f_{\Dis}^\#(e_j)\\
 	&\stackrel{(D)}{=}& ( f_{\Dis}^\#(d_1) \cplus \dots \cplus f_{\Dis}^\#(d_n)) \pplus{p} (f_{\Dis}^\#(e_1) \cplus \dots \cplus f_{\Dis}^\#(e_m))\\
 	&=& f^\#(S) \pplus{p} f^\#(T).
 \end{eqnarray*}

\medskip

Finally, assume $f^*\colon (CX,\cplus,+_p) \to \mathbb A$ is another homomorphism that extends $f$ on $\eta(X)$, i.e., such that $Uf^* \after \eta = f$.
 Then $f^*(\{\delta_x\}) = f^\#(\{\delta_x\}) = f(x)$.
 Since both $f^\#$ and $f^*$ are convex homomorphisms, and $\{\sum_i p_i x_i\} = \sum_i p_i \{\delta_{x_i}\}$, we get
\[
 f^*(\{\sum_i p_i x_i\})
 = \sum_i p_i f^*(\{ \delta_{x_i}\})
 =  \sum_i p_i f^\#(\{\delta_{x_i}\})
 = f^\#(\{\sum_i p_i x_i\}).
 \]
 Further on, for $S = \convex\{d_1, \dots, d_n\}$ we have $S = \{d_1\} \cplus \dots \cplus \{d_n\}$ and hence
 $f^*(S) = f^*(\{d_1\}) \cplus \dots \cplus f^*(\{d_n\}) = f^\#(\{d_1\}) \cplus \dots \cplus f^\#(\{d_n\}) = \f^\#(S)$
 shows that $f^* = f^\#$ and completes the proof.
\end{proof}

The final missing property for the presentation, Lemma~\ref{lem:CX--mu-pres}, is an easy consequence of the next property that clarifies the definition of $f^\#$.

\begin{lem}\label{lem:f-sharp}
Let $X$ be a set and $f\colon X \to CY$ a map. Then for all $S$ in $CX$
\[f^\#(S) =  \bigcup \overline{f_{\Dis}^\#}(S) = \bigcup_{\Phi \in S} \sum_{u \in \supp(\Phi)} \Phi(u)\cdot f(u).\]
\end{lem}

\begin{proof}
	The first task is to prove that $f^\#(S) = \bigcup \overline{f_{\Dis}^\#}(S)$.
	Before we proceed, let's recall all the types. We have $f\colon X \to CY$ (and $CY$ is the carrier of a convex semilattice), so $f^\# \colon CX \to CY$. Also, $f_{\Dis}^\# \colon \Dis X \to CY$ and hence $\overline{f_{\Dis}^\#}\colon \Pow_u\Dis X \to \Pow_u CY$ for $\Pow_u$ denoting the unrestricted (and not just finite) powerset. Finally, here $\bigcup\colon \Pow_u\Pow_u\Dis Y \to \Pow_u\Dis Y$. Clearly, $CZ \subseteq \Pow_u\Dis Z$ for any set $Z$.

	Now, since $S$ is convex, by Lemma~\ref{lem:convex-image-is-convex2} also $\overline{f_{\Dis}^\#}(S)$ is convex. Each element of $\overline{f_{\Dis}^\#}(S)$ is of the form $f_{\Dis}^\#(\Phi)$ for $\Phi \in S$ and hence it is in $CY$, i.e., is convex.
	By Lemma~\ref{lem:union-is-convex}, we get that $\bigcup \overline{f_{\Dis}^\#}(S)$ is convex.

	Let $\Psi_1, \dots, \Psi_n \in \Dis X$ be such that $S = \convex\{\Psi_1, \dots, \Psi_n\}$. Clearly, $\Psi_1, \dots, \Psi_n \in S$. Now, we have
	\[\{f_{\Dis}^\#(\Psi_i) \mid i = 1, \dots, n\} \subseteq \{f_{\Dis}^\#(\Phi) \mid\Phi \in S\}\]
	and hence
	\[
	\bigcup\{f_{\Dis}^\#(\Psi_i) \mid i = 1, \dots, n\} \subseteq \bigcup\{f_{\Dis}^\#(\Phi) \mid\Phi \in S\}
	= \bigcup \overline{f_{\Dis}^\#}(S)
	\]
	and since the set on the right hand side is convex, as we noted above,
	\[
	f^\#(S)  = \bigoplus_i f_{\Dis}^\#(\Psi_i) = \convex \bigcup\{f_{\Dis}^\#(\Psi_i) \mid i = 1, \dots, n\}
	\subseteq \bigcup\{f_{\Dis}^\#(\Phi) \mid\Phi \in S\},
	\]
	where the first equality is simply the definition of $f^
\#$.

	For the other inclusion, let $\Phi \in S$. Then $S = \convex\{\Psi_1, \dots, \Psi_n, \Phi\}$ and
	\[f^\#(S) = \convex \bigcup\{f_{\Dis}^\#(\Psi_1), \dots,  f_{\Dis}^\#(\Psi_n), f_{\Dis}^\#(\Phi)\}\]
	by the definition of $f^\#$.
	Therefore, $f_{\Dis}^\#(\Phi) \subseteq f^\#(S)$ and since $\Phi$ was arbitrary,
	\[\bigcup\{f_{\Dis}^\#(\Phi) \mid\Phi \in S\} \subseteq f^\#(S).\]
	This proves the first equality of our statement.
	For the second equality, note that
	\begin{eqnarray*}f^\#(S) & = & \bigcup \{f_{\Dis}^\#(\Phi) \mid \Phi \in S \}\\
& \stackrel{(*)}{=} & 	\bigcup\{ \sum_{u \in \supp(\Phi)} \Phi(u)\cdot f(u) \mid \Phi \in S\}\\
& = & \bigcup_{\Phi \in S} \sum_{u \in \supp(\Phi)} \Phi(u)\cdot f(u)
\end{eqnarray*}
	where the  equality $(*)$ holds as $f^\#_{\Dis}$ is convex.
\end{proof}

\begin{proof}[Proof of Lemma~\ref{lem:CX--mu-pres}] Using Lemma~\ref{lem:f-sharp}, we immediately get
\[(\id_{CX})^\#(S)  =  \bigcup_{\Phi \in S} \sum_{A \in \supp(\Phi)} \Phi(A)\cdot A  =  \mu_X(S).\qedhere
\]
\end{proof}

\section{Adding termination}\label{sec:termination}

So far, we have provided a presentation for the monad $C$ which combines probability and nondeterminism. In order to properly model NPLTS, we need a last ingredient: termination. As discussed in Section~\ref{sec:monad}, termination is given by the monad $\cdot+1$ which can always be safely combined with any monad. Following the discussion at the end of Section~\ref{sec:monad}, the theory $\PCS= (\Sigma_{NP} \cup \Sigma_T, E_{NP})$ presents the monad $C(\cdot +1)$ which is the monad of finitely generated non empty convex sets of \emph{sub}distributions.
We call this theory $\PCS$ since algebras for this theory are \emph{pointed convex semilattices}, namely convex semilattices with a pointed element denoted by $\star$.

Like for the monad $\nePow$, there exist more than one interesting way of combining $C$ with $\cdot+1$. Rather than pointed convex semilattices, one can consider \emph{convex semilattices with bottom}, namely algebras for the theory $\CSB= (\Sigma_{NP}\cup \Sigma_T, E_{NP} \cup \{(B)\})$ obtained by adding $(B)$ to $\PCS$. Otherwise, one can add the axiom $(T)$ and obtain the theory $\CST= (\Sigma_{NP}\cup \Sigma_T, E_{NP} \cup \{(T)\})$ of \emph{convex semilattices with top}. We denote by $\TCSB$ and $\TCST$ the corresponding monads.

As we will illustrate in Section~\ref{sec:GenDet}, particularly relevant for defining trace semantics is the free algebra $\mu \colon\T\T{\{\succes\}} \to \T\{\succes\}$ generated by  a singleton $\{\succes\}$.
In the next three propositions we respectively identify these algebras for the monad $\TPCS$ (that is, $C(\cdot +1)$), the monad $\TCSB$, and the monad $\TCST$ in concrete terms.

\begin{prop}\label{prop:intervalspcs}
$\mathbb M_{\intervals,[0,0]} = (\intervals,\minmax,\psuminterval, [0,0])$ is the free pointed convex semilattice generated by a singleton set $ 1 = \{\succes\}$.
\end{prop}

\begin{proof}
Recall that $\mathbb{M}_{\intervals} = (\intervals,\minmax,\psuminterval)$ is the convex semilattice of  intervals from Section~\ref{sec:C}. Then, by interpreting the pointed element $\star$ as the interval $[0,0]$ we have that $\mathbb M_{\intervals,[0,0]}$ is a pointed convex semilattice.

Let $2=\{\succes, \star\}$.
Note that the carrier of the free pointed semilattice generated by $\{\bullet\}$ is $C(1 + 1) = C(2)$.
Recall that $(C(2), \cplus, \pplus{p})$, where $\cplus$ is the convex union and $\pplus{p}$ is the Minkowski sum, is the free convex semilattice generated by $2$.

We next show that $(C(2), \cplus, \pplus{p})$ is isomorphic to $\mathbb{M}_{\intervals}$. Indeed $\Dis(2)$ is isomorphic to $[0,1]$: the real number $0$ corresponds to $\delta_{\star}$, $1$ to $\delta_{\succes}$ and $p\in (0,1)$ to $\succes \pplus{p} \star$. Furthermore, the non-empty finitely-generated convex subsets of $[0,1]$ are the closed intervals.
To conclude, it suffices to see that $\minmax$ is $\cplus$ on $\intervals$ and $\psuminterval$ is the Minkowski sum.
\end{proof}

\begin{prop}\label{prop:max}
$\maxalg_B = ([0,1],\max,\pplus p,0)$ is the free convex semilattice with bottom generated by $ 1 = \{\succes\}$. \qed%
\end{prop}

\begin{proof}
By Proposition~\ref{prop:intervalspcs},  we know that $C(2)$ is isomorphic to $\intervals$.
We show that $\mathbb M_{\intervals,[0,0]}$ modulo the axiom (B) is isomorphic to $\maxalg_B$. We have
\[\minmax( [x,y] , [0,0]) \stackrel{(B)}{=} [x,y],\]
for $[x,y]\in \intervals$.
From \[[0,y]=\minmax( [x,y] , [0,0]) = [x,y],\]
we derive that $[x_1,y]=[x_{2},y]$ for any $x_{1},x_{2},y$. Hence, we define the isomorphism $[x,y]\mapsto y$ mapping any interval
$[x,y]$ to its upper bound $y$.

The interval $[0,0]$ is mapped to the bottom element $0$, and the operations are such that:
\begin{align*}
\minmax( [x_1,y_1] , [x_{2},y_{2}]) &=\minmax( [0,y_1] , [0,y_{2}])\\
&= [0, \max (y_{1},y_{2})]
\end{align*}
hence $\minmax( [x_1,y_1] , [x_{2},y_{2}]) \mapsto \max (y_{1},y_{2})$ and
\[[x_1,y_1] \pplus{p} [x_{2},y_{2}] =[0,y_1] \pplus{p} [0,y_{2}] = [0, y_{1} \pplus{p} y_{2}] \mapsto y_{1} \pplus{p} y_{2}.\qedhere \]
\end{proof}

\begin{prop}\label{prop:min}
$\minalg_T =([0,1],\min,\pplus p,0)$ is the free convex semilattice with top generated by $1 = \{\succes\}$. \qed%
\end{prop}

\begin{proof}
We show that $\mathbb M_{\intervals,[0,0]}$ modulo the (T) axiom
\[\minmax( [x,y] , [0,0]) \stackrel{(T)}{=} [0,0]\]
is isomorphic to $\minalg_{T}$.
First, we derive $[x,y_1]=[x,y_{2}]$ for any $x,y_{1},y_{2}$ as follows.
For $x=1$ the property trivially holds.
For $x=0$ we have
\begin{equation}\label{eq-star}
[0,y_1]=\minmax( [x,y_1] , [0,0])
 \stackrel{(T)}{=} [0,0]
\stackrel{(T)}{=}  \minmax( [x,y_2] , [0,0]) =  [0,y_{2}]
\end{equation}
Finally, for $x\in (0,1)$ and $y_1, y_2 \geq x$ we derive
\[[x,y_1]= [1,1] \pplus {x} [0,\frac {y_1-x}{1-x}]  \stackrel{(\ref{eq-star})}{=}
[1,1] \pplus {x} [0,\frac {y_2-x}{1-x}] =  [x,y_{2}].\]
Hence, we can now map every interval $[x,y]$ to its lower bound $x$.
Then $[0,0]$ is mapped to the top element $0$,
and
\[\minmax( [x_1,y_1] , [x_{2},y_{2}])
= [\min (x_{1},x_{2}), \max({y_{1}, y_{2}})]
\mapsto \min (x_{1},x_{2})
\]
\[[x_1,y_1] \pplus{p} [x_{2},y_{2}] =  [x_1 \pplus{p} x_{2},y_1 \pplus{p} y_{2}] \mapsto x_{1} \pplus{p} x_{2}. \qedhere \]
\end{proof}

At this point the reader may wonder what happens when one considers the axioms $(B_p)$ and $(T_p)$ in place of $(B)$ and $(T)$. We have already shown at the end of Section~\ref{sec:alg-th}, that the axiom $(B_p)$ makes the probabilistic structure collapse.
 When focussing on the free algebra generated by $\{\succes\}$, also quotienting  by $(T_p)$ is not really interesting: one can show by induction on the terms in $T_{\Sigma_{NP}\cup\Sigma_T}(\{\succes\})$ that every term is equal via $E_{NP}\cup \{(T_p)\}$ to either $\succes$ or $\star$ or $\succes \cplus \star$.

So, we have found three interesting ways of combining termination with probability and nondeterminism. Table~\ref{table:theories} summarises these theories, their monads, and their algebras.

\begin{table*}
\begin{center}
\begin{tabular}{ccccc}
Theory $(\Sigma, E) $& Monad $\T$ &  free algebra $\mu_1\colon \T\T1 \to\T1$ \\
\toprule
$\PCS= (\Sigma_{NP}\cup \Sigma_T, E_{NP})$ & $C(\cdot +1)=\TPCS$ & $\mathbb{M}_{\intervals,[0,0]} = (\intervals,\minmax,\psuminterval, [0,0])$ \\
$\CSB= (\Sigma_{NP}\cup \Sigma_T, E_{NP} \cup \{(B)\})$ & $\TCSB$ & $\maxalg_B = ([0,1],\max,+_p,0)$ \\
$\CST= (\Sigma_{NP}\cup \Sigma_T, E_{NP} \cup \{(T)\})$ & $\TCST$& $\minalg_T = ([0,1],\min,+_p,0)$  \\
\end{tabular}
\end{center}
\caption{The theories of pointed convex semilattices, with bottom, and with top.}\label{table:theories}
\end{table*}

\medskip

This completes our exploration of monads and algebras. In the next section, we will commence investigating coalgebras. But first, we show a useful result that illustrates the relationships among the monads encountered so far, where $\mathcal{SB}$ and $\mathcal{ST}$ are respectively the theory of semilattices with bottom and the theory of semilattices with top.

\begin{lem}\label{lemma:monadmaps}
There exist the following monad maps:
\[\xymatrix{
T_{\mathcal{SB}} \ar@{=>}[d]_{e^B} & \Powne(\cdot +1) \ar@{=>}[l]_{q^B} \ar@{=>}[r]^{q^T}  \ar@{=>}[d]|{\chi^{\Powne}{(\cdot +1)}} & T_{\mathcal{ST}} \ar@{=>}[d]^{e^T}\\
 T_{\mathcal{CSB}} & C(\cdot +1) \ar@{=>}[l]^{q^B} \ar@{=>}[r]_{q^T}  &T_{\mathcal{CST}} \\
 & \Dis(\cdot +1) \ar@{=>}[u]^{\chi^{\Dis}{(\cdot +1)}}
}\]
Moreover: 1.\ the vertical maps are injective; 2.\ the diagonal maps $q^B \circ \chi^{\Dis}{(\cdot +1)}$ and $q^T \circ \chi^{\Dis}{(\cdot +1)}$ are injective; 3.\ the two squares commute.
\end{lem}

\begin{proof}
We define $e^B\colon \TSB \Rightarrow \TCSB$ by $e_X([t]_\SB) = [t]_\CSB$ for any term $t$ with variables in $X$ in signature $\Sigma_N \cup \Sigma_T$, where $[t]_\SB$ on the left denotes the equivalence class of $t$ modulo $E_N \cup \{(B)\}$ and $[t]_\CSB$ on the right the equivalence class of $t$ modulo $E_{NP} \cup \{(B)\}$. This is justified as every $\TSB$-term is a $\TCSB$-term as well. This is easily seen to be a monad map, we need to check well-definedness and injectivity: $t =_{\SB} t' \Leftrightarrow t =_\CSB t'$. Well-definedness, the implication left-to-right, is immediate as the equations of a semilattice with bottom are included in the equations of a convex semilattice with bottom. Assume $t =_\CSB t'$. Let $\bar s$ denote the term obtained from a term $s$ in $\TCSB$ by replacing every occurrence of $\pplus{p}$ by $\cplus$. Then we have
\[s_1 =_{\CSB} s_2 \Rightarrow \bar s_1 =_\SB \bar s_2\]
which is easy to show by checking that it holds for each of the equations.

Now, let $t = t_1 =_{\CSB} t_2 \cdots =_\CSB t_n = t'$. Then
$t = \bar t_1 =_{\SB} \bar t_2 \cdots =_{\SB} \bar t_n =  t'$ where the first and last equality hold since $t$ and $t'$ are terms in $\Sigma_N \cup \{(B)\}$ showing injectivity.

\medskip

The definition and the proof for $e^T\colon \TST \Rightarrow \TCST$ are as above but replacing axiom $(B)$ by  $(T)$.
The maps $\chi^{\Powne}{(\cdot +1)}$ and $\chi^{\Dis}{(\cdot +1)}$ are injective monad maps due to Lemma~\ref{lem:monad-map-termination} applied to the injective monad maps $\chi^{\Powne}$ and $\chi^{\Dis}$ from Lemma~\ref{lemma:injmapC}.
The horizontal maps are obtained by quotienting by axioms $(B)$ and $(T)$ the monads $\TPS= \Powne(\cdot+1)$ and $\TPCS=C(\cdot +1)$.
Since, for every $\Delta \in \Dis(X+1)$, $\chi^{\Dis}{(\cdot +1)}_X(\Delta)$ is a singleton set, the quotients $q^B$ and $q^T$ do not affect such set: both $q^B \circ \chi^{\Dis}{(\cdot +1)}$ and $q^T \circ \chi^{\Dis}{(\cdot +1)}$ are injective.
Commutation of the two squares is immediate.
\end{proof}

\section{Coalgebras and Determinisation }\label{sec:GenDet}

In this section, we briefly introduce coalgebra (Section~\ref{subsec:coalgebra}) and the generalised determinisation~\cite{SBBR10} construction (Section~\ref{subsec:determinisation}), as well as several examples of transitions systems and automata featuring either nondeterministic or probabilistic behaviour. We present some simple facts and a novel general result (Section~\ref{subsec:theorem}) that will be useful in Section~\ref{sec:maymust} to prove some key properties for systems featuring --at the same time-- nondeterminism and probability. At the end of this section (Section~\ref{subsec:sys-aut}), we provide a general notion of trace semantics for transition systems and we illustrate some important examples.

\subsection{Coalgebra}\label{subsec:coalgebra}

The theory of coalgebra provides an abstract framework for state-based transition systems and automata. A \emph{coalgebra} for a functor $F$ in $\Sets$ (also called $F$-coalgebra) is a pair $(S,c)$ of a state space $S$ and a function $c\colon S \to FS$
  where $F\colon\Sets \rightarrow \Sets$ specifies the type of transitions. Sometimes we say the coalgebra $c\colon S \to FS$, meaning the coalgebra $(S,c)$.

A \emph{coalgebra homomorphism} from a coalgebra $(S,c)$ to a coalgebra $(T,d)$ is a function $h\colon S \to T$ that satisfies $d \circ h = Fh \circ c$. Coalgebras for a functor $F$ and their coalgebra homomorphisms form a category, denoted by $\CoAlg{(F)}$.

The final object in $\CoAlg{(F)}$, when it exists, is the \emph{final $F$-coalgebra}. We write $\smash{\zeta\colon Z \stackrel{\cong}{\longrightarrow}
FZ}$ for the final $F$-coalgebra. For every coalgebra $c\colon S \to FS$, there is a unique homomorphism $\llbracket \cdot\rrbracket_c$ to the final one, the \emph{final coalgebra map}, making the diagram below commute:
\[\xymatrix@R-1pc{
S \ar[d]_{c} \ar@{-->}[rr]^{\exists!\,\llbracket \cdot\rrbracket_c} && Z \ar[d]^{\zeta}_{\cong} \\
FS  \ar@{-->}[rr]^{F\llbracket\cdot \rrbracket_c} && FZ
}\]
The \emph{final coalgebra semantics} $\sim$ is the kernel of the final coalgebra map, i.e., two states $s$ and $t$ are equivalent in the final coalgebra semantics iff  $\llbracket s\rrbracket_c = \llbracket t \rrbracket_c$.

Even without a final coalgebra, coalgebras over a concrete category are equipped with a generic behavioural equivalence.
Let $(S,c)$ be an $F$-coalgebra on $\Sets$. An equivalence relation $R \subseteq S \times S$ is a kernel bisimulation (synonymously, a cocongruence)~\cite{Staton11,Kurz00:thesis,Wol00:cmcs} if it is the kernel of a homomorphism, i.e., $R = \ker h = \{(s,t) \in S \times S\mid h(s) = h(t)\}$ for some coalgebra homomorphism $h\colon (S,c) \to (T,d)$ to some $F$-coalgebra $(T,d)$. Two states $s,t$ of a coalgebra are \emph{behaviourally equivalent} (notation: $s \approx t$) iff there is a kernel bisimulation $R$ with $(s,t) \in R$. If a final coalgebra exists, then the behavioural equivalence and the final coalgebra semantics coincide, i.e., $\approx\,\, = \,\,\sim$.

The following are well-known examples of $F$-coalgebras that will be useful later on:
\begin{enumerate}
\item 	Labelled transition systems, LTS, are coalgebras for the functor $F = (\Pow(\cdot))^A$. Behavioural equivalence coincides with strong bisimilarity.
\item 	Nondeterministic automata, NA, are coalgebras for $F = 2 \times (\Pow(\cdot))^A$ where $2 = \{0,1\}$ is needed to differentiate whether a state is accepting or not.
\item Deterministic automata, DA, are coalgebras for $F = 2 \times (\cdot)^A$. The final coalgebra is carried by the set of all languages $2^{A^*}$.
\item Moore automata, MA, are a slight generalisation of deterministic automata with observations $O$: they are coalgebras for $F = O \times (\cdot)^A$. The final coalgebra is carried by the set of all $O$-valued languages $O^{A^*}$.
\item Reactive probabilistic labelled transition systems, RPLTS, are  coalgebras for $F=(\Dis(\cdot)+1)^A$. Behavioural equivalence coincides with Larsen-Skou bisimilarity~\cite{LS91:ic}.
\item (Rabin) Probabilistic automata~\cite{Rab63}, PA,  are coalgebras for $F=[0,1]\times \Dis(\cdot)^A$.
\end{enumerate}

The following definition generalises the examples above.

\begin{defi}[Systems and Automata with $\T$-effects]
Let $\T$ be a monad and $O$ be a set. We call an $\T^A$-coalgebra a \emph{system with $\T$-effects}, and we call an $O \times \T^A$-coalgebra  an \emph{automaton with $\T$-effects and observations in $O$}.
We write $c = \langle o,t \rangle$ for an automaton with $\T$-effects and observations in $O$, where $o\colon X \to O$ is the observation map assigning observations to states, and $t\colon X \to (\T X)^A$ is the transition structure.
\end{defi}

For instance, LTS are systems with $\Pow$-effects while NA are automata with $\Pow$-effects and observations in $2$. Similarly, RPLTS are systems with $\Dis +1$-effects, while PA are automata with $\Dis$-effects and observations in $[0,1]$. Both DA and MA are automata with no effects ($\T=\Id$) and observations in $2$ and $O$, respectively.

\medskip

We write $x \stackrel{a}{\to} m$ for $t(x)(a) = m$ with $a \in A, x \in X, m \in \T X$ in a system or automaton with $\T$-effects. We also write $x\downarrow {o_x}$ for $o(x)=o_x$ with $o_x\in O$. For an LTS $t\colon X\to (\Pow X)^A$ we also write, as usual,  $x \stackrel{a}{\to} y$ for $y \in t(x)(a)$ and $x \stackrel{a}{\not\to}$ if $t(x)(a) = \emptyset$; for an RPLTS $t\colon X\to (\Dis X + 1)^A$, we may also write $x \stackrel{a}{\to_p} y$ for $t(x)(a)(y) = p$ and again $x \stackrel{a}{\not\to}$ if $t(x)(a) = \star$.
Note that in all our examples of systems and automata there is an implicit finite branching property ensured by the use of $\Pow$ and $\Dis$ involving only finite subsets and finitely supported distributions.

\subsection{Determinising Automata with \texorpdfstring{$M$}{M}-effects and Observations in \texorpdfstring{$O$}{O}}\label{subsec:determinisation}

The construction of generalised determinisation was originally discovered in~\cite{SBBR10,Bartels04:thesis}. It enables us to obtain trace semantics for coalgebras of type $c\colon X \to F\T X$ where $F$ is a functor and $\T$ a monad. The result is a determinised $F$-coalgebra $c^\#\colon \T X \to F\T X$ and the semantics is derived from behavioural equivalence for $F$-coalgebras.

Let $c\colon X\to F\T X$ be a coalgebra and $\lambda\colon \T F \Rightarrow F\T$ a functor distributive law. Then the determinisation is the $F$-coalgebra
\begin{equation}\label{eq:det}
c^\sharp = \xymatrix{\T X \ar[r]^{\T c}& \T F\T X \ar[r]^{\lambda}& F \T\T X \ar[r]^{F \mu }& F\T X} .
\end{equation}
It is easy to show that $c^\sharp \after \eta = c$ which justifies the notation $c^\sharp$:
the carrier $\T X$ carries the $\T$-algebra $\mu_X\colon \T\T X \to \T X$, the free $\T$-algebra generated by $X$, $F\T X$ carries the $\T$-algebra $F\mu\after\lambda \colon \T F\T X \to F\T X$, and $c^\sharp$ is the unique extension of $c$ to a homomorphism from the free $\T$-algebra $(\T X, \mu)$ to the $\T$-algebra $(F\T X, F\mu \after \lambda)$.

In this paper, we only consider determinisation of automata with $\T$-effects and observations in $O$, namely, $\funct\T$-coalgebras for
the Moore-automata functor $\funct = O \times (\cdot)^{\lset}$, where $O$ is some set of observations.
The following proposition shows that determinising automata with $\T$-effects and observations in $O$ is always possible when the observations carry an $\T$-algebra~\cite{SBBR10,JacobsSS15}.

\begin{prop}\label{prop:distr-law}
	For an Eilenberg-Moore algebra $a\colon \T O \to O$, for $\funct = O \times (\cdot)^{\lset}$ and any monad $\T$ on $\Sets$ there is a canonical distributive law $ \lambda \colon \T\funct \Rightarrow \funct\T$ given by
	\!\!\[\!\!\xymatrix@R-2pc{
{\T\big(O\!\times \!X^{A}\big)}
      \ar[rr]^-{\!\langle\T \pi_{1}, \T\pi_{2}\rangle\!}
   & & {\T O\!\times\! \T(X^{A})}
  \ar[r]^-{a\times\strength} &
   {O\!\times\! (\T X)^{A}}
}\]
where $\strength$ is the map $\strength\colon \T(X^A) \to (\T X)^A$ defined, for all labels $a \in A$, by $\strength(\varphi)(a) = \T\ev_a(\varphi)$ with $\ev_a\colon X^A \to X$ the evaluation map given as $\ev_a(\varphi) = \varphi(a)$.  \qed%
\end{prop}

\noindent As a consequence, we can determinise $c = \langle o,t\rangle \colon X \to O \times (\T X)^A$ to $c^\sharp = \langle o^\sharp, t^\sharp\rangle$ where

\begin{eqnarray*}o^\sharp &=& \xymatrix{\T X \ar[r]^{\T o} & \T O \ar[r]^a & O }\text{ and }\\ t^\sharp &=& \xymatrix{\T X \ar[rr]^{\T t} && \T (\T X)^A \ar[r]^{\strength}& (\T\T X)^A \ar[rr]^{\mu_X^A}&& (\T X)^A}\text{.}\end{eqnarray*}

The final $F$-coalgebra, for $\funct = O \times (\cdot)^{\lset}$, is carried by the set $O^{A^*}$ of $O$-valued languages over alphabet $A$, i.e., functions $\varphi\colon A^* \to O$. The final coalgebra is \[\zeta = \langle \epsilon, \dder \rangle \colon O^{A^*} \to O \times (O^{A^*})^A\] where for all $\varphi \in O^{A^*}$, $\epsilon\colon O^{A^*} \to O$ is defined as  $\epsilon(\varphi) = \varphi(\varepsilon)$ and $\dder \colon O^{A^*} \to (O^{A^*})^A$ as $\dder(\varphi)(a)(w)=\varphi(aw)$.

 The morphism into the  final coalgebra $\bb{\cdot}_{c^\sharp} \colon \T X \to O^{A^*}$ is defined (see e.g.~\cite{Jacobs:book}) for all $m\in \T X$ and $w\in A^*$ inductively as below on the right.

\begin{equation}\label{Eq:det-diagram}
\begin{array}{ll}
\xymatrix{ X \ar[d]_{c = \langle o,t \rangle} \ar[r]^{\eta}& \T X \ar[dl]^{c^\sharp = \langle o^\sharp,t^\sharp \rangle} \ar@{-->}[rr]^{\bb{\cdot}_{c^\sharp}}& & O^{A^*} \ar[d]^{\zeta = \langle \epsilon, \dder \rangle}\\
O\times \T X^A \ar@{-->}[rrr]_{id_O \times \bb{\cdot}_{c^\sharp}^A}& & & O \times (O^{A^*})^A
} &
\begin{array}{lcl}
\vspace{0.5cm}\\
\bb{m}_{c^\sharp}(\varepsilon) &= &o^\sharp(m)\\
\bb{m}_{c^\sharp}(aw) &=& \bb{t^\sharp(m)(a)}_{c^\sharp}(w)
\end{array}
\end{array}
\end{equation}

\begin{defi}[Language semantics for automata]\label{def:languageequivalence}%
Let $(\T,\eta, \mu)$ be a monad and $c=\langle o, t\rangle \colon X \to O \times \T X^A$ be an automaton with $\T$-effects and observations in $O$. Let $a\colon \T O \to O$ be a $\T$-algebra. Two states $x,y\in X$ are \emph{language equivalent}, written $x \equiv y$ iff
\begin{equation*}
\bb{\eta(x)}_{c^\sharp}=\bb{\eta (y)}_{c^\sharp}
\end{equation*}
where $c^{\sharp}$ is the determinisation w.r.t. $a$ of $c=\langle o,t\rangle$.
\end{defi}

The commuting diagram on the left in Eq.~\ref{Eq:det-diagram} summarises our setting. An important observation is that $O^{A^*}$  carries an $\T$-algebra defined as the point-wise extension of $a \colon \T O \to O$. Similarly, $O \times (O^{A^*})^A$ also carries an $\T$-algebra and the final $F$-coalgebra $\zeta = \langle \epsilon, der \rangle$ is actually  an $\T$-algebra homomorphism. Recall that also $c^\sharp= \langle o^\sharp, t^\sharp \rangle \colon \T X \to O \times \T X^A$ is at the same time an $F$-coalgebra and an $\T$-algebra homomorphism: both $\langle \epsilon, der \rangle$ and $\langle o^\sharp, t^\sharp \rangle$ are indeed examples of \emph{$\lambda$-bialgebras}~\cite{DBLP:conf/lics/TuriP97,DBLP:journals/tcs/Klin11}. Most importantly, the unique coalgebra morphism $\bb{\cdot}_{c^\sharp} \colon \T X \to O^{A^*}$ is also an $M$-algebra homomorphism\footnote{From a more abstract perspective, the distributive law $\lambda$ allows for lifting the functor $F$ to $\EM(\T)$ as well as the rightmost commuting square in~\eqref{Eq:det-diagram}: all objects are $\T$-algebras and all arrows are $\T$-algebra homomorphism. We refer the interested reader to~\cite{DBLP:journals/tcs/Klin11} for a gentle introduction to the subject.}. The latter entails the first item of the following.

\begin{thmC}[\cite{SBBR10,DBLP:journals/acta/BonchiPPR17}]\label{thm:det-prop}
The following hold for any coalgebra $c \colon X {\to} F\T X$ and its determinisation $c^\sharp \colon\T X {\to} F\T X$:
\begin{enumerate}
	\item Behavioural equivalence for $(\T X, c^\sharp)$ is a congruence for the algebraic structure of $\T$.
	\item Behavioural equivalence for $(X,c)$ implies language equivalence.
	\item Up-to context is a compatible~\cite{SP09b} proof technique.
\end{enumerate}
\end{thmC}

\noindent
The second item will be exploited in Section~\ref{sec:maymust} to show that convex bisimilarity implies trace equivalence for NPLTS, while the third item will be fundamental in Section~\ref{sec:upto} to prove the soundness of an effective proof technique.  Below we give a more concrete description of the overall construction by relying on a presentation $(\Sigma,E)$ for the monad $\T$.

\medskip

For an $n$-ary operation symbol $f \in \Sigma$ and a $(\Sigma, E)$-algebra $\mathbb A = (A, \Sigma_A)$ we write $f^{A}$ for the $n$-ary operation on $A$ that is the interpretation of $f$. We have
that $c^\sharp = \langle o^\sharp, t^\sharp\rangle$ can be defined as follows.
\begin{equation}\label{eq:alg-det}
\begin{array}{ll}
 o^\sharp(x) = o(x) & o^\sharp(f^{\T X}(s_1, \dots, s_n)) =  f^O(o^{\sharp}(s_1),\dots, o^\sharp(s_n))\\
t^\sharp(x) = t(x) & t^\sharp(f^{\T X}(s_1, \dots, s_n))(a) =  f^{\T X}(t^\sharp(s_1)(a),\dots,t^\sharp(s_n)(a))
\end{array}
\end{equation}

\begin{rem}
Technically, the above definition is telling us that $o^\sharp \colon \T X \to O$ is the unique $\T$-algebra homomorphism extending the function $o\colon X \to O$ and similarly for $t^\sharp \colon \T X \to \T X^A$. Such a definition (of a function on $E$-equivalence classes of $\Sigma$-terms) can also be seen as arising from an inductively defined function on $\Sigma$-terms. To be precise, it is important to recall that defining a function over equivalence classes based on representatives always requires to check whether the function is well-defined, in the sense that it is independent of the choice of a representative. Fortunately enough, monad properties guarantee that  $o^\sharp \colon \T X \to O$ and $t^\sharp \colon \T X \to \T X^A$ are well-defined in this respect\footnote{It is enough to observe that (a) the function $\hat{o}^\sharp \colon T_\Sigma X \to O$ defined inductively on terms in $T_\Sigma X$ is the unique $\Sigma$-algebra homomorphism extending $o$;
(b) that there is a monad quotient $q^E \colon T_\Sigma \Rightarrow \T$ and (c) $q^E_X \colon T_\Sigma X \to \T X$ is a $\Sigma$-algebra homomorphism. By uniqueness, $\hat{o}^\sharp=o^\sharp \circ q^E_X $.} and thus~\eqref{eq:alg-det} can effectively be thought of as an inductive definition on terms. In the remainder of this paper, we will instantiate~\eqref{eq:alg-det} to several examples in this way, without mentioning again that $o^\sharp$ and $t^\sharp$ are well-defined.
\end{rem}

The $\T$-algebra structure over the final $F$-coalgebra is given for $\varphi_i\in O^{A^*}$ and $w\in A^*$ as
\begin{align*}
\!\!&f^{O^{A^*}}(\varphi_1, \dots, \varphi_n)(w) =  f^{O}(\varphi_1(w), \dots, \varphi_n(w)) \nonumber
\end{align*}
The fact that $\bb{\cdot}_{c^\sharp} \colon \T X \to O^{A^*}$ is an $\T$-algebra homomorphism just means that
\begin{align}\label{eq:homo}
\!\!&\bb{f^{\T X}(s_1, \dots, s_n)}_{c^\sharp} =  f^{O^{A^*}}(\bb{s_1}_{c^\sharp}, \dots \bb{s_n}_{c^\sharp})
\end{align}
for all $n$-ary operator $f$ in $\Sigma$. Clearly this fact immediately entails that $\approx$ is a congruence w.r.t.\ the operations in $\Sigma$.

\begin{exa}[Determinisation of Nondeterministic Automata]
Applying this construction to $F = 2\times (\cdot)^A$ and $\T = \Pow$, one transforms $c\colon X\to 2\times (\Pow X)^A$ into $c^\sharp \colon \Pow X \to 2\times (\Pow X)^A$. The former is a nondeterministic automaton and the latter is a deterministic automaton which has $\Pow X$ as states space. The set of all languages $2^{A^*}$ (seen as functions in $2 = \{0,1\}$) carries the final $F$-coalgebra $\langle \epsilon, \dder \rangle \colon 2^{A^*} \to 2 \times (2^{A^*})^A$ which is exactly the emptiness operation and the derivatives operations by Brzozowski~\cite{Brzozowski}. The final coalgebra morphism $\bb{\cdot}_{c^\sharp}\colon \Pow X \to 2^{A^*}$ maps each state of the determinised automaton into the language that it accepts.

In~\cite{SBBR10}, see also~\cite{JacobsSS15}, it is shown that, using the distributive law from Proposition~\ref{prop:distr-law}, as $2 = \Pow 1$ is the carrier of the free $\Pow$-algebra $\mu_1\colon \Pow\Pow 1 \to \Pow 1$, this amounts exactly to the standard determinisation from automata theory and justifies the term \emph{generalised determinisation}.

\medskip

Recalling that $\Pow$ is presented by the algebraic theory of semilattices with bottom gives us a more concrete understanding. The set of observations $2=\{0,1\}$ carries the semilattice with bottom $0\sqsubseteq 1$ with supremum operation denoted by $\sqcup$. The interpretation of $\cplus$ and $\star$ in this algebra are defined for all $b_1,b_2\in 2$ as
\begin{equation*}
b_1\cplus^2 b_2 = b_1 \sqcup b_2  \qquad \star^2=0\text{.}
\end{equation*}
By instantiating~\eqref{eq:alg-det}, one has that the determinisation $c^{\sharp}=\langle o^{\sharp} , t^{\sharp} \rangle  \colon \Pow X \to 2\times (\Pow X)^A$ is defined inductively for all $S\in \Pow X$ and $a\in A$ as
\[\begin{array}{lclcl}
o^{\sharp}(S) & \!\!= \!\!&
\begin{cases}
o(x) & \text{ if } S= \eta(x) \text{;}\\
               \star^2               & \text{ if } S = \star^{\Pow X}\text{;}\\
               o^{\sharp}(S_1) \cplus^2 o^\sharp(S_2) & \text{ if } S=S_1\cplus^{\Pow X} S_2  \text{;}\\
           \end{cases}\\
           & \!\!= \!\!&
\begin{cases}
o(x) & \text{ if } S= \{x\} \text{;}\\
               0               & \text{ if } S = \emptyset\text{;}\\
               o^{\sharp}(S_1) \sqcup o^\sharp(S_2) & \text{ if } S=S_1\cup S_2  \text{;}\\
           \end{cases}
\end{array}\]
\[ \begin{array}{lclcl}
t^{\sharp}(S)(a) & \!\!=\!\! &
\begin{cases}
t(x)(a) & \text{ if } S= \eta(x) \text{;}\\
               \star^{\Pow X}              & \text{ if } S= \star^{\Pow X}\text{;}\\
               t^{\sharp}(S_1) \cplus^{\Pow X} t^{\sharp}(S_2) & \text{ if } S=S_1\cplus^{\Pow X} S_2  \text{;}\\
           \end{cases} \\
           & \!\!= \!\!&
\begin{cases}
t(x)(a) & \text{ if } S= \{x\} \text{;}\\
               \emptyset              & \text{ if } S= \emptyset\text{;}\\
               t^{\sharp}(S_1) \cup t^{\sharp}(S_2) & \text{ if } S=S_1 \cup S_2  \text{;}\\
           \end{cases}
\end{array}
\]
since $\oplus^{\Pow X}$ is union of subsets and $\star^{\Pow X}$ is the empty set. We have seen an example of this construction in Section~\ref{sec:intro}.

The semilattice with bottom over the final $F$-coalgebra $2^{A^*}$ is defined as the pointwise extension of the semilattice over $2$, that is, for all $\varphi_1,\varphi_2\in 2^{A^*}$ and $w\in A^*$
\begin{equation*}
(\varphi_1 \cplus^{2^{A^*}}\varphi_2)(w) = \varphi_1(w)\sqcup\varphi_2(w)  \qquad \star^{2^{A^*}}(w)=0
\end{equation*}
Observe that $\cplus^{2^{A^*}}$ is just the union of languages and $\star^{2^{A^*}}$ is the empty language. By instantiating~\eqref{eq:homo}, one has that
\begin{equation*}
\bb{S_1 \cplus^{\Pow X}S_2}_{c^\sharp}= \bb{S_1}_{c^\sharp}\cplus^{2^{A^*}} \bb{S_2}_{c^\sharp} \qquad \bb{\star^{\Pow X}}_{c^\sharp}= \star^{2^{A^*}}
\end{equation*}
meaning that (a) the language accepted by the union of two sets of states is exactly the union of the languages accepted by the two sets and (b) the empty set accepts the empty language. This immediately entails the first item of Theorem~\ref{thm:det-prop}. The second item, intuitively corresponds to the usual fact that bisimilarity implies language equivalence. The third item allows for exploiting bisimulations up-to $\cplus$ and, consequently, for an efficient algorithm to check language equivalence~\cite{BonchiP13}.
\end{exa}

\begin{exa}[Determinisation of Probabilistic Automata]\label{ex:PA}
Applying the construction for $F = [0,1]\times (\cdot)^A$ and $\T = \Dis$, one transforms $c\colon X\to [0,1]\times (\Dis X)^A$ into $c^\sharp \colon \Dis X \to [0,1]\times (\Dis X)^A$. The former is a probabilistic automaton and the latter is a Moore automaton with set of observations $O=[0,1]$.  The set of all $[0,1]$-valued languages $[0,1]^{A^*}$ carries the final $F$-coalgebra.
Since $[0,1]=\Dis(2)$, the set $[0,1]$ carries the $\Dis$-algebra $\mu_2\colon \Dis\Dis(2) \to \Dis(2)$ which amounts to taking convex combinations in $[0,1]$. By exploiting this algebra for the distributive law in Proposition~\ref{prop:distr-law}, one obtains a final coalgebra semantics that coincides with probabilistic language equivalence of~\cite{Rab63} (see~\cite{SBBR10}).

Next we illustrate more concretely this construction by relying on the presentation of $\Dis$ as the algebraic theory of convex algebras. For $p\in [0,1]$, the interpretation of the operation $\pplus{p}$ in $[0,1]$ is as expected: $q_1 \pplus{p}^{[0,1]} q_2 = p\cdot q_1 + (1-p)\cdot q_2$.
By instantiating~\eqref{eq:alg-det}, one has that the determinisation $c^\sharp=\langle o^{\sharp} , t^{\sharp} \rangle  \colon \Dis X \to [0,1]\times (\Dis X)^A$ is defined inductively for all $\Delta \in \Dis X$ and $a\in A$ as
\[\begin{array}{lcl}
o^{\sharp}(\Delta) & = &
\begin{cases}
o(x) & \text{ if } \Delta= \eta(x) = \delta_x \text{;}\\
               o^{\sharp}(\Delta_1) \pplus{p}^{[0,1]} o^\sharp(\Delta_2) & \text{ if } \Delta=\Delta_1\pplus{p}^{\Dis X} \Delta_2  \text{;}\\
           \end{cases}
\end{array}\]
\[ \begin{array}{lcl}
t^{\sharp}(\Delta)(a) & = &
\begin{cases}
t(x)(a) & \text{ if } \Delta= \eta(x) = \delta_x \text{;}\\
               t^{\sharp}(\Delta_1) \pplus{p}^{\Dis X} t^{\sharp}(\Delta_2) & \text{ if } \Delta=\Delta_1\pplus{p}^{\Dis X} \Delta_2  \text{;}\\
           \end{cases}
\end{array}
\]
We have seen an example of this construction in Section~\ref{sec:intro}.
The convex algebra over the final coalgebra $[0,1]^{A^*}$ is defined as the pointwise extension of the algebra over $[0,1]$, that is \begin{equation*}(\varphi_1 \pplus{p}^{[0,1]^{A^*}} \varphi_2) (w)= \varphi_1(w) \pplus{p}^{[0,1]}\varphi_2(w)\end{equation*} for $\varphi_1, \varphi_2 \in [0,1]^{A^*}$ and $w\in A^*$. By instantiating~\eqref{eq:homo}, one has that
\begin{equation*}
\bb{\Delta_1 \pplus{p}^{\Dis X}\Delta_2}_{c^\sharp}= \bb{\Delta_1}_{c^\sharp}\pplus{p}^{[0,1]^{A^*}} \bb{\Delta_2}_{c^\sharp}\text{.}
\end{equation*}
\end{exa}

We will apply similar methods to obtain languages and language equivalence for automata with both nondeterminism and probability in the rest of the paper.

\subsection{Invariance of the Semantics}\label{subsec:theorem}

We next state a theorem that guarantees invariance of the language semantics  for automata with $\T$-effects and observations in $O$, under controlled changes of the monad or the algebra of observations.

\begin{thm}[Invariance Theorem]\label{thm:transfert} Let $(\T,\eta, \mu)$ be a monad and $a\colon \T O \to O$ an $\T$-algebra. Let $c = \langle o,t\rangle \colon X \to O \times (\T X)^A$ be an automaton with $\T$-effects and observations in $O$ and $\llbracket \cdot \rrbracket \colon \T X \to O^{A^*}$ be the semantic map induced by the generalised determinisation w.r.t. $a$, i.e., $\llbracket \cdot \rrbracket = \llbracket \cdot \rrbracket_{c^\sharp}$
\begin{enumerate}
\setlength{\itemsep}{1em}
\item {\bf Transitions:}  Let $(\hat\T,\hat\eta,\hat\mu)$ be a monad and $\sigma\colon \T \Rightarrow \hat\T$ a monad map. Let $\hat a\colon \hat\T O \to O$ be an $\hat\T$-algebra. Consider the coalgebra \[\hat c = \langle o, \hat t \,\rangle = \langle o, \sigma_X^A \after t \rangle \colon X \to O \times (\hat\T X)^A\] and let $\hat{\bb{\cdot}}\colon \hat\T X \to O^{A^*}$ be the semantic map induced by its generalised determinisation wrt. $\hat a$. If $a = \hat a \after \sigma_O$, then $\bb{\cdot} \after \eta_X= \hat{\bb{\cdot}} \after \hat\eta_X$.
\item {\bf Observations:} Let $\hat a \colon \T \hat O \to \hat O$ be an $\T$-algebra and let $h\colon (O,a) \to (\hat O,\hat a)$ be an $\T$-algebra morphism. Consider the coalgebra \[\hat c = \langle \hat o, t \rangle = \langle h \after o,  t \rangle \colon X \to \hat O \times (\T X)^A\] and let $\hat{\bb{\cdot}}\colon \T X \to \hat O^{A^*}$ be induced by the generalised determinisation wrt. $\hat a$. Then $\hat{\bb{\cdot}}=h^{A^*} \after \bb{\cdot}$. \qed%
\end{enumerate}
\end{thm}

\begin{proof}
We prove the two items separately:
\begin{enumerate}
\setlength{\itemsep}{1em}
\item {\bf Transitions:}
The proof proceeds in two steps. First, we show that the following diagram commutes
\begin{equation}\label{eq:distributivelawmorphism}
\xymatrix{
\T FX \ar[r]^{\sigma_{FX}} \ar[d]_{\lambda_X} & \hat\T FX \ar[d]^{\hat\lambda_X}\\
F\T X \ar[r]_{F\sigma_X} & F\hat\T X\\
}
\end{equation}
where $\sigma$ is the monad map from the hypothesis, and $\lambda$ and $\hat\lambda$ are the distributive laws from Proposition~\ref{prop:distr-law} used for the determinisation of $F\T$- and $F\hat\T$-coalgebras using the algebras $a$ and $\hat a$, and the strengths $\strength$ and $\hat\strength$, respectively.

The following diagram commutes by naturality of $\sigma$:
\[\xymatrix@C+2pc@R-0.5pc{
\T(X^A) \ar[r]^{\sigma_{X^A}} \ar[d]_{\T\ev_a} & \hat\T(X^A) \ar[d]^{\hat\T\ev_a} \\
\T X \ar[r]_{\sigma_X} & \hat\T X
}\]
Using this, by definition of the strength, see Proposition~\ref{prop:distr-law}, we have
\begin{eqnarray*}
(\sigma_X^A \after\strength(\varphi)) (a) & = & \sigma_X^A(\strength(\varphi)(a))\\
& = & \sigma_X^A( \T\ev_a(\varphi))\\
& = & (\sigma_X^A \after \T\ev_a)(\varphi)\\
& \stackrel{(*)}{=} & \hat\T\ev_a \after \sigma_{X^A}(\varphi)\\
& = & \hat\T\ev_a(\sigma_{X^A}(\varphi))\\
& = & \hat\strength(\sigma_{X^A}(\varphi))(a)\\
& = & (\hat\strength\after\sigma_{X^A}(\varphi))(a).
\end{eqnarray*}
where the equality marked by $(*)$ holds by the commutativity of the diagram above. Hence, the following diagram commutes.
\[\xymatrix{
\T(X^A) \ar[r]^{\sigma_{X^A}} \ar[d]_{\strength} & \hat\T(X^A) \ar[d]^{\hat\strength} \\
(\T X)^A \ar[r]_{\sigma_X^A} & (\hat\T X)^A
}\]

Recall now that by hypothesis $a = \hat a \after \sigma_O$. Therefore, the following commutes.
\[\xymatrix@C=2cm{
\T O \times \T(X^A) \ar[r]^{\sigma_O \times \sigma_{X^A}} \ar[d]_{a\times \strength} & \hat\T O \times \hat \T(X^A) \ar[d]^{\hat a \times \hat\strength} \\
O \times (\T X)^A \ar[r]_{id_O \times \sigma_X^A} & O \times (\hat\T X)^A
}\]

Finally, the following two squares commute by naturality of $\sigma$.
\[\xymatrix{
\T(O \times X^A) \ar[r]^{\sigma_{O\times X^A}} \ar[d]_{\T\pi_1} & \hat\T (O\times X^A) \ar[d]^{\hat\T\pi_1} \\
\T O \ar[r]_{\sigma_O} & \hat\T O
}
\qquad \xymatrix{
\T(O \times X^A) \ar[r]^{\sigma_{O\times X^A}} \ar[d]_{\T\pi_2} & \hat\T(O\times X^A) \ar[d]^{\hat\T\pi_2} \\
\T(X^A) \ar[r]_{\sigma_{X^A}} & \hat\T(X^A)
}\]
By pasting together the last three diagrams, we obtain that the following commutes.
\[\xymatrix@C=2cm{
\T(O \times X^A) \ar[r]^{\sigma_{O\times X^A}} \ar[d]_{\langle\T\pi_1, \T\pi_2\rangle} & \hat\T(O\times X^A) \ar[d]^{\langle\hat\T\pi_1, \hat\T\pi_2\rangle}\\
\T O \times \T(X^A) \ar[r]^{\sigma_O \times \sigma_{X^A}} \ar[d]_{a\times \strength} & \hat\T O \times \hat\T(X^A) \ar[d]^{\hat a \times \hat\strength} \\
O \times \T(X)^A \ar[r]_{id_O \times \sigma_X^A} & O\times \hat\T(X)^A
}\]
Observe that, by the definition of the distributive law (Proposition~\ref{prop:distr-law}), this diagram is exactly~\eqref{eq:distributivelawmorphism}.
Using~\eqref{eq:distributivelawmorphism}, we can now easily show that the following commutes.
\[
\xymatrix{
\T X \ar[d]_{\T\langle o,t\rangle} \ar[r]^{\sigma_X}& \hat\T X \ar[d]^{\hat\T\langle o,t\rangle}\\
\T F\T X \ar[d]_{\lambda_{\T X}} \ar[r]^{\sigma_{F\T X}}& \hat\T F\T X\ar[d]_{\hat\lambda_{\T X}} \ar[rd]^{\hat\T F\sigma_X} \\
F\T\T X \ar[dd]_{F\mu} \ar[r]^{F\sigma_{\T X}} & F\hat\T\T X  \ar[d]_{F\hat\T \sigma_X } & \hat\T F\hat\T X \ar[ld]^{\hat\lambda_{\hat \T X}}\\
& F\hat\T\hat\T X \ar[d]^{F\hat\mu_X} \\
F\T X \ar[r]_{F\sigma_X} & F\hat\T X
}
\]
Indeed, commutativity of the topmost square is given by naturality of $\sigma$. The fact that $\sigma$ is a monad morphism entails commutativity of the bottom square. The rightmost square commutes by naturality of $\hat\lambda$. The missing square, the one in the centre, is exactly~\eqref{eq:distributivelawmorphism}.

Now observe that the leftmost border in the above diagram, the morphism $\T X \to F\T X$, equals $c^\sharp = \langle o^{\sharp}, t^\sharp \rangle$, see~\eqref{eq:det}. The determinisation $\hat c^\sharp$ of $\hat c = \langle o, \hat t\,\rangle = \langle o, (\sigma_X)^A \after t \rangle$ obtained using $\hat a$ and $\hat\lambda$ coincides with the rightmost border of the above diagram, the morphism $\hat \T X \to F\hat\T X$. The commuting of the above diagram means that $\sigma_X$ is a homomorphism of $F$-coalgebras. By postcomposing this homomorphism with the unique $F$-coalgebra morphism $\hat{\bb{\cdot}}\colon \hat\T X \to O^{A^*}$, one obtains an $F$-coalgebra morphism of type $\T X \to O^{A^*}$. Since $\bb{\cdot}$ is the unique such morphism, $\bb{\cdot}=\hat{\bb{\cdot}}\after \sigma_X$ follows.

\[
\xymatrix@C=2pc{
\T X\ar@(ur,ul)[rr]^{\bb{\cdot}}\ar[d]_{c^\sharp} \ar[r]^{\sigma_X}& \hat\T X  \ar[r]^{\hat{\bb{\cdot}}} \ar[d]^{ \hat c^\sharp}& O^{A^*}\ar[d]^{\zeta}_\cong\\
F\T X \ar@(dr,dl)[rr]_{F\bb{\cdot}} \ar[r]^{F\sigma_X} & F\hat\T X \ar[r]^{F\hat{\bb{\cdot}}} & F(O^{A^*})
}\]
Now, since $\sigma$ is a monad map, $\hat\eta = \sigma\after \eta$. Therefore $\bb{\cdot} \after \eta_X = \hat{\bb{\cdot}}\after \sigma_X \after \eta_X = \hat{\bb{\cdot}}\after \hat\eta_X$.

\item {\bf Observations:}
Consider the following diagram in $\Sets$.

	\[\xymatrix@C=2cm{\T X \ar[r]^{\bb{\cdot}} \ar[d]_{c^\sharp }& O^{A^*} \ar[d]^{\zeta}_\cong \ar[r]^{h^{A^*}}& \hat O^{A^*} \ar[dd]^{\hat\zeta}_\cong  \\
O \times (\T X)^A \ar[d]_{h\times id_{(\T X)^A}} \ar[r]^{id_O\times \bb{\cdot}^A}&O \times (O^{A^*})^A \ar[d]^{h\times id_{(O^{A^*})^A}} & \\
\hat O \times (\T X)^A  \ar[r]_{id_{\hat O}\times\bb{\cdot}^A} & \hat O \times (O^{A^*})^A \ar[r]_{id_{\hat O}\times (h^{A^*})^A}& \hat O \times (\hat O^{A^*})^A }\]

Both squares on the left trivially commute by definition. To prove that also the square on the right commutes, it is enough to show that $h^{A^*}\colon O^{A^*} \to \hat O^{A^*}$ coincides with the unique coalgebra morphism $\bb{\cdot}_d$ from the coalgebra \[d = (h \times id_{(O^{A^*})^A})\circ \zeta \colon O^{A^*} \to {\hat O}\times (O^{A^*})^A\]
to the final $\hat O\times (\cdot)^A$-coalgebra $\hat\zeta$.

From the inductive definition of $\bb{\cdot}_d$, see Eq.~\ref{Eq:det-diagram}, we get $\bb{\varphi}_d = \lambda w \in A^*.\, h\after\epsilon((\varphi)_w)$ where $(\varphi)_w(u) = \varphi(wu)$ which easily leads to $\bb{\varphi}_d = h \after \varphi = h^{A^*}(\varphi)$.

\medskip

Now observe that $h \after o^{\sharp}$ is equal to $(h\after o)^\sharp$, since $h$ is an algebra morphism.

From this observation and the commuting of the above diagram, it follows that $h^{A^*} \after \bb{\cdot}$ is the unique coalgebra morphism from $\hat c^\sharp = \langle (h\after o)^\sharp, t^{\sharp} \rangle$ to the final $\hat O\times (\cdot)^A$-coalgebra, and hence it equals $\hat{\bb{\cdot}}$. \qedhere
\end{enumerate}
\end{proof}

\noindent
Whenever the monad $M$ has a presentation $(\Sigma,E)$, the first part of the theorem above provides a convenient technique that allows one to work with syntactic terms in $T_{\Sigma}$ and forget about the axioms in $E$. Let $q^E\colon T_{\Sigma}\Rightarrow M$ be the monad morphism quotienting  $T_{\Sigma}$ by the axioms in $E$ and let $r_X\colon MX \to T_{\Sigma}X$ be a right inverse of $q^E_X$, namely $q^E_X  \circ r_X  =id_{MX}$. Now, the automaton with $\T$-effects $c = \langle o,t\rangle \colon X \to O \times (\T X)^A$ can be translated into an automaton with $T_\Sigma$-effects by taking
\[ c' = \langle o,  t' \,\rangle = \langle o, r_X^A \after t \rangle \colon X \to O \times (T_{\Sigma} X)^A\text{.}\]
Similarly the $\T$-algebra of observation $a\colon \T O\to O$ gives rise to the $T_{\Sigma}$-algebra
\[ a' = a \after q^E_O \colon T_{\Sigma}O \to O\text{.}\]
Now, rather than determinise $c$ with respect to $a$, one can determinise $c'$ w.r.t. $a'$ and, by virtue of Theorem~\ref{thm:transfert}.1, the semantics does not change.

\begin{cor}\label{cor:invarianceGSOS}
Let $\T$ be a monad with presentation $(\Sigma, E)$.
Let  $c = \langle o,t\rangle \colon X \to O \times (\T X)^A$  and  $c' = \langle o,t'\rangle \colon X \to O \times (T_{\Sigma} X)^A$ be such that $q^E_X\circ t'= t$. Let $a\colon \T O\to O$ be an $\T$-algebra.
Let $\bb{\cdot}$ and $\bb{\cdot}'$ be the semantics maps induced by the determinisation of $c$ w.r.t. $a$ and, respectively, $c'$ w.r.t. $a'= a \after q^E_O$. Let $\eta$ and $\eta'$ be the units of $\T$ and $T_{\Sigma}$.
Then, $\bb{\cdot} \after \eta_X= \bb{\cdot} \after \eta_X'$.
\end{cor}
\begin{proof}
The property follows immediately by Theorem~\ref{thm:transfert}.1, but the reader has to be careful not to be mislead: Theorem~\ref{thm:transfert}.1 cannot be applied to $r_X$, since this is not a monad map (actually, it is not even a natural transformation). The monad map $\sigma \colon \T \Rightarrow \hat \T$ in the statement of the theorem is the quotient $q^E: T_\Sigma \Rightarrow \T$ in the corollary, indeed the monads $\T$ and $\hat \T$ in the theorem correspond to the monads $T_\Sigma$ and $\T$ in the corollary, respectively.
Similarly, the automaton with $\T$-effects $c$ in the theorem is $c'$ in the corollary, while $\hat c$ is $c$ in the corollary. The algebras $\hat a$ and $a$ in the theorem correspond to $a$ and $a'$ in the corollary, respectively.
\end{proof}
Determinising $c'$ rather than $c$, makes our approach similar to processes calculi. Indeed,
in $c'^\sharp=\langle o'^\sharp, t'^\sharp \rangle \colon T_\Sigma X \to O\times T_\Sigma X^A$, states are syntactic terms rather than elements in $\T X$, e.g., sets, distributions, \dots. The output and transition functions $o'^\sharp, t'^\sharp$ can be defined
by means of GSOS rules. For each $n$-ary operator $f$ in $\Sigma$, we have one rule for output and one rule  per action $a\in A$ for transitions:
\begin{equation}\label{generalGSOS1} \infer{s_i\downarrow o_i \quad  i\in 1\dots n}{f(s_1, \dots, s_n)\downarrow f^O (o_1, \dots, o_n)}  \qquad \infer{s_i \stackrel{a}{\to}s_i' \quad i\in 1\dots n}{f(s_1,\dots, s_n )\stackrel{a}{\to}f(s_1', \dots, s_n')}\end{equation}
Here $f^O$ is the interpretation of $f$ in the algebra of observation $a\colon \T O \to O$. For each $x\in X$, we have the following axioms:
\begin{equation}\label{generalGSOS2} \infer{-}{x\downarrow o(x)}  \qquad \infer{-}{x\stackrel{a}{\to}t'(x)(a)}\end{equation}
\begin{rem}\label{rem:GSOS}
The rules in~\eqref{generalGSOS1} and~\eqref{generalGSOS2} are equivalent --- modulo the axioms in $E$ --- to the definition of $o^\sharp$ and $t^\sharp$ provided by~\eqref{eq:alg-det}\footnote{The reader can easily check by induction that $o^\sharp \after q^E_X = o'^{\sharp}$ and that $t^{\sharp}\after q^E_X = (q^E_X)^A \after t'^\sharp$.}. However, the two rules above provide a specification which is closer to (the standard way of giving) the structured operational semantics via so-called GSOS rules where states are syntactic terms and transitions are derived by the structure of terms. A notable difference, beyond the outputs, is that the transitions specified by~\eqref{generalGSOS1} and~\eqref{generalGSOS2} are deterministic: $s\tr{a}s'$ means that $t'^{\sharp}(s)(a)=s'$ and not that $s' \in t'(s)(a)$ as common in structured operational semantics.

The connection between GSOS rules and bialgebras was first identified in the seminal work by Turi and Plotkin~\cite{DBLP:conf/lics/TuriP97}: they show that GSOS rules are in one-to-one correspondence with certain natural transformations that, in turn, give rise to distributive laws of type $\lambda \colon \T F \Rightarrow F\T$, see Section~\ref{subsec:distr-laws}, for $F=\Pow(\cdot)^A$ and $\T$ the term monad $T_\Sigma$ for some signature $\Sigma$. The rules in~\eqref{generalGSOS1} instead can be regarded as a natural transformation $\rho\colon \Sigma F \Rightarrow F\Sigma$, for $\Sigma$ the functor corresponding to a signature and $F=O\times (\cdot)^A$. In~\cite{DBLP:journals/tcs/Klin11}, it is shown that natural transformations of this kind, that are named there ``simple distributive laws'', give rise to a distributive law of type $\lambda \colon T_{\Sigma} F \Rightarrow F T_{\Sigma}$. In our case, the rules in~\eqref{generalGSOS1} give rise exactly  to the distributive law from Proposition~\ref{prop:distr-law} for $\T=T_\Sigma$.
\end{rem}

\subsection{From Systems to Automata and the Role of Termination}\label{subsec:sys-aut}
So far, we have seen how language semantics arises by determinisation of automata. However, our initial interest concerns trace semantics. Trace semantics is closely related to language semantics, but it concerns systems that lack the notion of observation. At the same time, dealing with automata, i.e., having observations, is crucial for determinisation. In this section, we discuss the move from systems to automata that enables defining trace semantics with help of language semantics. The intuition can easily be explained in the case of LTS\@: Create an automaton from an LTS by making every state accepting. Then trace equivalence for the LTS is language equivalence for the created automaton.

More precisely, starting from an LTS $t \colon X \to (\Pow X)^A$, we can add observations in $2 = \Pow 1$ in the simplest possible way, making every state an accepting state: \begin{equation*}o = ( X \stackrel{!}{\longrightarrow} 1 \stackrel{\eta_1}{\longrightarrow} \Pow 1 = 2)\end{equation*} and determinise the NA $\langle o,t\rangle \colon X \to 2 \times (\Pow X)^A$ w.r.t.\ the algebra $\mu_1 \colon \Pow\Pow 1 \to \Pow1$. The induced language semantics  on the state space $X$ is the standard trace semantics for LTS~\cite{DBLP:journals/mscs/BonchiBCR016}.

 This same approach can be applied in the case of any system with $\T$-effects $t\colon X \to (\T X)^A$. We can add observations in $O = \T 1$ by
\begin{equation}\label{eq:observation} o = ( X \stackrel{!}{\longrightarrow} 1 \stackrel{\eta_1}{\longrightarrow} \T 1),\end{equation} determinise the automaton $\langle o, t \rangle$ with $\T$-effects using the free algebra on $\T 1$, and consider the induced language semantics.
This idea is summarised in the following definition.

\begin{defi}[Trace semantics for systems with $\T$-effects]\label{def:traceequivalence}
Let $(\T,\eta, \mu)$ be a monad and $t\colon X \to \T X^A$ be a system with $\T$-effects. Two states $x,y\in X$ are trace equivalent, written $x \equiv y$ iff
\begin{equation*}
\bb{\eta(x)}_{c^\sharp}=\bb{\eta (y)}_{c^\sharp}
\end{equation*}
where $c^{\sharp}$ is the determinisation w.r.t. $\mu_1 \colon \T \T 1 \to \T1$ of $c=\langle o,t\rangle \colon \T X \to \T 1 \times \T X^A $ for $o$ defined as in~\eqref{eq:observation}.
\end{defi}

However, not everything is settled yet, as shown by the following example --- the way we model systems, i.e., the choice of the monad involved, may make a huge difference in the semantics that we obtain.

\begin{exa}\label{ex:rpltstermination}
Let us apply Definition~\ref{def:traceequivalence} to an RPLTS $t\colon X \to (\dset X+1)^A$, namely a system with $\dset+1$-effects. Recall the monad $\Dis + 1$ from Example~\ref{Example-D+1}. By taking as sets of observations $O=\T 1$, one would obtain that $O=\dset(1)+1$ which is isomorphic to $2$. Therefore, the final coalgebra would be $2^{A^*}$ making impossible to distinguish probabilistic behaviours.

The problem is not in the choice of observations but already in the monad $\dset+1$. Indeed, even by allowing different observations, determinising w.r.t.\ such monad would lead a too coarse semantics.
Consider for instance the RPLTS in the top of Figure~\ref{fig:RPLTS}. The states $x$ and $y$ should not be trace equivalent, since $x$ has probability $\frac 1 2$ of performing trace $ab$, and $y$ has probability $\frac 1 4$ of performing trace $ab$.
Let us look at what happens, however, if we determinise this system  with respect to the monad $\dset +1$.
The determinised transition function $t^\sharp$
will give us states in $\dset X +1$, i.e., states that are either full distributions or the element $\star\in 1$ and we have
\[t^\sharp (x)(a) = x_{1} \pplus {\frac 1 2} x_{2} \qquad t^\sharp (y)(a) =  y_{1} \pplus {\frac 1 4} y_{2} \]

\begin{align*}
\text{However, } &&t^\sharp (x_{1} \pplus {\frac 1 2} x_{2})(b) = t(x_{1})(b) \pplus {\frac 1 2} t(x_{2})(b) = \star \\
 &&t^\sharp ( y_{1} \pplus {\frac 1 4} y_{2})(b) = t(y_{1})(b) \pplus {\frac 1 4} t(y_{2})(b) = \star
\end{align*}
Hence, whatever $(\dset +1)$-algebra of observation we take, these states in the determinised system will return the same observation, i.e., $o^{\sharp}(x)(ab)= o^{\sharp}(y)(ab)$.
As a consequence, $x$ and $y$ will be equivalent.
\end{exa}
\begin{figure}
\begin{center}
\begin{tabular}{c}
\begin{tikzpicture}[thick]
\matrix[matrix of nodes, row sep= 0.6cm, column sep=.1cm,ampersand replacement=\&]
{
					\&\node (x) {$x$};		\\
					\&\node (d1) {$\Delta$};	\\
\node (x1) {$x_{1}$};		\&						\&\node (x2) {$x_{2}$};		\\
	};
\draw[-latex] (x) to node[left] {$a$} (d1);
\draw[-latex] (x1) to[bend left=70] node[left] {$b$} (x);
\draw[dotted,->] (d1) to node[left] {$\frac 1 2$}  (x1);
\draw[dotted,->] (d1) to node[right] {$\frac 1 2$}  (x2);

\begin{scope}[xshift=4cm]
\matrix[matrix of nodes, row sep= 0.6cm, column sep=.1cm,ampersand replacement=\&]
{
					\&\node (x) {$y$};		\\
					\&\node (d1) {$\Theta$};	\\
\node (x1) {$y_{1}$};		\&						\&\node (x2) {$y_{2}$};		\\
	};
\draw[-latex] (x) to node[left] {$a$} (d1);
\draw[-latex] (x1) to[bend left=70] node[left] {$b$} (x);
\draw[dotted,->] (d1) to node[left] {$\frac 1 4$}  (x1);
\draw[dotted,->] (d1) to node[right] {$\frac 3 4$}  (x2);
\end{scope}
\end{tikzpicture}

\vspace{3mm}\\
\xymatrix{x\downarrow_1 \ar[r]^a \ar[rd]^b& x_{1} \pplus {\frac 1 2} x_{2}\downarrow_1 \ar[d]^a\ar[r]^b& x \pplus {\frac 1 2} \star \downarrow_{\frac 1 2}\ar[r]^{a} \ar[dl]_b& \dots\\ % chktex 1
& \star\downarrow_0 \ar@(ur,dr)^{a,b}\\
y\downarrow_1 \ar[r]^a \ar[ru]^b& y_{1} \pplus {\frac 1 4} y_{2}\downarrow_1 \ar[u]_a \ar[r]^b& y \pplus {\frac 1 4} \star\downarrow_{\frac 1 4} \ar[r]^{a} \ar[ul]_b& \dots} % chktex 1

\vspace{5mm}\\
$\begin{array}{ccccc}
\bb{\eta (\cdot)}_{c^{\sharp}} & \varepsilon & a& b & ab\\
\toprule
x & 1 & 1 &  0 & \frac{1}{2} \\[0.4em]
y & 1 & 1 & 0 & \frac{1}{4} \\[0.4em]
\end{array}$
\end{tabular}
\end{center}
\caption{An RPLTS (top), part of its determinisation (center) and of its trace semantics (bottom)}\label{fig:RPLTS}
\end{figure}
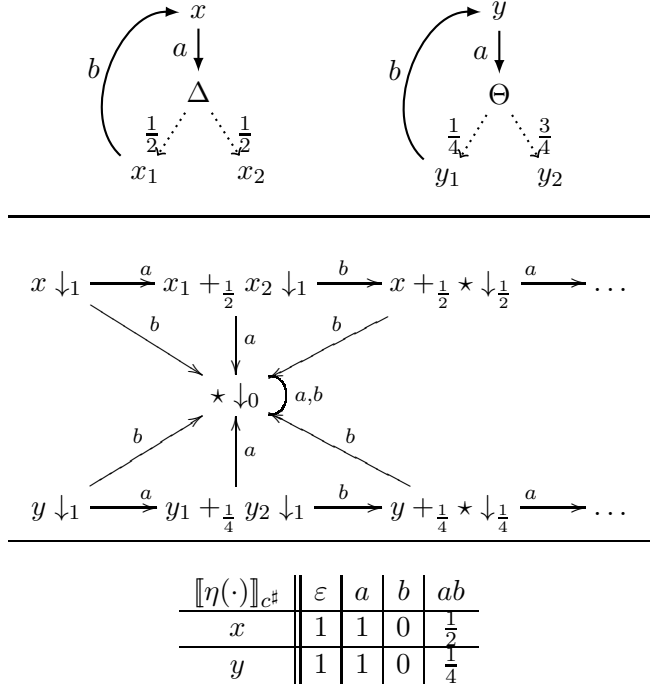

The example above suggests that the general trace equivalence given in Definition~\ref{def:traceequivalence} when dealing with systems with effects in a monad of shape $\T + 1$ might not provide the best results. Indeed, the state space of the determinised system for such a monad would always be of the shape $\T X + 1$, intuitively either an $\T$-combination of states in $X$ or $\star \in 1$. % chktex 1
The solution that we propose here is to work with the monad $\T(\cdot +1)$ rather than  $\T + 1$. In order to do this it is enough to recall the morphism $\iota \colon \T + 1 \to \T(\cdot +1)$ from~\eqref{eq:iota}, transform any system $t\colon X \to (\T X + 1)^A$ into % chktex 1
\begin{equation}\label{eq:postiota}
\bar{t}\colon \xymatrix{X \ar[rr]^t && (\T X + 1)^A \ar[rr]^{\iota^A_X} && (\T (X + 1))^A }
\end{equation}
and then adopt Definition~\ref{def:traceequivalence}.
With this recipe, we can recover the appropriate definition of trace equivalence for RPLTS\@.
\begin{exa}[Trace semantics for RPLTS]\label{ex:rplts}
We first turn an RPLTS $t\colon X \to (\dset X+1)^A$ into a system with $\Dis(\cdot + 1)$-effects by postcomposing with $\iota^A$, as in~\eqref{eq:postiota}. Observe that the monad map $\iota \colon \dset +1 \Rightarrow \Dis(\cdot + 1)$ only changes the type by embedding distributions into subdistributions and regards the unique element of $1$ as the empty subdistribution (namely the one mapping everything to $0$).

Now, we can apply Definition~\ref{def:traceequivalence}. The set of observation is $[0,1]=\dset(1+1)$ equipped with the free $\Dis(\cdot + 1)$-algebra generated by $1$. The observation function $o\colon X \to [0,1]$ maps every state $x\in X$ into the element $1\in [0,1]$. The function $\llbracket{\cdot}\rrbracket_{c^\sharp} \after \eta \colon X \to [0,1]^{A^*}$ obtained via the generalised determinisation of $c = \langle o, \bar{t} \rangle$ assigns to each state $x\in X$ and trace $w\in A^*$ the probability of reaching from $x$ any other state via $w$. We write $\equiv^{RP}$ for the induced trace equivalence. A similar construction has been used in~\cite{YZ14, YJZ17}.

Figure~\ref{fig:RPLTS} illustrates the semantics obtained in this way: observe that, differently from the semantics from Example~\ref{ex:rpltstermination}, now $x \not \equiv^{RP} y$ since $\bb{\eta (x)}_{c^{\sharp}}(ab)=\frac{1}{2}$ and $\bb{\eta (y)}_{c^{\sharp}}(ab)=\frac{1}{4}$.
\end{exa}

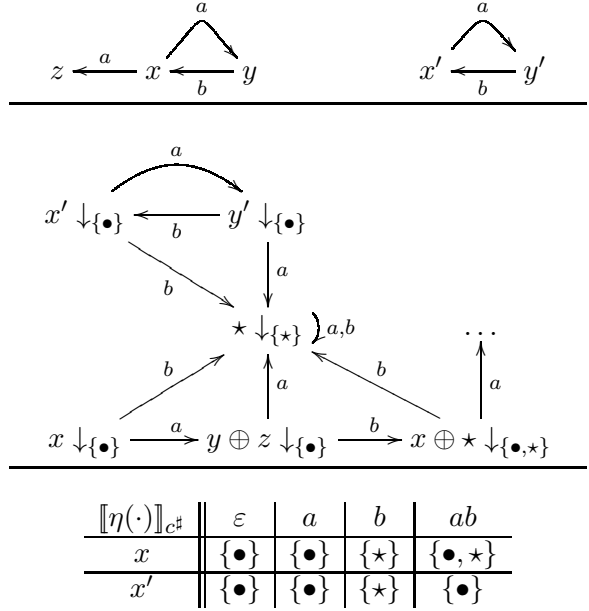
\begin{figure}
\begin{center}
\begin{tabular}{c}
$\xymatrix{
 z &
x \ar[l]_a \ar@(ur,ul)[r]^a & y \ar[l]^b & & x'\ar@(ur,ul)[r]^a & y' \ar[l]^b
}$
\vspace{5mm}\\
\begin{tabular}{c}
$\xymatrix{
x'\downarrow_{\{\bullet\}}\ar@(ur,ul)[r]^a \ar[rd]_b& y'\downarrow_{\{\bullet\}} \ar[l]^b \ar[d]^a\\
& \star \downarrow_{\{\star\}}\ar@(ur,dr)^{a,b} & \dots\\
x\downarrow_{\{\bullet\}}\ar[ru]^b \ar[r]^a & y\cplus z\downarrow_{\{\bullet\}} \ar[r]^b \ar[u]_a & x\cplus \star\downarrow_{\{\bullet,\star\}} \ar[ul]_{b} \ar[u]_a\\
}$
\end{tabular}
\vspace{5mm}\\
$\begin{array}{ccccc}
\bb{\eta (\cdot)}_{c^{\sharp}} & \varepsilon & a& b & ab\\
\toprule
x & \{\bullet\}& \{\bullet\}& \{\star\}& \{\bullet, \star\}\\[0.4em]
x' & \{\bullet\}& \{\bullet\}& \{\star\}& \{\bullet\}\\[0.4em]
\end{array}$
\end{tabular}
\end{center}
\caption{An LTS (top), part of its may-must determinisation (center); and the corresponding may-must semantics (bottom)}\label{figmaymustLTS}
\end{figure}

For LTS, we have already seen that Definition~\ref{def:traceequivalence} leads to the standard notion of trace semantics. However, by recalling that $\Pow=\Powne+1$, one could first transform an LTS into a systems with $\Powne(\cdot  +1)$-effects by~\eqref{eq:postiota} and then use Definition~\ref{def:traceequivalence}. The following example illustrates the alternative trace equivalence obtained in this way.

\begin{exa}[May/Must trace semantics for LTS]\label{ex:mmLTS}
Given an LTS $t\colon X \to (\Pow X)^A$, we define $\bar{t}\colon X \to \Powne(X+1)^A$ as in~\eqref{eq:postiota} by recalling that $\Pow= \Powne+1$. The embedding $\iota_X \colon \Pow X \to \Powne(X + 1)$  maps the empty set $\emptyset \in \Pow X$ into the singleton $\{\star\}$, where $\star$ is the unique element of $1$, and any nonempty subset to itself.

To obtain the algebra of observations we take the free $\Powne(\cdot +1)$-algebra generated by the singleton set $1=\{\bullet\}$. This is the semilattice
\begin{equation}\label{eq:domain}
\xymatrix@R=0.2cm@C=0.2cm{
& \{\bullet, \star\}\\
\{\bullet\} \ar@{-}[ur]& & \{\star\} \ar@{-}[ul]
}
\end{equation} with point $\{\star\}$. Hereafter we refer to this as $3$. The observation function $o\colon X \to 3$ maps every state $x\in X$ into the element $\{\bullet\}\in 3$. We rely on a \emph{testing scenario} to give an intuition of the semantic map $\llbracket{\cdot}\rrbracket_{c^\sharp} \after \eta \colon X \to 3^{A^*}$ obtained via the generalised determinisation of $c = \langle o, \bar{t} \rangle$: words $w\in A^*$ can be thought as tests to be performed on each state $x\in X$ of the original LTS;\@ the outcome of a test is  $\star$ if $x$ does not pass the test, namely it gets stuck during the execution of $w$; the outcome is $\bullet$ if $x$ passes the test, namely it has entirely executed $w$; since the system is nondeterministic different runs of the same test on the same state can lead to different outcomes. If the outcome is always $\bullet$, then $\llbracket{x}\rrbracket_{c^\sharp}(w)=\{\bullet\}$; if it is always $\star$, then $\llbracket{x}\rrbracket_{c^\sharp}(w)=\{\star\}$; if instead the outcome is sometimes $\bullet$ and sometimes $\star$, then $\llbracket{x}\rrbracket_{c^\sharp}(w)=\{\bullet, \star\}$. Figure~\ref{figmaymustLTS} illustrates an LTS and the obtained semantics. Observe that $\bb{\eta(x')}(a)=\bb{\eta(x)}(a)=\{\bullet\}$ since both $x$ and $x'$ always pass the test $a$, while $\bb{\eta(x')}(b)=\bb{\eta(x)}(b)=\{\star\}$ since both $x$ and $x'$ always fail the test $b$. The behaviours of $x'$ and $x$ differ on the test $ab$: $x'$ always passes $ab$, while $x$ may fail it. Indeed, by reading $a$, $x$ can go to $z$ and then get stuck when reading $b$.

We denote the induced equivalence by $\equiv^{LTS}_*$, and we call it may-must trace equivalence for LTS\@. Such a name is justified below by introducing may trace and must trace semantics, and in Remark~\ref{rem:testlts}.

\medskip

Recall the presentation of the monad $\Powne(\cdot +1)$ from Section~\ref{sec:alg-th} and its two possible quotients by the axioms (B) and (T). Let us denote the corresponding quotient maps by $q^B$ and $q^T$, respectively. By applying Definition~\ref{def:traceequivalence} to $(q^B_X)^A \after \bar{t}$ and $(q^T_X)^A \after \bar{t}$ one obtains two different semantics which we call may trace equivalence, denoted by $\equiv^{LTS}_B$, and must trace equivalence, denoted by $\equiv^{LTS}_T$. To better illustrate both of them, we focus on their algebras of observations.

By quotienting the pointed semilattice~\eqref{eq:domain} by (B) and (T), one obtains, respectively, the semilattice with bottom and the semilattice with top freely generated by $1$, depicted on the left and on the right below.
\begin{equation*}%\label{eq:domain}
\xymatrix@R=0.4cm@C=0.2cm{
\{\bullet\}= \{\bullet, \star\}\\
\{\star\} \ar@{-}[u]
}
\qquad
\xymatrix@R=0.4cm@C=0.2cm{
\{\star\}= \{\bullet, \star\}\\
\{\bullet\} \ar@{-}[u]
}\end{equation*}
In the semilattice on the left, the top element is assigned when a state $x$ may pass a test $w$, while the bottom element when $x$ always fails. In the semilattice on the right, the top element is assigned when $x$ can fail, while the bottom element is assigned when $x$ always passes the test.
Figure~\ref{figmayandmustLTS} illustrates an example of the may semantics and the must semantics.

We conclude by observing that, since the quotient of $\Powne(\cdot +1)$ by (B) is exactly $\Pow$ and, since $q^B\after \iota = id$, then  $\equiv^{LTS}_B$ is the standard trace semantics for LTS\@.
\end{exa}

 \begin{figure}
  \begin{center}
 \begin{tabular}{ccccc}
$\xymatrix{
x'\downarrow_{\{\bullet\}}\ar@(ur,ul)[r]^a \ar[rd]_b& y'\downarrow_{\{\bullet\}} \ar[l]^b \ar[d]^a\\
& \star \downarrow_{\{\star\}}\ar@(ur,dr)^{a,b} \\
x\downarrow_{\{\bullet\}}\ar[ru]^b \ar[r]^a & y\cplus z\downarrow_{\{\bullet\}} \ar@(dl,dr)[l]_b \ar[u]_a \\
}$\qquad
&&&&\qquad
$\xymatrix{
x'\downarrow_{\{\bullet\}}\ar@(ur,ul)[r]^a \ar[rd]_b& y'\downarrow_{\{\bullet\}} \ar[l]^b \ar[d]^a\\
& \star \downarrow_{\{\star\}}\ar@(ur,dr)^{a,b} \\
x\downarrow_{\{\bullet\}}\ar[ru]^b \ar[r]^a & y\cplus z\downarrow_{\{\bullet\}}  \ar[u]_{a,b} \\
}$
\\
\\
\\
$\begin{array}{ccccc}
\bb{\eta (\cdot)}_{c^{\sharp}} & \varepsilon & a& b & ab\\
\toprule
x & \{\bullet\}& \{\bullet\}& \{\star\}& \{\bullet\}\\
x' & \{\bullet\}& \{\bullet\}& \{\star\}& \{\bullet\}\\
\end{array}$
&&&&
$\begin{array}{ccccc}
\bb{\eta (\cdot)}_{c^{\sharp}} & \varepsilon & a& b & ab\\
\toprule
x & \{\bullet\}& \{\bullet\}& \{\star\}& \{\star\}\\
x' & \{\bullet\}& \{\bullet\}& \{\star\}& \{\bullet\}\\
\end{array}$
\end{tabular}
 \end{center}
 \caption{May trace (left) and must trace (right) semantics for the LTS in Figure~\ref{figmaymustLTS}}\label{figmayandmustLTS}
 \end{figure}
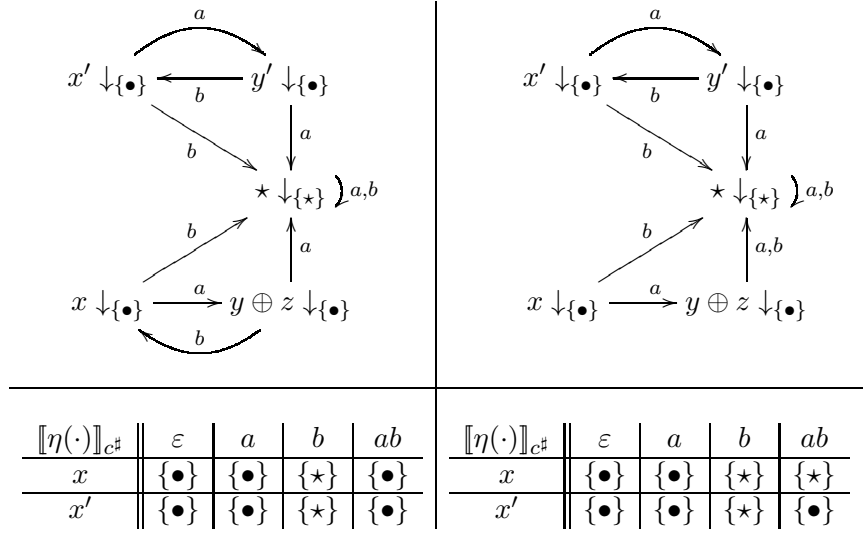

\begin{rem}\label{rem:testlts}
The testing scenario described in Example~\ref{ex:mmLTS}, using the three free algebras generated by the singleton set $1=\{\bullet\}$ as possible outcomes of tests,  is exactly the one of the standard theory of testing equivalences~\cite{DH84}, when taking as class of tests the set of finite traces.
Hence, in this testing scenario the equivalences $\equiv^{LTS}_*$, $\equiv^{LTS}_B$, and $\equiv^{LTS}_T$ are respectively the equivalences known as may-must (or test) testing equivalence, may testing equivalence, and must testing equivalence.
\end{rem}

In Section~\ref{sec:maymust}, we will introduce may-must, may and must semantics for systems with probability and nondeterminism, by basically applying the recipe described in this section to the three monads introduced in Section~\ref{sec:termination}. It is worth to announce here that the obtained semantics will be ``backward compatible'' ---in a sense that will be clarified later--- with $\equiv^{LTS}_*$, $\equiv^{LTS}_B$ and $\equiv^{LTS}_T$. To prove this property we will need the following corollary of Theorem~\ref{thm:transfert}.

\begin{cor}\label{cor:invariance}
Let $(\T,\eta,\mu)$ be a submonad of $(\hat \T,\hat\eta,\hat\mu)$ via an injective monad map $\sigma\colon \T \Rightarrow \hat\T$. Let $t\colon X \to (\T X)^A$ be a system with $\T$-effects and let $\hat t$ be the system with $\hat\T$-effects $\sigma_X^A\after t \colon X \to (\hat \T X)^A$. Let $\equiv, \hat\equiv\subseteq X\times X$ be the trace equivalences for $t$ and $ \hat t $, respectively.
Then $\equiv \,=\,\hat\equiv$. \qed%
\end{cor}

\begin{proof}
We fix
$o=\big(X \stackrel{!}{\longrightarrow}1 \stackrel{\eta_1}{\longrightarrow} \T 1\big)$ and $\hat o=\big(X \stackrel{!}{\longrightarrow}1 \stackrel{\hat\eta_1}{\longrightarrow} \hat\T 1\big)$.

We first transform the automaton $c = \langle o,t\rangle$ to $\hat c = \langle \hat o,t\rangle$ and apply Theorem~\ref{thm:transfert}.2 and then transform $\hat c = \langle \hat o,t\rangle$ to $\hat{\hat c} = \langle \hat o, \hat t\,\rangle$ and apply Theorem~\ref{thm:transfert}.1.

For the determinisation, take for $a$ in Theorem~\ref{thm:transfert}.2 the free algebra $\mu_1\colon \T\T 1 \to \T1$ and as $\hat a$ the $\T$-algebra $(\hat\mu \after \sigma_{\hat \T 1})\colon \T\hat \T 1\to \hat \T 1$. It is easy to see that $\hat a$ is indeed an $\T$-algebra using that $\sigma$ is a monad map, its naturality, and the associativity of $\hat\mu$. Since $\sigma$ is a monad map, the following  diagram commutes showing that $\sigma_1$ is an $\T$-algebra homomorphism.
\[
\xymatrix@R-1pc{
\T\T 1\ar[dd]_{\mu} \ar[r]^{\T\sigma_1} & \T\hat \T 1 \ar[d]^{\sigma_{\hat \T 1}}\\
& \hat \T\hat \T 1 \ar[d]^{\hat\mu}\\
\T 1 \ar[r]^{\sigma_1} & \hat\T 1
}
\]
Observe that $\hat o = \sigma_1 \after o$, again since $\sigma$ is a monad map. Then, by Theorem~\ref{thm:transfert}.2 $\hat{\beh{\cdot}} = \sigma_1^{A^*} \after \beh{\cdot}$
where
$\hat{\bb{\cdot}}$ is the semantics obtained by determinisation of $\hat c$ and $\bb{\cdot}$ the one after determinisation of $c$. Since $\sigma_1$ is injective, also $\sigma_1^{A^*}$ is injective and we have that for all $x,y\in X$, $\hat{\bb{\eta(x)}} = \hat{\bb{\eta(y)}}$ iff $\bb{\eta(x)}=\bb{\eta(y)}$, i.e., the semantics remains the same.

For the second step, take $\hat{\hat a} = \hat\mu\colon \hat\T\hat\T 1 \to \hat\T 1$ for the determinisation of $\hat{\hat{c}}$. By definition $\hat a = \hat{\hat a} \after \sigma_{\hat O}$. Therefore Theorem~\ref{thm:transfert}.1 guarantees that the language semantics of  $\hat{\hat c}$ (i.e., trace semantics of $\hat{t}$) again remains the same as the language semantics  of $\hat c$  (i.e, trace semantics of $t$).
\end{proof}

\section{May / Must Traces for NPLTS}\label{sec:maymust}

In this section, we put all the pieces together and give, using the general recipe from Section~\ref{subsec:sys-aut}, trace semantics for systems featuring nondeterminism, probability, and termination. We work with the monad $\TPCS = C(\cdot + 1)$ and consider its two quotients $\TCSB$ and $\TCST$ that we have illustrated in Section~\ref{sec:termination}. Each of these choices gives us a meaningful trace equivalence.

We start by recalling three types of coalgebras.

\mypar{NPLTS} Nondeterministic probabilistic labelled transition systems, NPLTS, also known as simple Segala systems, are coalgebras for the functor $F = (\Pow\Dis(\cdot))^A$. Behavioural equivalence coincides with strong probabilistic bisimilarity~\cite{BSV04:tcs,Sokolova11}.

\mypar{Convex NPLTS} Convex NPLTS are coalgebras for $(C+1)^A$. Behavioural equivalence coincides with convex probabilistic bisimilarity~\cite{Mio14}.

\mypar{NPA} Nondeterministic Probabilistic automata, NPA, with observations in $O$ are (for us in this paper)  coalgebras for $F = O\times (C(\cdot + 1))^A$. We explained in~Section~\ref{subsec:sys-aut} above how to move from (convex) NPLTS to NPA, which involves two steps: (1) Adding observations and  (2) Dealing with termination.
%\end{itemize}

\medskip

In order to define trace semantics for NPLTS via the general recipe from Section~\ref{subsec:sys-aut}, we first need to transform them into convex NPLTS which are systems with $C+1$-effects. Then, following  Section~\ref{subsec:sys-aut}, we transform them into systems with $C(\cdot +1)$-effects and, finally, we add observations, so to obtain automata with $C(\cdot + 1)$-effects, namely NPA\@.

We exploit the natural transformations $\convex$ from~\eqref{eq:conv} and $\iota$ from~\eqref{eq:iota}. Given an NPLTS $t\colon X \to (\Pow\Dis X)^A$ we fix
\begin{equation}\label{eq:bart}
\xymatrix@C=1.5cm{\bart = \big( X  \ar[r]^t &  (\Powne\Dis X+1)^A
\ar[rr]^{(\convex_X+1)^A}  && (CX+1)^A  \ar[r]^{\iota_X^A}  &  (C(X+1))^A  \big)}
\end{equation}
and add observations as prescribed in~\eqref{eq:observation}
\begin{equation}\label{eq:barO}
\barO= \big( X  \stackrel{!}{\xrightarrow{\hspace*{5pt}}} 1
\stackrel{\eta_1}{\xrightarrow{\hspace*{10pt}}}  C(1+1)  \big)
\end{equation}
Now recall from Proposition~\ref{prop:intervalspcs}  that the pointed convex semilattice freely generated by $1$ is $\mathbb{M}_{\intervals,[0,0]}$ and thus
 $C(1+1) =\intervals$: We easily derive that
 \[\barO(x)= \eta_1(!(x)) = \eta_1(\bullet) = \{\delta_{\bullet}\} = [1,1]\] for all $x\in X$, as $\delta_{\bullet} = 1$ (see proof of Proposition~\ref{prop:intervalspcs}).

 Let $\langle\barO^\sharp, \bart^\sharp\rangle \colon C(X+1) \to \intervals \times C(X+1)^A$ be the determinisation of $\langle o,t\rangle$ w.r.t.\ such algebra and $\llbracket\cdot\rrbracket \colon C( X+1) \to \intervals ^{A^*}$ be the final coalgebra map.
\begin{equation*}
\begin{array}{ll}
\xymatrix{ X \ar[d]_{\langle \barO,\bart \rangle} \ar[r]^{\eta}& C( X+1) \ar[dl]^{\langle \barO^\sharp,\bart^\sharp \rangle} \ar@{-->}[rr]^{\bb{\cdot}}& & \intervals^{A^*} \ar[d]^{\langle \epsilon, \dder \rangle}\\
\intervals\times C( X+1)^A \ar@{-->}[rrr]_{id_{\intervals} \times \bb{\cdot}^A}& & & \intervals \times (\intervals^{A^*})^A
} &
\begin{array}{lcl}
\vspace{0.5cm}\\
\bb{S}(\varepsilon) &= &\barO^\sharp(S)\\
\bb{S}(aw) &=& \bb{\bart^\sharp(S)(a)}(w)
\end{array}
\end{array}
\end{equation*}
\begin{defi}We say that two states $x,y\in X$ are \emph{may-must trace equivalent}, written $x \equiv y$, if and only if $\bb{\eta(x)}=\bb{\eta(y)}$.
\end{defi}

By instantiating~\eqref{eq:alg-det} and using the presentation of the monad $C(\cdot +1)$, we can obtain a convenient inductive definition of $\langle \barO^{\sharp}, \bart^{\sharp} \rangle$:
\[\begin{array}{lcl}
\barO^{\sharp}(S) & = &
\begin{cases}
[1,1] & \text{ if } S=x \text{;}\\
               [0,0]               & \text{ if } S=\star\text{;}\\
               \barO^{\sharp}(S_1) \minmax \barO^{\sharp}(S_2) & \text{ if } S=S_1\cplus S_2  \text{;}\\
               \barO^{\sharp}(S_1) +_p \barO^{\sharp}(S_2) & \text{ if } S=S_1\pplus{p} S_2  \text{.}\\
           \end{cases}
\end{array}\]

\[\begin{array}{lcl}
\bart^{\sharp}(S)(a) & = &
\begin{cases}
\bart(x)(a) & \text{ if } S= x  \text{;}\\
               \star              & \text{ if } S=\star\text{;}\\
               \bart^{\sharp}(S_1)(a) \cplus \bart^{\sharp}(S_2)(a) & \text{ if } S=S_1\cplus S_2  \text{;}\\
               \bart^{\sharp}(S_1)(a) \pplus{p} \bart^{\sharp}(S_2)(a) & \text{ if } S=S_1\pplus{p} S_2  \text{.}\\
           \end{cases}
\end{array}
\]

\mypar{May trace equivalence and must trace equivalence}
Now one may want  to treat termination in a different way and exploit the monads $\TCSB$ and $\TCST$ discussed in Section~\ref{sec:termination}. Given the monad morphisms  $\quotientB\colon \TPCS\Rightarrow \TCSB$ and $\quotientT\colon \TPCS\Rightarrow \TCST$ quotienting   $\TPCS$ by $(B)$ and $(T)$, respectively, one can construct the transition functions
\begin{equation}\label{eq:bartB}
\bartB = (\quotientB_X)^A \after \bart \colon X \to (\TCSB X)^A \qquad \text{ and } \qquad
\bartT =(\quotientT_X)^A \after \bart \colon X \to (\TCST X)^A.
\end{equation}
For the observations, we always use the general recipe of Section~\ref{subsec:sys-aut} and take the observation functions:
\begin{equation*}
\barob= \big( X  \stackrel{!}{\xrightarrow{\hspace*{5pt}}} 1
\stackrel{\eta_1}{\xrightarrow{\hspace*{5pt}}}  \TCSB 1  \big) \qquad \text{ and } \qquad
 \barot= \big( X  \stackrel{!}{\xrightarrow{\hspace*{5pt}}} 1
\stackrel{\eta_1}{\xrightarrow{\hspace*{5pt}}}  \TCST 1  \big).
\end{equation*}
Recall from Proposition~\ref{prop:max} and Proposition~\ref{prop:min} that $\mathbb M\text{ax}_{B} = ([0,1],\max,+_p, 0)$ and $\mathbb M\text{in}_{T} = ([0,1],\min,+_p, 0)$ are,  the free convex semilattice with bottom and, respectively, with top, generated by the singleton set $1$.
Therefore these algebraic structures gives us the determinisation of $\barob$ and $\barot$. Since $\barob(x) = 1$ and $\barot(x) = 1$ for all $x\in X$, the determinisation $\barob^{\sharp}\colon \TCSB X \to [0,1]$ and $\barot^{\sharp}\colon \TCST X \to [0,1]$ enjoy the following inductive definition.
\[\begin{array}{lcl}
\barob^{\sharp}(S) & = &
\begin{cases}
1 & \text{ if } S= x \text{;}\\
               0               & \text{ if } S= \star\text{;}\\
               \barO^{\sharp}(S_1) \max \barO^{\sharp}(S_2) & \text{ if } S=S_1\cplus S_2  \text{;}\\
               \barO^{\sharp}(S_1) +_p \barO^{\sharp}(S_2) & \text{ if } S=S_1\pplus{p} S_2  \text{.}\\
           \end{cases}
\end{array}\]

\[ \begin{array}{lcl}
\barot^{\sharp}(S) & = &
\begin{cases}
1 & \text{ if } S= x \text{;}\\
               0               & \text{ if } S= \star\text{;}\\
               \barO^{\sharp}(S_1) \min\barO^{\sharp}( S_2) & \text{ if } S=S_1\cplus S_2  \text{;}\\
               \barO^{\sharp}(S_1) +_p \barO^{\sharp}(S_2) & \text{ if } S=S_1\pplus{p} S_2  \text{.}\\
           \end{cases}
\end{array}
\]
The transition functions $\bartB^{\sharp} \colon \TCSB X \to (\TCSB X)^A$ and $\bartT^{\sharp} \colon \TCST X \to (\TCST X)^A$ are defined in the same way like $\bart^{\sharp}$ above.

The coalgebras $\langle \barob^{\sharp}, \bartB^{\sharp} \rangle$ and $\langle \barot^{\sharp}, \bartT^{\sharp} \rangle$ give rise to morphisms $\bbmay{\cdot}\colon \TCSB X \to [0,1]^{A^*}$ and $\bbmust{\cdot}\colon \TCST X  \to [0,1]^{A^*}$ and corresponding behavioural equivalences.

\begin{defi} The \emph{may trace equivalence} for the NPLTS is denoted by
 $\eqmay$ and defined as
 \[x \eqmay y \quad \Leftrightarrow \quad \bbmay{\eta(x)}=\bbmay{\eta(y)}\] where $\eta$ is the unit of the monad $\TCSB$.
 The \emph{must trace equivalence} for the NPLTS is denoted by $\eqmust$ and defined by
 \[x \eqmust y \quad \Leftrightarrow \quad\bbmust{\eta(x)}=\bbmust{\eta(y)}\] for $\eta$ denoting now the unit of the monad $\TCST$.
\end{defi}

\bigskip

\begin{rem}%
\label{rem:detinvariance}
One might exploit the generalised determinisation in different ways, but these always lead to the above semantics.

Consider the coalgebra $\langle \barob,\bart \rangle\colon X \to \TCSB 1 \times (\TPCS X)^A$ and observe that the algebra $\maxalg_B = ([0,1],\max,+_p, 0)$, namely $\mu_1 \colon \TCSB\TCSB 1 \to \TCSB 1$, is also a pointed convex semilattice---formally this is $\mu_1 \after \quotientB \colon \TPCS\TCSB 1 \to \TCSB 1$. One could thus perform the generalised determinisation w.r.t.\ this algebra and the monad $\TPCS$ and obtain an equivalence that we denote by $\eqmay'$.
Theorem~\ref{thm:transfert}.1 guarantees however that ${\eqmay'}={\eqmay}$. Similarly, one could start with the coalgebra $\langle \barot,\bart \rangle$, apply the same construction and end up with an equivalence which, by Theorem~\ref{thm:transfert}.1, coincides with $\eqmust$.
\end{rem}

\begin{rem}\label{rem:gsosC}
By recalling that the monad $C(\cdot + 1)$ is presented by the algebraic theory $\PCS= (\Sigma_{NP}\cup \Sigma_T, E_{NP})$, one can use Corollary~\ref{cor:invarianceGSOS} to determinise w.r.t.\ the term monad $T_{\Sigma_{NP}\cup \Sigma_T}$ without changing the resulting semantics. The states of the systems determinised in this way are now syntactic terms, i.e, without the axioms in $E_{NP}$, and the determinisation can be expressed by means of GSOS rules. Indeed, by instantiating the GSOS rules~\eqref{generalGSOS1} and~\eqref{generalGSOS2}  for an arbitrary signature to the signature $\Sigma_{NP}\cup \Sigma_T$, one obtains exactly the rules displayed in Table~\ref{table:GSOS}.

The rule on the right in~\eqref{generalGSOS1} gives rise to the rules for $\star$, $\cplus$ and $\pplus{p}$ in Table~\ref{table:GSOS}.(a). The rule on the left gives rise to the rules for $\star$, $\cplus$ and $\pplus{p}$ in Table~\ref{table:GSOS}.(b), (c) and (d) when taking as algebras of observations $\maxalg_B$, $\minalg_T$ and $\mathbb{M}_{\intervals,[0,0]}$, respectively.

The rules in~\eqref{generalGSOS2} require more explanation. In the rule on the left, $o$ should be instatiated with $\barob$, $\barot$ and $\barO$. By recalling that for all $x\in X$, $\barob(x)=\barot(x)=1$ and $\barO(x)=[1,1]$ one obtains the rules for $x$ in Table~\ref{table:GSOS}.(b), (c) and (d). In order to see that the rule on the right of~\eqref{generalGSOS2} gives rise to the rule for $x$ in Table~\ref{table:GSOS}.(a), one should take $t'$ as $(\termfun_X)^{A} \after t$ and observe that $(q^{E_{NP}}_X)^{A} \after t'$ coincides with $\bart$.
\end{rem}

\begin{exa}\label{ex:det}
Consider the convex closure of the \pss from Figure~\ref{fig:examplesys}.
Following Remark~\ref{rem:gsosC}, we can syntactically describe the convex sets of subdistributions reached by a state when performing a transition as follows:
\[x \ttrel a x_{1} \cplus (x_{3} \pplus{\frac 1 2} x_{2})\]

\[y \ttrel a y_{1} \cplus (y_{4} \pplus{\frac 1 2} y_{2}) \cplus ((y_{2} \pplus{\frac 1 2} y_{4}) \pplus{\frac 1 2} y_{3})\]

 \[x_{1} \ttrel b x \pplus{\frac 1 2} x_{3}
\qquad
 y_{1} \ttrel b y \pplus {\frac 1 2}  y_{4} \]

\[ x_{2} \ttrel b x_{3}
\qquad x_{2} \ttrel c x
\qquad
 y_{2} \ttrel b  y_{4}
\qquad y_{3} \ttrel c y
\]
In the determinised system, we have
\[x \ttrel a S_{1} \ttrel b S_{2} \qquad y \ttrel a S'_{1} \ttrel b S'_{2}\]
For
\begin{eqnarray*}S_{1} &=& x_{1} \cplus (x_{3} \pplus{\frac 1 2} x_{2}) \\
 S_{2} &=& (x \pplus{\frac 1 2} x_{3}) \cplus (\onev \pplus{\frac 1 2} x_{3})\\
S'_{1} &=& y_{1} \cplus (y_{4} \pplus{\frac 1 2} y_{2}) \cplus ((y_{2} \pplus{\frac 1 2} y_{4}) \pplus{\frac 1 2} y_{3})\\
 S'_{2} &=& (y\pplus{\frac 1 2} y_{4}) \cplus(\onev  \pplus{\frac 1 2} y_{4}) \cplus ((y_{4} \pplus{\frac 1 2} \star) \pplus{\frac 1 2} \onev)
 \end{eqnarray*}

Consider now the observations associated to the terms in the may-must semantics. We have
$\barO^{\sharp}(x)=[1,1]=\barO^{\sharp}(y)$
and hence
\[
\barO^{\sharp}(S_{1}) = [1,1] \minmax ([1,1] \pplus{\frac 1 2} [1,1])= [1,1].
\]
Analogously, $\barO^{\sharp}(S'_{1})=[1,1]$.
Furtheron
\[
\barO^{\sharp} (S_{2}) =  ([1,1] \pplus {\frac 1 2} [1,1])  \minmax ([0,0] \pplus{\frac 1 2} [1,1])= [\frac 1 2, 1]
\]
and in the same way we derive $\barO^{\sharp} (S'_{2}) = [\frac 1 4,1]$.

Hence, $x$ and $y$ are not may-must trace equivalent:
$\bb {x} (ab) = \barO^{\sharp} (S_{2}) \neq \barO^{\sharp} (S'_{2}) = \bb {y} (ab)$.

However, using $\maxalg_B$, we get $\barob^{\sharp} (S_{2})= \barob^{\sharp} (S'_{2})$ as the intervals obtained via the may-must observation over $S_{2},S'_{2}$ have the same upper bound $1$, which is the value returned by both $\barob^{\sharp} (S_{2})$ and $\barob^{\sharp} (S'_{2})$.
Hence,
$\bb {x}_B (ab) = \barob^{\sharp} (S_{2}) = \barob^{\sharp} (S'_{2}) = \bb {y}_B (ab)$.
More generally, it holds that $x$ and $y$ are may trace equivalent. We can elegantly prove this by using up-to techniques, as shown in Section~\ref{sec:upto}.%~\ref{sec:upto}.
\end{exa}

\subsection{Properties of the semantics}

The three notions of trace equivalence for NPLTS that we have introduced above, $\equiv$, $\eqmay$ and $\eqmust$, enjoy several desirable properties that, on the one hand, confirm the appropriateness of our semantics and, on the other, provide useful techniques for reasoning about them. In this section we show such properties while in Section~\ref{sec:upto}, we illustrate such techniques at work.

\subsubsection{Bisimilarity implies trace equivalence}
For NPLTS, there exist two main notions of bisimilarity, which are usually called \emph{(strong) probabilistic bisimilarity} and \emph{convex probabilistic bisimilarity}~\cite{SL94}. Originally, these relations were called (strong) bisimulation and (strong) probabilistic bisimulation, respectively. We show here that both imply may-must, may, and must trace equivalence. For this purpose, we will repeatedly use the following result from the general theory of coalgebra.

\begin{lemC}[{\cite[Theorem 15.1]{Rut00:tcs}}]\label{lemma:nattrans}
Let $F$ and $G$ be two endofunctors on $\Sets$ and $c\colon X \to FX$ be an $F$-coalgebra.
If there exists a natural transformation $\alpha \colon F \Rightarrow G$, then behavioural equivalence for $c$ implies  behavioural equivalence for $\alpha \after c$. \qed%
\end{lemC}

\begin{thm}
Probabilistic bisimilarity (called bisimilarity in the original paper) and convex probabilistic bisimilarity (called probabilistic bisimilarity in the original paper) from~\cite{SL94} imply $\equiv$, $\eqmay$, and $\eqmust$.
\end{thm}

\begin{proof}
Given an NPLTS, namely a coalgebra $t\colon X \to (\Pow\Dis X)^A$, we know from~\cite{BSV04:tcs,Sokolova11} that probabilistic bisimilarity of~\cite{SL94} (denoted $\approx^p$) coincides with behavioural equivalence. Moreover, by~\cite{Mio14}, behavioural equivalence for $(\convex_X +1)^A \after t$ coincides with convex probabilistic bisimilarity ($\approx^c$). Using Lemma~\ref{lemma:nattrans} twice, we have that $\approx^p \,\, \subseteq \,\,\approx^c \,\,\subseteq\,\, \approx^{\iota}$ where $\approx^{\iota}$ is behavioural equivalence for $\iota_X^A \after (\convex_X +1)^A \after t$, namely $\bart$ from~\eqref{eq:bart}.

Now, by postcomposing $\bart$ with the natural transformation $C(\cdot+1)^A \stackrel{\cong}{\Rightarrow} 1 \times C(\cdot+1)^A \stackrel{\eta_1 \times id}{\Rightarrow}C(1+1) \times C(\cdot+1)^A$ instantiated at $X$, one obtains exactly $c=\langle \barO, \bart\rangle$. Again by Lemma~\ref{lemma:nattrans}, $\approx^{\iota} \,\,\subseteq\,\, \approx$ where $\approx$ denotes here behavioural equivalence for $c$. By Theorem~\ref{thm:det-prop}.2, we now can conclude that $ \approx\,\, \subseteq \,\,\equiv$. A similar argument applies for $\eqmay$ and $\eqmust$.
\end{proof}

\subsubsection{Backward compatibility}
Both LTS and RPLTS can be regarded as special cases of NPLTS:\@ LTS are NPLTS where all distributions are Dirac distributions; RPLTS are NPLTS where all subsets are at most singletons. Formally, we can express this using the natural transformations $\Powne \eta^{\Dis} \colon \Powne \Rightarrow \Powne \Dis$ and $\eta^{\Powne}{\Dis} \colon \Dis \Rightarrow \Powne \Dis$.
Given an LTS $t\colon X \to (\Pow X)^A$ we call
\begin{equation*}
\tilde{t} = \xymatrix{X \ar[rr]^t && (\Powne X+1)^A \ar[rr]^{ (\Powne \eta^{\Dis}_X+1)^A } && (\Powne \Dis X +1)^A }
\end{equation*}
the corresponding NPLTS\@. Similarly, the NPLTS corresponding to an RPLTS $t\colon X \to (\Dis X +1)^A$ is
\begin{equation*}
\tilde{t} = \xymatrix{X \ar[rr]^t && (\Dis X+1)^A \ar[rr]^{ (\eta^{\Powne}{\Dis}_X+1)^A } && (\Powne \Dis X +1)^A }
\end{equation*}

The following two theorems state that $\eqmay$, $\eqmust$, and $\equiv$ generalize---through these translations---the trace equivalences for RPLTS and LTS introduced in Example~\ref{ex:rplts} and Example~\ref{ex:mmLTS}.

\begin{thm}[Backward Compatibility for RPLTSs]
Let $\equiv^{RP}$ be trace equivalence defined on an RPLTS $t$ and $\eqmay$, $\eqmust$ and $\equiv$ the may, must and may-must equivalences defined on the corresponding NPLTS $\tilde{t}$. Then \[\equiv^{RP} \,\, =\,\,  \equiv \,\, = \,\, \eqmay \,\, = \,\, \eqmust \]
\end{thm}
\begin{proof}
We prove $\equiv^{RP}\,\,=\,\,  \equiv $ by applying Corollary~\ref{cor:invariance} to the injective monad map $\chi^{\Dis}{(\cdot +1)}$ from Lemma~\ref{lemma:monadmaps}. In order to do so, it is enough to observe that the following diagram commutes: the two triangles by definition of $\tilde{t}$ and $\chi^{\Dis}$ from Lemma~\ref{lemma:injmapC}; the square by Lemma~\ref{lem:map-iotas}.
\begin{equation*}
\xymatrix{
X \ar[rr]^t \ar[dd]_{\tilde{t}}&& (\Dis X +1)^A \ar[lldd]^{(\eta^{\Powne}{\Dis}_ X+1)^A} \ar[rr]^{\iota_X^A} \ar[dd]^{(\chi^{\Dis}_X +1)^A}&&  \Dis(X+1)^A \ar[dd]^{{\chi^{\Dis}(\cdot +1)^A_X}} \\&&&& \\
(\Powne \Dis X +1)^A \ar[rr]_{(\convex_X+1)^A}
&& (C X+1)^A \ar[rr]_{\iota_X^A} && C(X+1)^A
}\end{equation*}
In the diagram above, the topmost arrow is exactly $\bar{t}\colon X \to \Dis(X+1)^A$ used to define $\equiv^{RP}$  over $t$ (Example~\ref{ex:rplts}). The down-and-right arrow $\iota_X^A \after (\convex_X+1)^A \after \tilde{t}$ is exactly the arrow $\bar{\tilde{t}}$ used to define $\equiv$ over $\tilde{t}$, see~\eqref{eq:bart}.
Commutativity of the above diagram states that $(\chi^{\Dis}(\cdot+1))^A_X \circ \bar{t} = \bar{\tilde{t}}$ which, by Corollary~\ref{cor:invariance} entails that $\equiv^{RP}\,\,=\,\,  \equiv $.

To prove $\equiv^{RP}\,\,= \,\, \eqmay$ we use Corollary~\ref{cor:invariance} with the injective monad map $q^B_X \circ \chi^{\Dis}(\cdot +1)_X$ from Example~\ref{lemma:monadmaps}. Observe that the commutativity of the above diagram gives us also that $(q^B_X \after \chi^{\Dis}(\cdot +1)_X)^A \circ \bar{t} = (q^B_X)^A \circ \bar{\tilde{t}}$ and that $(q^B_X)^A \circ \bar{\tilde{t}}$ is exactly the system used to define $\eqmay$ on $\tilde{t}$, see~\eqref{eq:bartB}.

The proof for $\equiv^{RP}\,\,=\,\,  \eqmust$ is obtained by replacing in the previous paragraph $B$ by $T$.
\end{proof}

\begin{thm}[Backward Compatibility for LTSs]
Let $\equiv^{LTS}_B$, $\equiv^{LTS}_T$, and $\equiv^{LTS}_*$ be may, must, and may-must trace equivalences  on an LTS $t$ and $\eqmay$, $\eqmust$, and $\equiv$ the equivalences defined on the corresponding NPLTS $\tilde{t}$. Then \[\equiv^{LTS}_* \,\, =\,\,  \equiv \qquad \equiv^{LTS}_B\,\,=\,\,  \eqmay \qquad \equiv^{LTS}_T \,\,=\,\, \eqmust \]
\end{thm}
\begin{proof}
The proof for $\equiv^{LTS}_* \,\, =\,\,  \equiv$ is like the one for $\equiv^{RP} \,\, = \,\, \equiv$, but replacing $\Dis$ by $\Powne$. The proof for $\eqmay$ requires one additional step: in the diagram below, the rightmost square commutes by Lemma~\ref{lemma:monadmaps}.3.
\begin{equation*}\xymatrix{
X \ar[rr]^t  \ar[dd]_{\tilde{t}} && (\Powne X +1)^A \ar[lldd]^{(\nePow\eta^{\Dis}_X+1)^A} \ar[r]^{\iota_X^A} \ar[dd]^{(\chi^{\Powne}_X +1)^A}&  \Powne(X+1)^A \ar[dd]^{{(\chi^{\Powne}(\cdot +1))^A_X}}   \ar[r]^ {(q^B_X)^A}  & \TSB(X)^A \ar[dd]^{(e_X^B)^A}  \\ &&&&\\
(\Powne \Dis X+1)^A \ar[rr]_{(\convex_X+1)^A}
&& (C X+1)^A \ar[r]_{\iota_X^A} &C(X+1)^A \ar[r]_{(q^B_X)^A} & \TCSB(X)^A
}\end{equation*}
Like in the case of RPLTS, the two triangles commute by definition of $\tilde{t}$ and $\chi^{\Powne}$ and the central square by Lemma~\ref{lem:map-iotas}. So, the whole diagram commutes. Observe that  the topmost arrow is exactly $(q^B_X)^A\circ \bar{t}$ used to define $\equiv^{LTS}_B$  over $t$ (Example~\ref{ex:mmLTS}). The down-and-right arrow $(q^B_X)^A\after \iota_X^A \after (\convex_X+1)^A \after \tilde{t}$ is exactly the arrow $(q^B_X)^A\circ \bar{\tilde{t}}$ used to define $\eqmay$ over $\tilde{t}$ (see~\eqref{eq:bartB}). By using Corollary~\ref{cor:invariance} with the injective monad map $e^B$ from Lemma~\ref{lemma:monadmaps}, we obtain that $\equiv^{LTS}_B\,\,=\,\, \eqmay$.
The proof for $\equiv^{LTS}_T\,\,=\,\,  \eqmust$ is obtained by replacing in the previous paragraph $B$ by $T$.
\end{proof}

\subsubsection{The bialgebra of probabilistic traces}

The generalised determinisation outlined in Section~\ref{subsec:determinisation} allows to think of the final coalgebra as a denotational universe of behaviours and of the final  map $\bb{\cdot}$ as a denotational semantics assigning to each state its behaviour. We better illustrate this idea here by considering the may-must semantics, but the same arguments hold for the may and the must cases.

\medskip

The set $ \intervals^{A^*}$ of nondeterministic-probabilistic languages, namely functions $\varphi \colon A^* \to \intervals$, carries at the same time the final coalgebra $\langle \epsilon, \dder\rangle \colon \intervals^{A^*} \to \intervals \times \intervals^{A^*}$ (defined as in Section~\ref{subsec:determinisation} for $O=\intervals$) and a pointed convex semilattice.
The latter is defined as the pointwise extension of the algebra of observations $\mathbb M_{\intervals,[0,0]}$, see Proposition~\ref{prop:intervalspcs}. This means that the structure of pointed convex semilattice is defined as follows for all $\varphi_1,\varphi_2\in \intervals^{A^*}$ and for all $w\in A^*$.
\begin{equation*}
\begin{array}{rcl}
(\varphi_1 \cplus^{\intervals^{A^*}} \varphi_2)(w)&=& \varphi_1(w)\minmax \varphi_2(w) \\ (\varphi_1 \pplus{p}^{\intervals^{A^*}} \varphi_2)(w)&=& \varphi_1(w)\pplus{p}^\intervals \varphi_2(w) \\ \star^{\intervals^{A^*}}(w)&=&[0,0]
\end{array}
\end{equation*}
By the generalised determinisation, we know that the final coalgebra map $\bb{\cdot}\colon C(X+1) \to \intervals^{A^*}$ is also a homomorphism of pointed convex semilattices which means that
\begin{equation*}
\bb{S_1\cplus S_2}= \bb{S_1} \cplus^{\intervals^{A^*}}\bb{S_2}\text{, }  \bb{S_1\pplus{p} S_2}= \bb{S_1} \pplus{p}^{\intervals^{A^*}}\bb{S_2} \text{ and } \bb{\star}=\star^{\intervals^{A^*}}
\end{equation*}
for all $S_1,S_2\in C(X+1)$. This guarantees the following result.

\begin{thm}
Let $\cong$, $\cong_B$ and $\cong_T$ be the kernels of $\bb{\cdot}$, $\bbmay{\cdot}$, and $\bbmust{\cdot}$, respectively. Then $\cong$, $\cong_B$, and $\cong_T$ are congruences w.r.t. $\cplus$ and $\pplus{p}$.
\end{thm}

The theorem above is an instantiation of Theorem~\ref{thm:det-prop}.1 to our construction for NPLTS\@. By instantiating Theorem~\ref{thm:det-prop}.3, we have the following useful fact.

\begin{thm}\label{them:upto}
Up-to context is compatible for each of the three equivalences.\qed%
\end{thm}
In the next section, we illustrate the implication of such result. Since $x \equiv y$ iff $\eta(x)\cong\eta(y)$, hereafter we will sometimes write $\equiv$ instead of $\cong$. Similarly for $\eqmay$ and $\eqmust$.

\subsection{Coinduction Up-to}\label{sec:upto}

As anticipated in Theorem~\ref{them:upto}, $\eqmaymust$, $\eqmay$, and $\eqmust$  can be proved coinductively by means of bisimulation up-to. In order to define uniformly the proof techniques for the three equivalences, we let $\eqthree$ to range over  $\eqmaymust$, $\eqmay$, and $\eqmust$; $\Tthree$ to range over $\TPCS$, $\TCSB$, and  $\TCST$; $\tthree$, over $\bart^{\sharp}$, $\bartB^{\sharp}$,  and $\bartT^{\sharp}$; $\othree$ over $\barO^{\sharp}$, $\barob^{\sharp}$, and $\barot^{\sharp}$.

\begin{defi}\label{def:bis} Let $(X,t)$ be an NPLTS and $(\Tthree X,\langle \othree, \tthree \rangle)$ the corresponding determinised system. A relation $\R \subseteq \Tthree X \times \Tthree X$ is a bisimulation iff for all $(\termone,\termtwo)\in R$ it holds that
\begin{enumerate}
\item $\othree(\termone)=\othree(\termtwo)$ and
\item $\tthree(\termone)(a)\, \RR \, \tthree(\termtwo)(a)$ for all $a\in A$.
\end{enumerate}
\end{defi}

\noindent
The coinduction proof principle  (see e.g.~\cite{DBLP:journals/acta/BonchiPPR17}) asserts that for all $x,y\in X$, $x \eqthree y$ iff there exists a bisimulation $\R$ such that $x \, \R \, y$.

We can now prove that states $x,y$ in Figure~\ref{fig:examplesys} are may trace equivalent by showing that there exists a bisimulation $\R$ such that $x\RR y$ on the system determinized using as algebra of observations $\maxalg_B$.

\begin{exa}%
\label{secapp:example}
Consider the \pss $(X,t)$ depicted in Figure~\ref{fig:examplesys}, and discussed in Example
%\ref{ex:sys} and~\ref{ex:det}.
Figure~\ref{fig:exampledet} shows the determinization of the system, where
the terms are, as already partially mentioned in Example
%\ref{ex:sys} and~\ref{ex:det},
as follows
\[S_{1} = x_{1} \cplus (x_{3} \pplus{\frac 1 2} x_{2})
\qquad S_{2} = (x \pplus{\frac 1 2} x_{3}) \cplus (\onev \pplus{\frac 1 2} x_{3})\]
\[S'_{1} = y_{1} \cplus (y_{4} \pplus{\frac 1 2} y_{2}) \cplus ((y_{2} \pplus{\frac 1 2} y_{4}) \pplus{\frac 1 2} y_{3})\]
\[ S'_{2} = (y\pplus{\frac 1 2} y_{4}) \cplus(\onev  \pplus{\frac 1 2} y_{4}) \cplus ((y_{4} \pplus{\frac 1 2} \star) \pplus{\frac 1 2} \onev) \]
\[S_{3} = \onev \cplus (\onev \pplus{\frac 1 2} x) \qquad S'_{3} = \onev \cplus (\onev \pplus{\frac 1 2} y)
 \]
 and the depicted transitions are those given by $\bart^\sharp$.

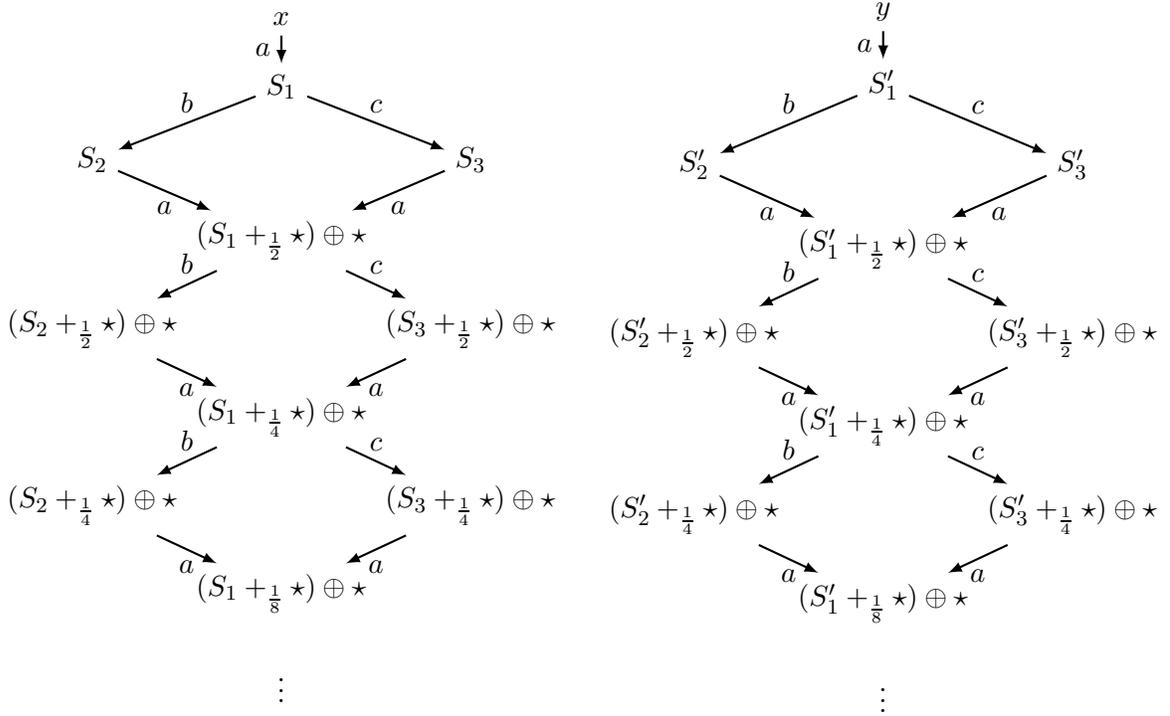
\begin{figure*}[t]
\begin{center}
\begin{tikzpicture}[thick]

\begin{scope}[thick,scale=0.5]

\matrix[matrix of nodes, row sep= 0.4cm, column sep=0.007pt,ampersand replacement=\&]
{
				\&\node (x) {$x$};		\\
				\&\node (t1) {$S_{1}$};		\\
\node (t2) {$S_{2}$};	\&					\&\node (t3) {$S_{3}$};	\\
				\&\node (t11) {$(S_{1} \pplus {\frac 1 2 } \onev) \cplus \onev$};		\\
\node (t21) {$(S_{2} \pplus {\frac 1 2 } \onev) \cplus \onev$};	\&					\&\node (t31) {\!\!\!\!\!\!$(S_{3} \pplus {\frac 1 2 } \onev) \cplus \onev$};	\\
				\&\node (t12) {$(S_{1} \pplus {\frac 1 4 } \onev) \cplus \onev$};		\\
\node (t22) {$(S_{2} \pplus {\frac 1 4 } \onev) \cplus \onev$};	\&					\&\node (t32) {\!\!\!\!\!\!$(S_{3} \pplus {\frac 1 4 } \onev) \cplus \onev$};	\\
				\&\node (t13) {$(S_{1} \pplus {\frac 1 8 } \onev) \cplus \onev$};		\\
				\&\node (vdots) {$\vdots$};		\\
	};
\draw[-latex] (x) to node[left] {$a$} (t1);
\draw[-latex] (t1) to node[above] {$c$} (t3);
\draw[-latex] (t1) to node[above] {$b$} (t2);
\draw[-latex] (t2) to node[below] {$a$} (t11);
\draw[-latex] (t3) to node[below] {$a$} (t11);
\draw[-latex] (t11) to node[above] {$c$} (t31);
\draw[-latex] (t11) to node[above] {$b$} (t21);
\draw[-latex] (t21) to node[below] {$a$} (t12);
\draw[-latex] (t31) to node[below] {$a$} (t12);
\draw[-latex] (t12) to node[above] {$c$} (t32);
\draw[-latex] (t12) to node[above] {$b$} (t22);
\draw[-latex] (t22) to node[below] {$a$} (t13);
\draw[-latex] (t32) to node[below] {$a$} (t13);

\end{scope}

\begin{scope}[xshift=8cm,scale=0.5]
\matrix[matrix of nodes, row sep= 0.4cm, column sep=0.007pt,ampersand replacement=\&]
{
				\&\node (x) {$y$};		\\
				\&\node (t1) {$S'_{1}$};		\\
\node (t2) {$S'_{2}$};	\&					\&\node (t3) {$S'_{3}$};	\\
				\&\node (t11) {$(S'_{1} \pplus {\frac 1 2 } \onev) \cplus \onev$};		\\
\node (t21) {$(S'_{2} \pplus {\frac 1 2 } \onev) \cplus \onev$};	\&					\&\node (t31) {\!\!\!\!\!\!\!\!\!\!\!\!$(S'_{3} \pplus {\frac 1 2 } \onev) \cplus \onev$};	\\
				\&\node (t12) {$(S'_{1} \pplus {\frac 1 4 } \onev) \cplus \onev$};		\\
\node (t22) {$(S'_{2} \pplus {\frac 1 4 } \onev) \cplus \onev$};	\&					\&\node (t32) {\!\!\!\!\!\!\!\!\!\!\!\!$(S'_{3} \pplus {\frac 1 4 } \onev) \cplus \onev$};	\\
				\&\node (t13) {$(S'_{1} \pplus {\frac 1 8 } \onev) \cplus \onev$};		\\
				\&\node (vdots) {$\vdots$};		\\
	};
\draw[-latex] (x) to node[left] {$a$} (t1);
\draw[-latex] (t1) to node[above] {$c$} (t3);
\draw[-latex] (t1) to node[above] {$b$} (t2);
\draw[-latex] (t2) to node[below] {$a$} (t11);
\draw[-latex] (t3) to node[below] {$a$} (t11);
\draw[-latex] (t11) to node[above] {$c$} (t31);
\draw[-latex] (t11) to node[above] {$b$} (t21);
\draw[-latex] (t21) to node[below] {$a$} (t12);
\draw[-latex] (t31) to node[below] {$a$} (t12);
\draw[-latex] (t12) to node[above] {$c$} (t32);
\draw[-latex] (t12) to node[above] {$b$} (t22);
\draw[-latex] (t22) to node[below] {$a$} (t13);
\draw[-latex] (t32) to node[below] {$a$} (t13);
\end{scope}

\end{tikzpicture}
\end{center}
\caption{Determinization}\label{fig:exampledet}
\end{figure*}
The determinized \pss is a system with infinitely many states, which are given by the presence of cycles in the original system. In the determinization of the automaton with algebra of observation $\maxalg_B$, each state $S$ is assigned an observation $\barob^{\sharp}(S) \in [0,1]$.
We prove that $x$ and $y$ are may trace equivalent by exhibiting the following bisimulation $\R$:
\begin{align*}
\R = &\{(x,y),(S_{1},S'_{1}), (S_{2},S'_{2}), (S_{3},S'_{3})\}\\
&\cup \{((S_{i} \pplus {\frac 1 {2^{n}} } \onev) \cplus \onev, (S'_{i} \pplus {\frac 1 {2^{n}} } \onev) \cplus \onev) |\, 1\leq i \leq 3, n\geq 1 \}
\end{align*}
The relation satisfies the two clauses required by Definition~\ref{def:bis} of  bisimulation. As it emerges from Figure~\ref{fig:exampledet}, the clause on transitions (clause 2) is satisfied by each pair in the relation.
As to the clause on the observation (clause 1), we can derive as in Example~\ref{ex:det} that for every pair $(S,S')\in \{(x,y),(S_{1},S'_{1}), (S_{2},S'_{2}), (S_{3},S'_{3})\}$ it holds
$\barob^{\sharp}(S)=\barob^{\sharp}(S')$.

Finally, clause 1 also holds for the remaining pairs, since for $1\leq i \leq 3 $ and $n\geq 1$ we have
\[
\barob^{\sharp}((S_{i} \pplus {\frac 1 {2^{n}} } \onev) \cplus \onev)
=(\barob^{\sharp}(S_{i}) \pplus {\frac 1 {2^{n}} } 0 )\max 0
=(\barob^{\sharp}(S'_{i}) \pplus {\frac 1 {2^{n}} } 0 )\max 0
= \barob^{\sharp}((S'_{i} \pplus {\frac 1 {2^{n}} } \onev) \cplus \onev)
\]
Hence, $\R$ is a bisimulation.

As shown in Example~\ref{ex:det}, $x,y$ are not bisimilar if the algebra of observations for the must equivalence, i.e, $\minalg_T$,  is used instead of the one for the may equivalence, since $\barot^{\sharp} (S_{2}) \neq \barot^{\sharp} (S'_{2})$. Analogously, they are not equivalent if we take the may-must algebra of observation $\mathbb{M}_{\intervals,[0,0]}$.
\end{exa}

To make this proof principle more effective, one can use \emph{up-to techniques}~\cite{Milner89,SP09b}. Particularly relevant for us is up-to \emph{contextual closure} which, for all relations $\R\subseteq \Tthree(X) \times \Tthree(X)$, is defined inductively by the following rules.
\[
\inferrule*{\term \mathrel{\R} \term'}{\term \mathrel{\Ctx(\R)} \term'} \quad \quad \inferrule*{-}{* \mathrel{\Ctx(\R)} *}
\]\[ \inferrule*{\termone \mathrel{\Ctx(\R)} \termone' \and
 \termtwo \mathrel{\Ctx(\R)} \termtwo'}{\termone \cplus \termtwo \mathrel{\Ctx(\R)} \termone' \cplus \termtwo' }
\]
\[\inferrule*{\termone \mathrel{\Ctx(\R)} \termone' \and
 \termtwo \mathrel{\Ctx(\R)} \termtwo'}{\termone \pplus{p} \termtwo \mathrel{\Ctx(\R)} \termone' \pplus{p} \termtwo' }
\]
\begin{defi} Bisimulations up-to context are defined as in Definition~\ref{def:bis}, but with $\Ctx(\R)$ instead of $\R$ in point (2).
\end{defi}
By virtue of the general theory in~\cite{DBLP:journals/acta/BonchiPPR17}, one has that $\Ctx$ is a \emph{sound} up-to technique, that is $x \eqthree y$ iff there exists a bisimulation up-to context $\R$ such that $x \, \R \, y$.
Actually, the theory in~\cite{DBLP:journals/acta/BonchiPPR17} guarantees a stronger property known as compatibility~\cite{San98MFCS,pous:aplas07:clut,SP09b}. Intuitively, this means that the technique is sound and it can be safely combined with other compatible up-to techniques. We refer the interested reader to~\cite{SP09b} for a detailed introduction to compatible up-to techniques.

We conclude with an example illustrating a finite bisimulation up-to context witnessing that the states $x$ and $y$ from  Figure~\ref{fig:examplesys} are in $\eqmay$.

\begin{exa}\label{ex:uptos}
Consider the \pss depicted in Figure~\ref{fig:examplesys}.
We have seen in Example~\ref{secapp:example} how to prove that $x \eqmay y$ by exhibiting a bisimulation on $(\TCSB X, \langle \barob^{\sharp}, \bartB^{\sharp} \rangle )$ relating them. However, due to the presence of cycles, the determinization of the NPLTS is infinite and the bisimulation relation contains infinitely many pairs.

With bisimulations up-to, only few pairs are necessary.
Indeed, we prove that the relation
\begin{align*}
\R = \{
&(x,y),
(x_{1}, y_{1}),
(x_{3}, y_{4}), \\
&
(x_{3} \pplus{\frac 1 2} x_{2},
(y_{4} \pplus{\frac 1 2} y_{2}) \cplus ((y_{2} \pplus{\frac 1 2} y_{4}) \pplus{\frac 1 2} y_{3}))
%((\onev \pplus{\frac 1 2} z) \pplus{\frac 1 2} x_{2},
%((\onev \pplus{\frac 1 2} z) \pplus{\frac 1 2} y_{2}) \cplus ((\onev \pplus{\frac 1 2} y_{2}) \pplus{\frac 1 2} y_{3}))
 \}
 \end{align*}
is a bisimulation up-to context.
First, note that the observation is trivially the same for all pairs in the relation, since $\barob^{\sharp}(S)=1$ for all $S$ in the relation.
heck that the clauses of bisimulation up-to context on the transitions are satisfied. Consider the first pair.
In $(\TCSB X, \langle \barob^{\sharp}, \bartB^{\sharp} \rangle )$, we have
\begin{align*}
x &\ttrel a x_{1} \cplus (x_{3} \pplus{\frac 1 2} x_{2})\\
y &\ttrel a y_{1} \cplus (y_{4} \pplus{\frac 1 2} y_{2}) \cplus ((y_{2} \pplus{\frac 1 2} y_{4}) \pplus{\frac 1 2} y_{3})
\end{align*}
The reached states are in $\Ctx(\R)$ by the second and fourth pairs of $\R$.
For any action $a' \neq a$, we have
$x\ttrel {a'} \onev$,  $y \ttrel {a'} \onev$
and $\onev \,\Ctx(\R)\,\onev$.

The second and the third pairs can be checked in a similar way.
For the fourth pair, we have
\[ x_{3} \pplus{\frac 1 2} x_{2} \ttrel b \onev \pplus{\frac 1 2} x_{3}\]

\[
(y_{4} \pplus{\frac 1 2} y_{2}) \cplus ((y_{2} \pplus{\frac 1 2} y_{4}) \pplus{\frac 1 2} y_{3})
\ttrel b (\onev  \pplus{\frac 1 2} y_{4}) \cplus ((\onev \pplus{\frac 1 2} y_{4}) \pplus{\frac 1 2} \onev)
\]
We observe that \begin{align*}
&(\onev  \pplus{\frac 1 2} y_{4}) \cplus ((\onev \pplus{\frac 1 2} y_{4}) \pplus{\frac 1 2} \onev)\\
&\stackrel{(B)}{=} (\onev  \pplus{\frac 1 2} y_{4}) \cplus ((\onev \pplus{\frac 1 2} y_{4})  \pplus{\frac 1 2} \onev) \cplus \onev \\
&\stackrel{(C)}{=}(\onev  \pplus{\frac 1 2} y_{4}) \cplus \onev\\
&\stackrel{(B)}{=} \onev  \pplus{\frac 1 2} y_{4}
\end{align*}
and we conclude by $\onev \pplus{\frac 1 2} x_{3} \,\Ctx(\R)\, \onev \pplus{\frac 1 2} y_{4}$.
The cases for $a$ and $c$ are simpler.
\end{exa}

\section{From the global to the local perspective}\label{sec:resolutions}

Usually trace semantics for NPLTS is defined in terms of \emph{schedulers}, or \emph{resolutions}: intuitively, a scheduler resolves the nondeterminism by choosing, at each step of the execution of an NPLTS, one of its possible transitions; the transition systems resulting from these choices are called resolutions.

This perspective on trace semantics is somehow opposed to ours, where the generalised determinisation keeps track of all possible executions at once. In this sense, the determinisation provides a perspective which is \emph{global}, opposite to those of resolutions that are \emph{local}. In this section, we show that our semantics can be characterised through such local views, by means of resolutions, defined as follows.

\begin{defi}\label{def:res}
Let $t\colon X \to (\Pow\mathcal{D}X)^A$ be an NPLTS\@. A \emph{(randomized) resolution} for $t$ is a triple $\R=(Y,\corr ,r)$ where $Y$ is a set of states, $\corr \colon Y \to X$ is \emph{the correspondence function}, and $r\colon Y \to (\mathcal{D}Y+1)^A$ is an RPLTS such that for all $y\in Y$ and $a\in A$,
\begin{enumerate}
\item $r(y)(a)= \star$ iff $t(\corr (y))(a)= \star$,
\item if $r(y)(a)\neq \star$ then $\mathcal{D}(\corr ) (r(y)(a)) \in \convex (t(\corr  (y))(a))$.
\end{enumerate}
\end{defi}

\noindent
Intuitively, this means that a resolution of an NPLTS is built from the original system by discarding internal nondeterminism (the possibility to perform multiple transitions labelled with the same action) and in such a way that the structure of the original system is preserved.

\begin{exa}\label{ex:resolutions}
Consider the \pss  on the left of Figure~\ref{fig:examplesys}.
The RPLTSs $\mathcal{R}_1$ and $\mathcal{R}_2$  in Figure~\ref{fig:resolutions} are two resolutions for it, both having the identity as correspondence function. In the resolution $\mathcal{R}_1$, the nondeterministic choice of $x$ is resolved by choosing the leftmost $a$-transition. The resolution $\mathcal{R}_2$ is obtained by taking a convex combination of the two distributions $\dirac {x_{1}}$ and $\Delta_{2}$, assigning one half probability to each of them.
\end{exa}

\begin{figure}
\begin{center}
\begin{tikzpicture}[thick%, font=\fontsize{7}{7}\selectfont
]

\matrix[matrix of nodes, row sep= 0.5cm, column sep=0cm,ampersand replacement=\&]
{
					\&\node (x) {$x$};		\\
%\&\node (dirx1) {$\dirac{x_{1}}$};	\\
						\&\node (x1) {$x_{1}$};		\\
						  	\&\node (d1) {$\Delta_{1}$};\\
  	\&							\& \node (z) {${x_{3}}$};	\\
	};
\draw[-latex] (x) to node[left] {$a$} (x1);
\draw[-latex] (x1) to node[left] {$b$} (d1);
%\draw[dotted,->] (dirx1) to node[left] {$1$}  (x1);
\draw[dotted,->] (d1) to node[left] {$\frac 1 2$}  (z);
\draw[dotted,->] (d1) to[bend left=70] node[left] {$\frac 1 2$}  (x);

\begin{scope}[xshift=4cm]

\matrix[matrix of nodes, row sep= 0.5cm, column sep=0.1cm,ampersand replacement=\&]
{
	\&							\&\node (x) {$x$};		\\
%\&\node (d) {$\Delta$}; \&  \\
	\&	\&	\node (d3) {$\Delta_{3}$}; 					\&	\& \\
	\&\node (x1) {$x_{1}$};			\& 						\&						\&\&\node (x2) {$x_{2}$}; 	\& \\
  	\&\node (d1) {$\Delta_{1}$};		\& 						\& \node (z) {${x_{3}}$};	\&				 	\& \\
  	\&							\&	\\
	};
\draw[-latex] (x) to node[left] {$a$} (d3);
\draw[-latex] (x1) to node[left] {$b$} (d1);
\draw[-latex] (x2) to node[right] {$b$} (z);
\draw[-latex] (x2) to[bend right=50] node[right] {$c$} (x);

\draw[dotted,->] (d3) to node[above] {$\frac 1 2$}  (x1);
\draw[dotted,->] (d1) to node[below] {$\frac 1 2$}  (z);
%\draw[dotted,->] (dirz) to node[right] {$1$}  (z);
\draw[dotted,->] (d3) to node[left] {$\frac 1 4$}  (z);
\draw[dotted,->] (d3) to node[above] {$\frac 1 4$}  (x2);
\draw[dotted,->] (d1) to[bend left=70] node[left] {$\frac 1 2$}  (x);
\end{scope}

\begin{scope}[xshift=8cm]

\matrix[matrix of nodes, row sep= 0.5cm, column sep=0cm,ampersand replacement=\&]
{
	\&			\&				\&\node (x) {$x$};		\\
%\&\node (d) {$\Delta$}; \&  \\
	\&	\&\&	\node (d3) {$\Delta_{2}$}; 					\&	\& \\
	\&			\& 						\&						\&\node (x2) {$x_{2}$}; 	\& \\
  %	\&		\& 					\&	 	\&				 	\& \node (dirx4) {$\dirac{x4}$};\\
  	\&							\& 	\&\node (z) {${x_{3}}$};		\&		 	\& \node (x4) {${x_{4}}$};\\

 % 	\&							\&\&\&\&\node (dirx1) {$\dirac{x_{1}}$};	\\
	\&							\&\&	\&		\&\node (x1) {$x_{1}$};		\\
	\&							\& \&	  \&	\&\node (d1) {$\Delta_{1}$};\\
	};

\draw[-latex] (x) to node[left] {$a$} (d3);
\draw[-latex] (x1) to node[left] {$b$} (d1);
\draw[-latex] (x2) to node[left] {$b$} (z);
\draw[-latex] (x2) to node[right] {$c$} (x4);
\draw[-latex] (x4) to node[left] {$a$} (x1);

\draw[dotted,->] (d1) to node[left] {$\frac 1 2$}  (z);
\draw[dotted,->] (d3) to node[left] {$\frac 1 2$}  (z);
\draw[dotted,->] (d3) to node[above] {$\frac 1 2$}  (x2);
\draw[dotted,->] (d1) to[bend right=50] node[right] {$\frac 1 2$}  (x);
\end{scope}
\end{tikzpicture}
\caption{The resolutions $\mathcal{R}_{1}$ (left), $\mathcal{R}_{2}$ (center), and $\mathcal{R}_3$ (right).}\label{fig:resolutions}
\end{center}
\end{figure}
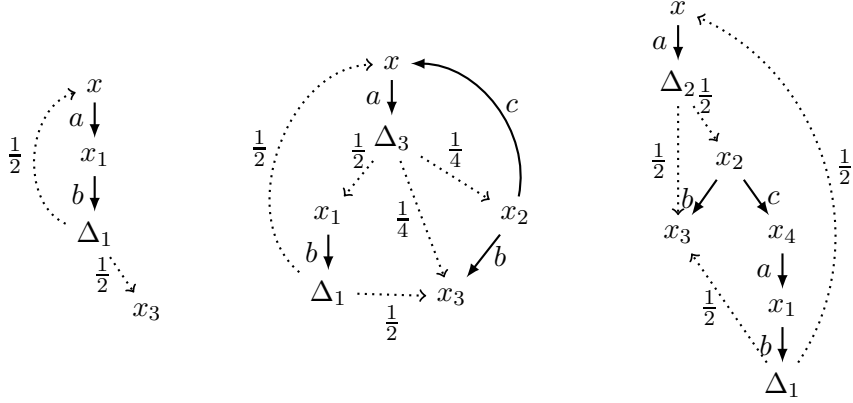

The reason why we take arbitrary $\corr$ functions, rather than just injective ones, is that the original NPLTS might contain cycles, in which case we want to allow the resolution to take different choices at different times.

\begin{exa}\label{ex:resolutionsappendix}
In order to understand how a resolution allows to resolve differently nondeterministic choices at different times, when cycles occur in the original system, consider the \pss  on the left of Figure~\ref{fig:examplesys} and its resolution $\mathcal{R}_{3}$ in
Figure~\ref{fig:resolutions}. In the latter, the state space is enlarged with state $x_{4}$, which is mapped to $x$ by the correspondence function. On the remaining states, the correspondence function is the identity over $X$. In this resolution, $x$ first chooses the right-hand transition of the original NPLTS, and at the next cycle iteration, represented by $x_{4}$, the left-hand transition is chosen. Observe that $\pprob_{\mathcal{R}_{3}}(x)(abab)=0$.
\end{exa}

Given a resolution $\mathcal{R}=(Y,\corr ,r)$, we define the function $\pprob_{\mathcal{R}} \colon Y \to [0,1]^{A^*}$ inductively for $y\in Y$ and $w\in A^*$ as
{{\begin{align*}
&\pprob_{\mathcal{R}}(y)(\varepsilon) &=&\quad1\text{;}\\
&\pprob_{\mathcal{R}}(y)(aw) &=
& \quad\begin{cases}
               0               & \text{ if } r(y)(a)=\star\text{;}\\
               \sum_{y'\in \supp(\Delta)} \Delta(y') \cdot \pprob_{\mathcal{R}}(y')(w) & \text{ if } r(y)(a)= \Delta \text{.}\\
           \end{cases}
\end{align*}}}
Intuitively, for all states $y\in Y$,  $\pprob_{\mathcal{R}}(y)(w)$ gives the probability of $y$ performing the trace $w$. For instance, in the resolutions in Figure~\ref{fig:resolutions},
$\pprob_{\mathcal{R}_{1}}(x)(abab)=\frac 1 2$ and  $\pprob_{\mathcal{R}_{2}}(x)(abab)=\frac 3 {16}$.

Now, given an NPLTS $(X,t)$, define $\bbressup{\cdot} \colon  X \to [0,1]^{A^*}$ by,  for $x\in X$ and $w\in A^*$,
\begin{align*}
\bbressup{x} (w) = \bigsqcup \{\pprob_{\mathcal{R}}(y)(w) \, \mid \, \mathcal{R}=(Y,\corr ,r) \text{ is a resolution of }(X,t) \text{ and } \corr (y)=x  \}\text{.}
\end{align*}

Similarly, we define
\begin{align*}
\bbresinf{x} (w) = \bigsqcap \{\pprob_{\mathcal{R}}(y)(w) \, \mid \, \mathcal{R}=(Y,\corr ,r) \text{ is a resolution of }(X,t) \text{ and } \corr (y)=x  \}\text{.}
\end{align*}

The following theorem states that the global view of trace semantics developed in Section~\ref{sec:maymust} coincides with the trace semantics defined locally via resolutions.

\begin{thm}[Global/local correspondence]\label{thm:correspondence}
Let $(X,t)$ be an NPLTS\@. For all $x\in X$ and $w\in A^*$, it holds that
\[\bb{x}(w)= [\, \bbresinf{x} (w), \bbressup{x} (w)\,]\text{.}\]
\end{thm}

\begin{cor}\label{cor:correspondencemaymust}
Let $(X,t)$ be an NPLTS\@. For all $x\in X$ and $w\in A^*$,
$\bbmay{x}(w)= \bbressup{x} (w) $ and $ \bbmust{x}(w)= \bbresinf{x} (w)$.
\end{cor}

Theorem~\ref{thm:correspondence} and Corollary~\ref{cor:correspondencemaymust} provide a  characterisation of $\eqmaymust$, $\eqmay$ and $\eqmust$ in terms of resolutions. Their proofs are presented in Section~\ref{sec:appthmcorr}.
Moreover, we show in Section~\ref{sec:correspondence_may} that $\eqmay$ coincides with the randomized $\sqcup$-trace equivalence investigated in~\cite{Cast18} and inspired by~\cite{BDL14a,BDL14b}.

\subsection{Proof of the global/local correspondence theorems}\label{sec:appthmcorr}

Given a resolution $\mathcal{R}=(Y,\corr ,r)$, we define the function $\reachres_{\mathcal{R}} \colon Y \to (\mathcal{D}(Y+1))^{A^*}$ inductively as
\begin{align*}
&\reachres_{\mathcal{R}}(y)(\varepsilon) &=&\quad \delta_y \text{;}\\
&\reachres_{\mathcal{R}}(y)(aw) &=& \quad
 \begin{cases}
               \delta_\star               & \text{ if } r(y)(a)=\star\text{;}\\
               \sum_{y'\in \supp(\Delta)} \Delta(y') \cdot \reachres_{\mathcal{R}}(y')(w) & \text{ if } r(y)(a)= \Delta \text{.}\\
           \end{cases}
           \end{align*}
Intuitively, this assigns to each state $y\in Y$ and word $w\in A^*$ a subdistribution over $Y$, which is the state of the determinised system that $y$ reaches via $w$.

Let $o'^{\sharp} \colon \mathcal{D}(Y+1) \to [0,1]$ be the function assigning to a subdistribution $\Delta$ its total mass, namely $1-\Delta(\star)$. More formally, this is defined inductively as
\[
o'^{\sharp}(\Delta)=\begin{cases}
              0              & \text{ if } \Delta=   \delta_\star \text{;}\\
              1              & \text{ if } \Delta=   \delta_y \text{ for } y\in Y\text{;}\\
             o'^{\sharp}(\Delta_1) +_p o'^{\sharp}(\Delta_2) & \text{ if } \Delta=  \Delta_1 \pplus{p} \Delta_2 \text{.}
             \end{cases}
\]

\begin{lem}\label{lemma:oprob}
$o'^\sharp \after \reachres_{\mathcal{R}}= \pprob_{\mathcal{R}}$.
\end{lem}
\begin{proof}
We prove that $o'^\sharp ( \reachres_{\mathcal{R}}(y)(w))= \pprob_{\mathcal{R}}(y)(w)$ for all $y\in Y$ and $w\in A^*$. The proof proceeds by induction on $w$.

Base case: $w=\varepsilon$.
\[\pprob_{\mathcal{R}}(y)(\varepsilon)= 1 = o'^{\sharp}(\delta_y) = o'^{\sharp}(\reachres_{\mathcal{R}}(y)(\varepsilon))\]

Inductive case: $w=aw'$. If $r(y)(a)=\star$, then
\[\pprob_{\mathcal{R}}(y)(aw')= 0 = o'^{\sharp}(\delta_\star) = o'^{\sharp}(\reachres_{\mathcal{R}}(y)(aw'))\text{.}\]

If $r(y)(a)=\Delta$, then
\begin{align*}
\pprob_{\mathcal{R}}(y)(aw') & =
 \sum_{y'\in \supp(\Delta)} \Delta(y') \cdot \pprob_{\mathcal{R}}(y')(w') \tag{definition}\\
& =  \sum_{y'\in \supp(\Delta)} \Delta(y') \cdot o'^\sharp ( \reachres_{\mathcal{R}}(y')(w'))  \tag{IH} \\
& =  o'^\sharp( \sum_{y'\in \supp(\Delta)} \Delta(y') \cdot  ( \reachres_{\mathcal{R}}(y')(w')))  \tag{$o'^{\sharp}$ hom.} \\ % chktex 35
& =  o'^\sharp (\pprob_{\mathcal{R}}(y)(aw'))  & \tag*{(definition) \qedhere}
\end{align*}
\end{proof}

\medskip

Given an NPLTS $(X,t)$, we define the function $\reach \colon X \to (C(X+1))^{A^*}$ inductively as
\begin{align*}
&\reach(x)(\varepsilon) &=& \{\delta_x  \}\text{;}\\
&\reach(x)(aw) &=
&\begin{cases}
               \{ \delta_\star \}               & \text{ if } t(x)(a)=\star\text{;}\\
              \underset{_{\Delta\in \convex (S)}}{\displaystyle\bigoplus}\displaystyle\sum_{x'\in \supp(\Delta)} \Delta(x') \cdot \reach(x')(w)& \text{ if } t(x)(a)=S{.}\\
           \end{cases}
\end{align*}

For each NPLTS $(X,t)$, we have a function $\bbangle{\cdot}\colon C(X+1)\to C(X+1)^{A^*}$ defined for all $S\in C(X+1)$ and $w\in A^*$ as
\[\begin{array}{lcl}
\bbangle{S}(\varepsilon) & = & S\text{;}\\
\bbangle{S}(aw) & = & \bbangle{\bart^{\sharp}(S)(a)}(w)\text{.}\\
\end{array}
\]

The following property follows directly from the inductive definitions of $\bb{\cdot}$ and $\bbangle{\cdot}$.

\begin{lem}\label{lemmatrivial}
$\bb{\cdot}=\barO^{\sharp}\after \bbangle{\cdot}$ \qedhere
\end{lem}

\begin{lem}\label{lemmasophisticated}
$\bbangle{\cdot} \after \eta =\reach$
\end{lem}
\begin{proof}
The proof goes by induction on $w\in A^*$.

Base case: if $w=\varepsilon$, then $\reach(x)(\varepsilon)=\{\delta_x  \} = \eta(x)= \bbangle{\eta(x)}(\varepsilon)$.

Inductive case: $w=aw'$. If $t(x)(a)=\star$, then $\bbangle{\eta(x)}(aw')=\bbangle{\bart^{\sharp}(\{\delta_x\})(a)}(w')= \bbangle{\{\delta_{\star}\}}(w')= \{\delta_{\star}\}= \reach(x)(aw')$.

If $t(x)(a)=S$, then $\reach(x)(aw)= \bigoplus_{\Delta\in \convex (S)}\sum_{x'\in \supp(\Delta)} \Delta(x') \cdot \reach(x')(w)$. By induction hypothesis, the latter is equal to
$\bigoplus_{\Delta\in \convex (S)}\sum_{x'\in \supp(\Delta)} \Delta(x') \cdot \bbangle{\eta(x')}(w')$. Since $\bart^{\sharp}$ is a homomorphism of convex semilattices, then $\bbangle{\cdot}$ is a homomorphism of convex semilattices.
Hence, the latter is equal to
$\bbangle{\bigoplus_{\Delta\in \convex (S)}\sum_{x'\in \supp(\Delta)} \Delta(x') \cdot \eta(x')}(w)$
that is $\bbangle{\convex(S)}(w')=\bbangle{\bart^{\sharp}(\{\delta_x\})(a)}(w')=\bbangle{\eta(x)}(aw')$.
\end{proof}

\begin{prop}\label{prop:bb-bbangle}
$ \bb{\cdot} \after \eta = \barO^{\sharp}\after \reach$
\end{prop}
\begin{proof}
By Lemma~\ref{lemmatrivial}, $\bb{\cdot} \after \eta =\barO^{\sharp}\after \bbangle{\cdot} \after \eta$. By Lemma~\ref{lemmasophisticated}, $\barO^{\sharp}\after \bbangle{\cdot} \after \eta = \barO^{\sharp}\after \reach$.
\end{proof}

\begin{prop}\label{prop:correspondence}
Let $(X,t)$ be an NPLTS and let $\mathcal{R}=(Y,\corr ,r)$ be one of its resolutions. Let $x\in X$ and $y\in Y$  such that $\corr (y)=x$. For all $w\in A^*$, \[\mathcal{D}(\corr +1)(\reachres_\mathcal{R}(y)(w))\in \reach(x)(w)\text{.}\]
\end{prop}
\begin{proof}
By induction on the structure of $w$.
If $w=\epsilon$ then
\begin{align*}
\mathcal{D}(\corr +1)(\reachres_\mathcal{R}(y)(\epsilon))&= \mathcal{D}(\corr +1)(\delta_{y})\\
&= \delta_{x} \\
&\in \{\delta_{x}\}\\
&= \reach(x)(\epsilon)
\end{align*}
If $w=aw'$ and $t (x) (a)=\star$, then $r (y) (a)=\star$, and
\begin{align*}
\mathcal{D}(\corr +1)(\reachres_\mathcal{R}(y)(aw'))&= \mathcal{D}(\corr +1)(\delta_{\star})\\
&= \delta_{\star} \\
&\in \{\delta_{\star}\}\\
&= \reach(x)(aw')
\end{align*}
If $t (x) (a)\neq \star$ then we have $r (y) (a)\neq \star$. Let $r (y) (a)=\Delta \in\mathcal D (Y)$. We have:
\begin{align*}
&\mathcal{D}(\corr +1)(\reachres_\mathcal{R}(y)(aw'))\\
&= \mathcal{D}(\corr +1)(\sum_{y'\in \supp(\Delta)} \Delta(y') \cdot \reachres_{\mathcal{R}}(y')(w'))\\
&= \sum_{y'\in \supp(\Delta)} \Delta(y') \cdot \mathcal{D}(\corr +1)(\reachres_{\mathcal{R}}(y')(w'))\\
\end{align*}
By the inductive hypothesis, for each $y'$ we have \[\mathcal{D}(\corr +1)(\reachres_{\mathcal{R}}(y')(w'))\in \reach (\corr  (y')) (w') \text{.}\]
Hence, by the definition of Minkowski sum,
\begin{eqnarray*}
\sum_{y'\in \supp(\Delta)} \Delta(y') \cdot \mathcal{D}(\corr +1)(\reachres_{\mathcal{R}}(y')(w'))
&\in& \msum_{y'\in \supp(\Delta)} \Delta(y') \cdot \reach (\corr  (y')) (w')
\end{eqnarray*}
Since $\R$ is a resolution, there is a $\Delta'\in \convex (t(x)(a))$ such that $\mathcal D (\corr ) (\Delta) = \Delta'$. The latter means that $\Delta'(x')= \sum_{\{y' \in \supp(\Delta) | \corr (y')=x'\}}  \Delta(y')$, and thus:
\begin{eqnarray*}
\msum_{y'\in \supp(\Delta)} \Delta(y') \cdot \reach (\corr  (y')) (w')
&=& \msum_{x' \in \supp(\Delta')} \Delta'(x') \cdot \reach (x') (w')
\end{eqnarray*}
as easily follows from the axioms of convex algebras.  We can then conclude by the definition of $\reach (x)(aw')$
\begin{align*}
\msum_{x' \in \supp(\Delta')} \Delta'(x') \cdot \reach (x') (w')
&\subseteq  \bigoplus_{\Delta'\in \convex (t(a)(x))}\;\msum_{x'\in \supp(\Delta')} \Delta'(x') \cdot \reach(x')(w')\\
&= \reach(x)(aw'). \qedhere
\end{align*}
\end{proof}

\begin{prop}\label{prop:correspondence2}
Let $(X,t)$ be an NPLTS\@. For all $x\in X$ and $w\in A^*$, if $\Delta \in \reach(x)(w)$ then there exists a resolution $\mathcal{R}=(Y,\corr ,r)$ and a state $y \in Y$ such that
\begin{enumerate}
\item[(1)] $\corr (y)=x$ and
\item[(2)] $\mathcal{D}(\corr +1)(\reachres_\mathcal{R}(y)(w))= \Delta$.
\end{enumerate}
\end{prop}
\begin{proof}
The proof proceeds by induction on $w \in A^*$.

In the base case $w=\varepsilon$. For all $x\in X$ and $a\in A$ such that $t(x)(a)\neq \star$, we can choose one distribution $\Delta_{x,a}\in t(x)(a)$. Then, we take $\mathcal{R}=(X,id_X,r)$  where $r\colon X \to (\mathcal{D}X+1)^A$ is defined for all $x\in X$ and $a\in A$ as
\[r(x)(a) = \begin{cases}
              \star             & \text{ if } t(x)(a)=\star\text{;}\\
              \Delta_{x,a} & \text{ otherwise.}
\end{cases}\]
By construction $\mathcal{R}$ is a resolution.  Then we take $x$ as the selected state $y$ of the resolution $\mathcal{R}$. Since the correspondence function is $id_X$, $(1)$ is immediately satisfied. Now, by definition, $\reachres_{\mathcal{R}}(x)(\varepsilon)=\delta_x$ and $\reach(x)(\varepsilon)=\{\delta_x\}$. We conclude by observing that  $\mathcal{D}(id_X+1)(\delta_x)=\delta_x\in \{\delta_x\}=\reach(x)(\varepsilon)$.

\medskip

In the inductive case $w=aw'$. Now we have two cases to consider: either $t(x)(a)=\star$ or $t(x)(a)=S$ for $S\in \Powne\mathcal{D}(X)$.

Assume $t(x)(a)=\star$. Then $\reach(x)(aw')=\{\delta_\star\}$.
Let $\mathcal{R}=(X,id_X,r)$ be the resolution defined as in the base case, and take $x$ as the selected state $y$ of the resolution $\mathcal{R}$. Since the correspondence function is $id_X$, $(1)$ is immediately satisfied. Since $\mathcal R$ is a resolution, $t(x)(a)=\star$ implies $r(x)(a)=\star$. Hence, $\reachres_{\mathcal{R}}(x)(a)=\delta_\star$ and $\reach(x)(a)=\{\delta_\star\}$. We conclude by  $\mathcal{D}(id_X+1)(\delta_\star)=\delta_{\star}\in \{\delta_\star\}=\reach(x)(a)$.

Assume $t(x)(a)=S$.
Then \[\reach(x)(aw')= \bigoplus_{\Delta'\in \convex (S)}\sum_{x'\in \supp(\Delta')} \Delta'(x') \cdot \reach(x')(w')\text{.}\]
By Lemma~\ref{lem:f-sharp},
it holds that
\[\reach(x)(aw') =  \bigcup_{\Delta'\in \convex (S)}\sum_{x'\in \supp(\Delta')} \Delta'(x') \cdot \reach(x')(w')\text{.}\]
Hence, $\Delta\in \reach(x)(aw')$ if and only if
there exists a $\Delta' \in \convex (S)$ such that
\[\Delta \in \sum_{x'\in \supp(\Delta')} \Delta'(x') \cdot \reach(x')(w')\text{.} \]
This is in turn equivalent to saying (by the definition of Minkowski sum) that for every $x' \in \supp(\Delta')$ there exists a $\Delta'_{x'}\in \reach(x')(w')$ such that
\[\Delta=\sum_{x'\in \supp(\Delta')} \Delta'(x') \cdot \Delta'_{x'}\text{.} \]
We can now use the induction hypothesis on $\Delta'_{x'}\in \reach(x')(w')$: for each $\Delta'_{x'}\in \reach(x')(w')$ there exists a resolution $\mathcal{R}_{x'}=(Y_{x'},\corr_{x'}, r_{x'})$ and a $y_{x'}\in Y_{x'}$ such that
\begin{enumerate}
\item[(c)] $\corr_{x'}(y_{x'})=x'$ and
\item[(d)] $\mathcal{D}(\corr_{x'}+1)(\reachres_{\mathcal{R}_{{x'}}}(y_{x'})(w))=\Delta'_{x'}$.
\end{enumerate}
Now we construct the coproduct of all the resolutions $\mathcal{R}_{{x'}}$. Take $Z$ to be the disjoint union of all the $Y_{x'}$ and define $\corr_Z \colon Z \to X$ as
$\corr_Z(z)=\corr_{x'}(z)$ if $z\in Y_{x'}$. Similarly, we define $r_Z \colon Z \to (\mathcal{D}Z+1)^A$ as $r_Z(z)=r_{x'}(z)$ if $z\in Y_{x'}$. By construction, $\mathcal{R}_Z=(Z,\corr_Z,r_Z)$ is a resolution of $(X,t)$.

Let $\mathcal{R'}=(X,id_X,r')$ be a resolution defined as in the base case, i.e., by arbitrarily choosing a distribution $\Delta_{x,a} \in t(x)(a)$, for any $x$ and $a$, as value of $r'(x)(a)$, whenever $t(x)(a) \neq \star$.
We define the resolution $\R=(Y,\corr,r)$ needed to conclude the proof as follows:
the state space is $Y=Z+X+\{y\}$, namely the disjoint union of $Z$, of $X$, and of the singleton containing a fresh state $y$;  the correspondence function $\corr \colon Y \to X$ and the transition function $r\colon Y \to (\mathcal{D}Y+1)^A$ are defined for all $u\in Y$ as
\[ \corr(u)= \begin{cases}
		\corr_{Z}(u) & \text{ if } u\in Z\text{,}\\
                    id_X(u)      & \text{ if } u\in X\text{,}\\
                    x      & \text{ if } u=y\text{,}
\end{cases}
\]
\[
r(u)(b) = \begin{cases}
		r_Z(u)(b) & \text{ if } u\in Z\text{,}\\
		r'(u)(b) & \text{ if } u\in X\text{,}\\
              \Delta_{x,b} & \text{ if } u=y, a \neq b, t(x)(b)\neq \star \text{,}\\
  \star & \text{ if } u=y, a \neq b, t(x)(b)= \star \text{,}\\
                 \Delta''      & \text{ if } u = y, \, a=b\\
\end{cases}\]
where $\Delta''$ is the distribution having as support the set of states $\{y_{x'}| x' \in \supp(\Delta)\} \subseteq Z$, and such that $\Delta''(y_{x'})=\Delta'(x')\text{.}$ Note that $\Delta''$ is a distribution, since $\Delta'$ is and since we are taking exactly one $y_{x'}$ for each $x'\in \supp(\Delta')$.

The fact that $\R$ is a resolution follows from $\R_{Z}$ and $\R'$ being resolutions and
$y$ respecting --- by construction --- the conditions of resolution: indeed $\corr (y)=x$, and \begin{itemize}
\item if $a \neq b$ and $t(x)(b)\neq \star$, then
$r(y)(b)= \Delta_{x,b}$ and
$\mathcal{D}(\corr )  (\Delta_{x,b})= \mathcal{D}(id_{X})  (\Delta_{x,b})= \Delta_{x,b} \in t(x)(b)$
\item if $a \neq b$ and $t(x)(b) = \star$, then $r(y)(b)= \star$;
\item if $a=b$, then $r(y)(b)= \Delta''$, and $\mathcal{D}(\corr ) (\Delta'')=  \Delta'$, with $\Delta' \in \convex (t(x)(a))$.
\end{itemize}

\noindent
To conclude the proof we only need to show that (1) and (2) hold. The former is trivially satisfied by definition of $\corr $. For (2), we derive:

\begin{align*}
\mathcal{D}(\corr +1)(\reachres_{\mathcal{R}}(y)(aw'))
&=  \mathcal{D}(\corr +1) (\sum_{y_{x'}\in \supp(\Delta'')}( \Delta''(y_{x'})\cdot \reachres_{\mathcal{R}}(y_{x'})(w')  )  )\\
&=\mathcal{D}(\corr +1) (\sum_{x'\in \supp(\Delta')}( \Delta'({x'})\cdot \reachres_{\mathcal{R}}(y_{x'})(w')  )  )\\
&=\sum_{x'\in \supp(\Delta')} \Delta'({x'})\cdot ( \mathcal{D}(\corr +1) ( \reachres_{\mathcal{R}}(y_{x'})(w')  ) ) \\
&=\sum_{x'\in \supp(\Delta')}\Delta'({x'})\cdot \Delta'_{x'} \\
&= \Delta \qedhere
\end{align*}

\end{proof}

\begin{proof}[Proof of Theorem~\ref{thm:correspondence}]
Before starting with the actual proof, we need the following elementary observation: for all $f\colon X \to Y$ and $\Delta\in \mathcal{D}(X+1)$, it holds that
\begin{equation}\label{eq:elemtaryobservation}
o'^{\sharp}(\mathcal{D}(f+1)(\Delta))= o'^{\sharp}(\Delta)\text{,}
\end{equation}
namely, the total mass is preserved by applying $\mathcal{D}(f+1)$.

Now, suppose that $\bb{\eta(x)}(w)=[p,q]$ for some $p,q\in [0,1]$ with $p\leq q$.
By Proposition~\ref{prop:bb-bbangle}, it holds that $\barO^{\sharp}(\reach(x) (w))=[p,q] $. By definition of $\barO^{\sharp}$ there exists $\Delta_{\min}, \Delta_{\max} \in \reach(x)(w)$ such that the total mass of $\Delta_{\min}=p$, the total mass of $\Delta_{\max}=q$ and for an arbitrary $\Delta \in \reach(x)$, its total mass is in between $p$ and $q$. In other words,
\begin{enumerate}
\item[(a)] $o'^{\sharp}(\Delta_{\min})=p$,
\item[(b)] $o'^{\sharp}(\Delta_{\max})=q$ and
\item[(c)] $p\leq o'^{\sharp}(\Delta)\leq q$ for all $\Delta \in \reach(x)$.
\end{enumerate}

\noindent
By Proposition~\ref{prop:correspondence}, for all resolutions $\R$, states $y$ such that $\corr (y)=x$, and distributions $\Delta'$ such that $\reachres_{\R}(y)(w)=\Delta'$, one has that $\mathcal{D}(\corr +1)(\Delta')\in \reach(x)(w)$. By (c), $p\leq o'^{\sharp}(\mathcal{D}(\corr +1)(\Delta')) \leq q$ and, by~\eqref{eq:elemtaryobservation}, $p \leq o'^{\sharp}(\Delta') \leq q$. This means $p \leq o'^{\sharp}(\reachres_{\R}(y)(w)) \leq q$ that, by Lemma~\ref{lemma:oprob}, coincides with $p \leq \pprob_{\R}(y)(w) \leq q$. This proves that $\bbresinf{x} (w)\geq p$ and $\bbressup{x} (w)\leq q$.

\medskip

We now prove that $\bbresinf{x} (w)\leq p$; the proof for $\bbressup{x} (w)\geq q$ is analogous.

By Proposition~\ref{prop:correspondence2}, there exist resolutions $\R$, a state $y$ and a distribution $\Delta''$ such that
\begin{enumerate}
\item[(d)] $\corr (y)=x$,
\item[(e)] $\mathcal{D}(\corr +1)(\Delta'')=\Delta_{\min}$,
\item[(f)] $\reachres_{\R}(y)(w)=\Delta''$.
\end{enumerate} By (e) and~\eqref{eq:elemtaryobservation}, one immediately has that $o'^{\sharp}(\Delta'')=o'^{\sharp}(\Delta_{\min})=p$. By (f), the latter means that $o'^{\sharp}(\reachres_{\R}(y)(w))=p$ that, by Lemma~\ref{lemma:oprob}, allows to conclude that $\pprob_{\R}(y)(w)=p$.
This proves that $\bbresinf{x} (w)\leq p$.
\end{proof}

\begin{proof}[Proof of Corollary~\ref{cor:correspondencemaymust}.]

Consider the monad morphism  $\quotientB\colon \TPCS\Rightarrow \TCSB$ quotienting  $\TPCS$ by $(B)$ (see Section~\ref{sec:maymust}), and let $\eta^{B},\mu^{B}$ respectively denote the unit and multiplication of the monad $\TCSB$.
Following Remark~\ref{rem:detinvariance} (by Theorem~\ref{thm:transfert}.1) we have
\[\bbmay{\cdot} \circ \eta^{B}= \bb{\cdot}_{B'}\circ \eta\]
where $\bb{\cdot}_{B'}: \TPCS X \to \TCSB 1^{A^*}$ is the semantic map induced by the determinisation of
$\langle \barob,\bart \rangle\colon X \to \TCSB 1 \times (\TPCS X)^A$ using the algebra $\mu^{B}_1 \after \quotientB_{\TCSB 1} \colon \TPCS\TCSB 1 \to \TCSB 1$.
The monad map $\quotientB\colon \TPCS\Rightarrow \TCSB$ gives a $\TPCS$-algebra morphism $\quotientB_{1}\colon (\TPCS 1, \mu_{1})  \to (\TCSB 1, \mu^{B}_{1} \circ \quotientB_{\TPCS 1})$.
By
 Theorem~\ref{thm:transfert}.2 we have
\[\bb{\cdot}_{B'}= {\quotientB_{1}}^{A^*} \circ \bb{\cdot}.\]
Hence, we derive
\[\bbmay{\cdot} \circ \eta^{B}= {\quotientB_{1}}^{A^*} \circ \bb{\cdot} \circ \eta.\]
For an interval $[p,q]$,
$\quotientB_{1} ([p,q])= q$.
Then by Theorem~\ref{thm:correspondence} we conclude
\begin{align*}
({\quotientB_{1}}^{A^*} \circ \bb{\eta(x)}) (w)
&= \quotientB_{1} (\bb{\eta(x)} (w)) \\
&= \quotientB_{1} ([\, \bbresinf{x} (w), \bbressup{x} (w)\,]) \\
&=\bbressup{x} (w)\text{.}
\end{align*}
To prove that $\bb{\eta(\cdot)}_T=\bbresinf{\cdot}$ we proceed in the same way, but taking the monad morphism $\quotientT\colon \TPCS\Rightarrow \TCST$ quotienting   $\TPCS$ by  $(T)$.
\end{proof}

\subsection{Coincidence with randomized \texorpdfstring{$\sqcup$}{supremum}-trace equivalence}%
\label{sec:correspondence_may}
Let $t\colon X \to (\Pow\mathcal{D} X)^A$ be an NPLTS\@. A \emph{fully probabilistic resolution} for $t$ is a triple $\R=(Y,\corr ,r)$ such that $Y$ is a set, $\corr\colon Y \to X$, and $r\colon Y \to (A \times \dset Y) +1$ such that for every $y\in Y$ and $a\in A$:

\qquad\qquad if $r(y)= \langle a, \Delta\rangle$ then $\mathcal{D}(\corr )(\Delta) \in \convex (t(\corr  (y))(a))$.

\medskip
While resolutions resolve only internal nondeterminism, fully probabilistic resolutions resolve both internal and external nondeterminism. Indeed, in a resolution a state can perform transitions with different labels, while in a fully probabilistic resolution a state can perform at most one transition. Moreover, a state $y$ in a fully probabilistic resolutions might not perform any transition (i.e., $r(y)=\star$), even if the corresponding state $\corr (y)$ may perform a transition (i.e., $t(\corr (y))(a)\neq\star$ for some $a$).

\begin{exa}\label{ex:fpresolutions}
As in Example~\ref{ex:resolutions}, consider the NPLTS on the left of Figure~\ref{fig:examplesys}.
The resolution $\mathcal{R}_1$ in Figure~\ref{fig:resolutions} is a fully probabilistic resolution, while $\mathcal{R}_2$ and $\mathcal{R}_3$ are not, since $x_{2}$ is allowed to perform more than one transition, even if labelled by different actions.

In Figure~\ref{fig:fpresolutions}, we show three examples of fully probabilistic resolutions of the same NPLTS\@. Note that neither of them is a resolution. In $\R_{1}$, state $x_{1}$ does not satisfy the first clause of the definition of resolution, since $x_{1}$ does not move while its corresponding state in the original NPLTS does. In $\R_{2}$ and $\R_{3}$, state $x_{2}$ respectively only performs a $b$-labelled transition and only performs a $c$-labelled transition. In a resolution, it should perform both.
\begin{figure}
\begin{center}
\begin{tikzpicture}[thick]

\matrix[matrix of nodes, row sep= 0.5cm, column sep=0cm,ampersand replacement=\&]
{
	\&							\&\node (x) {$x$};		\\
%	\&							\&\node (dirx1) {$\dirac{x_{1}}$};		\\
		\&							\&\node (x1) {${x_{1}}$};		\\
	};
\draw[-latex] (x) to node[left] {$a$} (x1);
%\draw[dotted] (dirx1) to node[left] {$1$} (x1);

\begin{scope}[xshift=2.6cm]

\matrix[matrix of nodes, row sep= 0.5cm, column sep=.0cm,ampersand replacement=\&]
{
	\&							\&\node (x) {$x$};		\\
%\&\node (d) {$\Delta$}; \&  \\
	\&	\&	\node (d3) {$\Delta_{3}$}; 					\&	\& \\
	\&\node (x1) {$x_{1}$};			\& 						\&						\&\node (x2) {$x_{2}$}; 	\& \\
  	\&\node (d1) {$\Delta_{1}$};		\& 						\& \node (z) {${x_{3}}$};	\&				 	\& \\
 % 	\&							\& 	\\
	};
\draw[-latex] (x) to node[left] {$a$} (d3);
\draw[-latex] (x1) to node[left] {$b$} (d1);
\draw[-latex] (x2) to node[right] {$b$} (z);

\draw[dotted,->] (d3) to node[above] {$\frac 1 2$}  (x1);
\draw[dotted,->] (d1) to node[below] {$\frac 1 2$}  (z);
%\draw[dotted,->] (dirz) to node[right] {$1$}  (z);
\draw[dotted,->] (d3) to node[left] {$\frac 1 4$}  (z);
\draw[dotted,->] (d3) to node[above] {$\frac 1 4$}  (x2);
\draw[dotted,->] (d1) to[bend left=70] node[left] {$\frac 1 2$}  (x);

\end{scope}

\begin{scope}[xshift=6.3cm]

\matrix[matrix of nodes, row sep= 0.5cm, column sep=0cm,ampersand replacement=\&]
{
	\&							\&\node (x) {$x$};		\\
%\&\node (d) {$\Delta$}; \&  \\
	\&	\&	\node (d3) {$\Delta_{3}$}; 					\&	\& \\
	\&\node (x1) {$x_{1}$};			\& 						\&						\& 	\&\node (x2) {$x_{2}$}; \\
  	\&\node (d1) {$\Delta_{1}$};		\& 						\& 	\&	\node (z) {${x_{3}}$};			 	\& \\
  	\&							\& 	\\
	};
\draw[-latex] (x) to node[left] {$a$} (d3);
\draw[-latex] (x1) to node[left] {$b$} (d1);
\draw[-latex] (x2) to[bend right=50] node[right] {$c$} (x);

\draw[dotted,->] (d3) to node[above] {$\frac 1 2$}  (x1);
\draw[dotted,->] (d1) to node[below] {$\frac 1 2$}  (z);
\draw[dotted,->] (d3) to node[left] {$\frac 1 4$}  (z);
\draw[dotted,->] (d3) to node[above] {$\frac 1 4$}  (x2);
\draw[dotted,->] (d1) to[bend left=70] node[left] {$\frac 1 2$}  (x);
\end{scope}

\end{tikzpicture}
\end{center}
\caption{Fully probabilistic resolutions ($\mathcal{R}_1$, $\mathcal{R}_2$, $\mathcal{R}_3$, from left to right)}\label{fig:fpresolutions}
\end{figure}
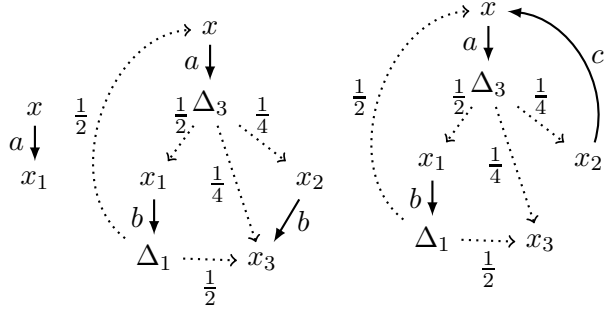

\end{exa}

As for resolutions, we can define $\pprob_{\mathcal{R}} \colon Y \to [0,1]^{A^*}$ for $\mathcal{R}=(Y,\corr ,r)$ a fully probabilistic resolution inductively for all $y\in Y$ and all $w\in A^*$ as
\begin{align*}
&\pprob_{\mathcal{R}}(y)(\varepsilon) &=&\quad 1\text{;}\\
&\pprob_{\mathcal{R}}(y)(aw) &=&\quad \begin{cases}
                              \sum_{y'\in \supp(\Delta)} \Delta(y') \cdot \pprob_{\mathcal{R}}(y')(w) & \text{ if } r(y)=\langle a, \Delta \rangle \text{,}\\
                              0               & \text{ otherwise.}\\
           \end{cases}
\end{align*}

Given an NPLTS $(X,t)$, we define  for $x\in X$ and $w\in A^*$:
\begin{align*}
\bbressupfp{x} (w) = \bigsqcup \{\pprob_{\mathcal{R}}(y)(w) \, \mid \, &\mathcal{R}=(Y,\corr ,r) \text{ is a fully probabilistic}\\
&\text{ resolution of }(X,t) \text{ and } \corr (y)=x  \}\text{.}
\end{align*}

In~\cite{Cast18} (following~\cite{BDL14a,BDL14b}), two states $x$ and $y$ are defined to be \emph{randomized $\sqcup$-trace equivalent} whenever $\bbressupfp{x} (w)=\bbressupfp{y} (w)$, for all $w\in A^*$.\footnote{Actually,~\cite{Cast18,BDL14a,BDL14b} use a notion of resolution which is equal to our fully-probabilistic resolution modulo a tiny modification due to a mistake in~\cite{BDL14a,BDL14b}, as confirmed by the authors in a personal communication.} The following proposition guarantees that such equivalence coincides with $\eqmay$.

\begin{prop}\label{prop:correpsondencefp}
Let $(X,t)$ be an NPLTS\@. For all $x\in X$ and $w\in A^*$, it holds that
$\bbmay {x}(w) = \bbressup{x} (w) = \bbressupfp{x} (w)$.
\end{prop}

\begin{proof}
We first prove that $\bbressup{x} (w) \leq \bbressupfp{x} (w)$.

Let $\R= (Y,\corr ,r)$ be a resolution of $(X,t)$, $x\in X$, and $w\in A^*$. Let  $y\in Y$ such that $\corr (y)=x$. We show that there exists a fully probabilistic resolution $\R'$ of $(X,t)$ with a state $z$ such that $z$ is mapped by the correspondence function of $\R'$ to $x$ and such that $\pprob_{\R}(y)(w)=\pprob_{\R' }(z)(w)$.

We define $\R'=(Y \times A^{*},\corr ',r')$ as follows.
The correspondece function $\corr '\colon Y\times A^* \to Y$ is $\corr  \after \pi_1$, namely $\corr '(y,w')=\corr (y)$ for all $w'\in A^{*}$. To define $r'$, we use the notation $\Delta_{w'}\in \dset (Y\times A^*)$ to denote, for all $\Delta\in \dset (Y)$ and $w'\in A^*$, the distribution over $Y \times A^{*}$ given as
\[\begin{array}{rcl}
\Delta_{w'}(y,w'')& = &
\begin{cases}
               \Delta(y) & \text{ if } w'=w'' \text{,} \\
               0 	& \text{ otherwise.}
           \end{cases}
\end{array}
\]
Now $r' \colon Y \times A^{*} \to (A \times \dset (Y \times A^{*})) +1$ is defined as:
\[\begin{array}{rcl}
r'(y,\epsilon)& = & \star\\
r'(y,aw')& = &
\begin{cases}
               \langle a, \Delta_{w'} \rangle & \text{ if } r(y)(a)= \Delta \neq \star                \\
               \star 	& \text{ otherwise.}
           \end{cases}
\end{array}
\]

First, it is necessary to observe that $\R'$ is indeed a fully probabilistic system, that is, for every state there is at most one transition that can be performed.
Indeed, by taking as set of states of $\R'$ the set $Y\times A^*$, we guarantee that despite a state $y$ in $\R$  might perform different transitions reaching distributions over states (one for each label), only one of these transitions is actually performed by a corresponding state $(y,aw')$ in $\R'$, namely the transition with label $a$. This allows us to move from a (reactive) resolution $\R$ to a fully probabilistic resolution $\R'$, while preserving the probability of performing traces. As we will prove below, for all $w$ and $y$ it holds $\pprob_{\R}(y)(w)=\pprob_{\R' }(y,w)(w)$. To illustrate this with an example, suppose the resolution $\R$ is the RPLTS with two states $y,y'$ and with $r (y)(a) =\dirac y$ and $r(y)( b)= \dirac {y'}$. Take the trace $ab$, for which we have $\pprob_{\R}(y)(ab)=1$.
Then in $\R'$ we have $r' (y,ab)=\langle a, \dirac {(y,b)}\rangle$ and $r' (y,b)=\langle b, \dirac {(y',\epsilon)}\rangle$, which gives $\pprob_{\R'}(y,ab)(ab)=1$.

\medskip

We proceed by proving that $\R'$ is a fully probabilistic resolution of $(X,t)$.
Suppose that $r'(y,w')\neq \star$. Then $w'=aw''$,
and $r'(y,w')=\langle a, \Delta_{w''}\rangle$ with $\Delta=r(y)(a)$. Hence, $\dset (\corr ') (\Delta_{w''})=\dset (\corr ) (\Delta)$. Since $\R$ is a resolution, $\dset (\corr ) (\Delta) \in \convex (t(x)(a))$ and therefore $\R'$ is a fully probabilistic resolution.

\medskip

We now prove that for all $w'\in A^{*}$
and for all $y\in Y$, it holds that $\pprob_{\R}(y)(w')=\pprob_{\R' }(y,w')(w')$.
The proof goes by induction on $w'$.

If $w'=\epsilon$ then, $\pprob_{\R}(y)(\epsilon)=1=\pprob_{\R' }(y,\epsilon)(\epsilon)$.

If $w'=aw''$ and $r(y)(a)=\star$, then $r'(y,w')=\star$ and $\pprob_{\R' }(y,w')(w')=0=\pprob_{\R}(y)(w')$.

If $w'=aw''$ and $r(y)(a)=\Delta\neq\star$, then $r'(y, aw'')=\langle a, \Delta_{w''}\rangle$ and
\begin{eqnarray*}
 \pprob_{\R} (y) (w')
 &=& \sum_{y'\in \supp(\Delta)} \Delta(y') \cdot \pprob_{\mathcal{R}}(y')(w'')\\
 &\stackrel{(IH)}{=}& \sum_{y'\in \supp(\Delta)} \Delta(y') \cdot \pprob_{\mathcal{R'}}(y', w'')(w'') \\
 &=& \sum_{(y', w'')\in \supp(\Delta_{w''})} \Delta_{w''}(y', w'') \cdot \pprob_{\mathcal{R'}}(y', w'')(w'')\\
 &=& \pprob_{\mathcal{R'}}(y, aw'')(aw'')\\
 \end{eqnarray*}

Hence, $\pprob_{\R}(y)(w)=\pprob_{\R' }(y,w)(w)$, with $\corr (y)=\corr '(y,w)=x$.

\bigskip

We now prove that $\bbressup{x} (w) \geq \bbressupfp{x} (w)$.

Let $\R= (Y,\corr ,r)$ be a fully probabilistic resolution of $(X,t)$, $x\in X$, and $w\in A^*$. Let $\corr (y)=x$. We show that there exists a resolution $\R'=(Y',\corr ',r')$ of $(X,t)$ with a state $z$ such that $\corr '(z)=x$ and $\pprob_{\R}(y)(w) \leq \pprob_{\R' }(z)(w)$.
We define $\R'=(Y',\corr ',r')$ as follows:
\begin{itemize}
\item $Y'=Y+X$ is the (disjoint) union of $Y$ and $X$
\item $\corr '=\corr + \ide_X$ that is for all $y'\in Y'$
\[\begin{array}{rcl}
\corr '(y')& = &
\begin{cases}
              \corr (y')  & \text{ if } y' \in Y \text{;}  \\
		\ide_X(y') & \text{ if } y' \in X  \text{.}
           \end{cases}
\end{array}
\]
\item $r': Y' \to (\dset (Y') +1)^{A}$ is defined as:
\[\begin{array}{rcl}r'(y')(a)& = &
\begin{cases}
             \star & \text{if $t(\corr '(y'))(a)= \star$}\\
              \Delta  & \text{ if $y' \in Y$ and } r(y')=\langle a, \Delta\rangle      \\
              \Delta_ {\corr '(y'),a} & \text{ otherwise}\\
           \end{cases}
\end{array}
\]

where $\Delta_{x,a}$ are defined like in the base case of the proof of Proposition~\ref{prop:correspondence2} (namely, an arbitrary choice amongst the distributions in $t(x)(a)$).
\end{itemize}

\noindent
We prove that $\R'$ is a resolution. For elements  $y'\in X$ the conditions of Definition~\ref{def:res} are trivially satisfied (see the analogous construction used in the base case in the proof of Proposition~\ref{prop:correspondence2}).
Suppose $y' \in Y$.
\begin{enumerate}
\item By definition, $r'(y')(a)=\star$ iff $t(\corr '(y'))(a)= \star$.
\item If $r'(y')(a)\neq \star$, then we are either in the second or in the third case of the definition of $r'$. If we are in the second case, $r'(y')(a)=\Delta  $ with $ r(y')=\langle a, \Delta\rangle$. Since $\Delta\in \dset Y$ we have $\dset(\corr ') (\Delta)= \dset(\corr )(\Delta)$, and by the definition of fully probabilistic resolution, it holds that $\dset(\corr )(\Delta) \in \convex  (t(\corr (y'))(a))$. Therefore $\dset(\corr ')(\Delta) \in \convex  (t(\corr '(y'))(a))$.
If we are in the third case, we have $r'(y')(a)= \Delta_ {\corr '(y'),a}$, with $ \Delta_ {\corr '(y'),a} \in t(\corr '(y'))(a)$. By definition of $\corr '$, $\dset(\corr ')(\Delta_ {\corr '(y'),a})=\Delta_ {\corr '(y'),a}$. Therefore $\dset(\corr ')( \Delta_ {\corr '(y'),a} ) \in \convex  (t(\corr '(y'))(a))$.
\end{enumerate}

\noindent
We conclude by showing that for all $y \in Y$ and for all $w'$, $\pprob_{\R}(y,w') \leq \pprob_{\R' }(y,w')$. The proof goes by induction on $w'$.
The case $w'=\epsilon$ is trivial, since $\pprob_{\R}(y,\epsilon) =1 = \pprob_{\R' }(y,\epsilon)$.
For the inductive case, take $w'=aw''$.
Suppose $r(y)=\langle a, \Delta\rangle$. Then, by definition of $r'$, $r'(y)(a)= \Delta$, and
 \begin{align*}
 \pprob_{\R} (y)(aw'')
 &= \sum_{y'\in \supp(\Delta)} \Delta(y') \cdot \pprob_{\mathcal{R}}(y')(w'')\\
 &\leq \sum_{y'\in \supp(\Delta)} \Delta(y') \cdot \pprob_{\mathcal{R'}}(y')(w'') \tag{by IH}\\
 &= \pprob_{\mathcal{R'}}(y)(aw'')
 \end{align*}
Now suppose that  $r(y)\neq \langle a, \Delta\rangle$. Then $r(y)=\star$ and, by definition of  $\pprob_{\R}$, $\pprob_{\R}(y)(aw'')=0$, so there is nothing to prove.

Hence, $\pprob_{\R}(y)(w)\leq \pprob_{\R' }(y)(w)$, with $\corr (y)=\corr '(y)=x$.
\end{proof}

\begin{rem}
The correspondence  in Proposition~\ref{prop:correpsondencefp} does not hold when infima are considered, instead of suprema. Indeed define $\bbresinffp{x} (w)$ as expected, namely, by replacing $\bigsqcup$ with $\bigsqcap$ in $\bbressupfp{x} (w)$. Then for any state $x$ of an arbitrary  NPLTS it holds that $\bbresinffp{x} (w)=0$ for all $w\neq \varepsilon$.
To see this, observe that $\R'= (\{y\}, \corr ', r')$ with $\corr '(y)=x$ and $r'(y)=\star$ is always a fully probabilistic resolution, and that $\pprob_{\mathcal{R}'}(y)(w)=0$.

To avoid this problem, one typically modifies the definition of $\bbresinffp{\cdot}$ by restricting only to those fully probabilistic resolutions that can perform a certain trace (see e.g.~\cite{BDL14a,BDL14b}). Instead, with our notion of resolution based on RPLTSs (Definition~\ref{def:res}), this problem does not arise and the definition of $\bbresinf{\cdot}$ is totally analogous to the one of $\bbressup{\cdot}$.
\end{rem}

\mypar{Why may, must, may-must? Trace equivalences as testing equivalences}
The notion of resolution is at the basis not just of the definitions of trace equivalences for NPLTS investigated in the literature, but also of testing equivalences for nondeterministic and probabilistic processes~\cite{YL92,JHY94,DGHMZ07a,DGHM09}, extending the theory of testing equivalences for purely nondeterministic processes~\cite{DH84} (see Remark~\ref{rem:testlts}).
In testing equivalences for nondeterministic and probabilistic processes,
we say that $x,y$ are may testing equivalent if, for every test, they have the same greatest probabilities of passing the test, with respect to any resolution $\R$ of the system resulting from the interaction between the test and the NPLTS\@. Analogously, $x,y$ are must testing equivalent if the smallest probabilities coincide, and the may-must testing equivalence requires both the greatest and the smallest probabilities to coincide.

Now, take tests to be finite traces, and the probability of passing a given test in a resolution as the probability of performing the trace in the resolution. Then it becomes clear, by the correspondence between the local and the global view proven in Theorem~\ref{thm:correspondence}, that each of our three trace equivalences indeed coincides with the corresponding testing equivalence, when tests are finite traces.

\section{Conclusion and future work}\label{sec:conc}

We developed an algebra-and-coalgebra-based trace theory for systems with nondeterminism and probability, that covers intricate trace semantics from the literature. The abstract approach sheds light on all choices and leaves no space for ad-hoc solutions.

Although the combination of nondeterminism and probability has been posing challenges to abstract approaches, this new algebraic theory of traces for NPLTS shows that it can be dealt with in a smooth and uniform way.

In~\cite{Jacobs08}, Jacobs leaves open the question of whether the semantics that he obtains by exploiting the Kleisli traces~\cite{DBLP:journals/lmcs/HasuoJS07} for the monad $C$ coincides with some scheduler based semantics. We would like to answer that such semantics corresponds to the may-must trace equivalence by our results and the correspondence between Kleisli traces and generalised determinisation studied in~\cite{JacobsSS15}. However, we cannot directly derive this: for technical reasons (namely, the absence of a bottom element in the Kleisli category of $C$), Jacobs has to hack the general framework in~\cite{DBLP:journals/lmcs/HasuoJS07} and therefore the correspondence does not follow from known results. Providing an answer to this interesting question is left as future work.

\bibliographystyle{alphaurl}

 \bibliography{biblio}

\end{document}